\documentclass[a4paper,11pt,reqno]{amsart}
\usepackage{amsmath}
\usepackage{amsfonts}
\usepackage{eucal}
\usepackage{amsthm}
\usepackage{color}

\usepackage{graphicx}
\numberwithin{equation}{section}

\newcommand{\ZM}{{\mathbb Z}}

\def\re{\mathbb{R}}
\def\co{\mathbb{C}}
\def\ze{\mathbb{Z}}

\def\S{\mathbb{S}}

\newcommand{\Slash}[1]{{\ooalign{\hfil#1\hfil\crcr\raise.167ex\hbox{/}}}}

\newtheorem{thm}{Theorem}
\newtheorem{lem}[thm]{Lemma}
\newtheorem{prop}[thm]{Proposition}

\theoremstyle{definition}

\newtheorem{ass}{Assumption}

\theoremstyle{remark}
\newtheorem{rem}{Remark\!}

\begin{document}

\title[The Noncommutative Index Theorem and the Periodic Table]{The Noncommutative Index Theorem 
and the Periodic Table for Disordered Topological Insulators and Superconductors}
\author[H. Katsura]{Hosho Katsura}
\author[T. Koma]{Tohru Koma}
\address[Katsura]{Department of Physics, Graduate School of Science, University of Tokyo, JAPAN}
\email{katsura@phys.s.u-tokyo.ac.jp}
\address[Koma]{Department of Physics, Gakushuin University, 
Mejiro, Toshima-ku, Tokyo 171-8588, JAPAN}
\email{tohru.koma@gakushuin.ac.jp}
\date{December 28, 2017, Revised Version}

\begin{abstract}
We study a wide class of topological free-fermion systems on a hypercubic lattice 
in spatial dimensions $d\ge 1$. 
When the Fermi level lies in a spectral gap or a mobility gap, the topological properties, e.g., 
the integral quantization of the topological invariant, 
are protected by certain symmetries of the Hamiltonian against disorder. 
This generic feature is characterized by a generalized index theorem 
which is a noncommutative analogue of the Atiyah-Singer index theorem. 
The noncommutative index defined in terms of a pair of projections gives a precise formula 
for the topological invariant in each symmetry class in any dimension ($d \ge 1$). 
Under the assumption on the nonvanishing spectral or mobility gap, 
we prove that the index formula reproduces Bott periodicity and all of the possible values of topological invariants 
in the classification table of topological insulators and superconductors. 
We also prove that the indices are robust against perturbations that do not break 
the symmetry of the unperturbed Hamiltonian. 
\end{abstract}

\maketitle 

\tableofcontents
\section{Introduction}
\label{sec:intro}

Topological insulators and superconductors are a wide class of free-fermion systems 
which show nontrivial topological quantities  \cite{HasanKane,QiZhang} 
in the sense that they are robust against any perturbations. 
Historically, the first example of such a system was the integer quantum Hall effect \cite{KawajiWakabayashi,KDP}, 
where the Hall conductance is quantized in integer multiples of the universal conductance. 
The integer is given by the Chern number which is the topological invariant of the system \cite{TKNN,Kohmoto}. 
Another example, called $\ze_2$ topological insulators, 
was later found by Kane and Mele \cite{KaneMele}. The key difference from the quantum Hall effect is 
that the topological properties of such a system is meaningful only when time-reversal symmetry is respected. 
This idea has been generalized to other types of topological insulators, 
in which their topological properties are robust against any perturbations 
that preserve certain discrete symmetries of the unperturbed Hamiltonian, 
provided the Fermi level is in a spectral gap or a mobility gap. 
The complete classification of all topological insulators and superconductors was worked out 
using the idea of random matrix theory \cite{RSFL} or K-theoretical approach \cite{KitaevPT}. 
In either case, the result obtained can be summarized in the same form of the ``periodic table" (see Table \ref{Ptable} below). 

In translationally invariant systems, the bulk topological invariant is typically written 
as an integral over the Brillouin zone, and is robust with respect to the variation of the parameters of the Hamiltonian 
unless the band gap closes. This picture is, however, not applicable to systems with disorder, 
where the Bloch momentum is no longer a good quantum number. 
Nevertheless, certain topological properties are expected to be robust against disorder, 
as is the case for the quantization of the Hall conductance in quantum Hall systems. 
So far, there have been two different approaches to tackle topological insulators with disorder. 

The first one was initiated by Niu, Thouless and Wu \cite{NTW}, 
which works for generic quantum Hall systems with disorder and/or interactions. 
In their approach, twisted boundary conditions at the boundaries are imposed 
and the existence of a spectral gap above the ground-state sector is assumed. 
Under these assumptions, they showed that the Hall conductance averaged over 
the boundary conditions is integrally or fractionally quantized, irrespective of disorder or interactions. 
Clearly, their result does not necessarily imply the quantization of the Hall conductance for a set of fixed twisted phases. 
But recently, it was proved that the Hall conductance exhibits the desired quantization irrespective of 
the values of the twisted phases for interacting lattice fermions \cite{HastingsMichalakis,Koma5}. 
However, its extension to interacting continuum systems still remains an open problem. 

The second approach is based on methods of noncommutative geometry \cite{Connes,CM,Higson1,Higson2}. 
The advantage is that it is applicable to the situation that the Fermi energy lies in a localization regime. 
Using the approach, the integral quantization of the Hall conductance 
is proved to be robust against disorder for noninteracting systems \cite{BVS,ASS,AG,ElSch,Koma2}. 
In particular, the emergence of the Hall conductance plateau of bulk order can be proved 
for varying the filling factor of the electrons rather than the Fermi level \cite{Koma2}.

The application of noncommutative geometry has been rather limited to the quantum Hall systems 
in two dimensions so far \cite{BVS,ASS,AG}. Recently, however, the extension to higher even dimensions has been  
done by \cite{PLB}. Besides,  it has been realized that this approach is applicable to 
other types of topological insulators. Among them, there are two classes of systems 
in which the topological invariant is proved to be robust against perturbations.

One class is topological insulators with chiral symmetry in odd dimensions. 
It was found in \cite{MHSP,PSB,PSBbook} that the methods of noncommutative geometry are applicable to 
proving the robustness of the integral quantization of the winding number related to the spin 
in odd dimensional Pauli-Dirac theory \cite{So,IM,Kitaev,SRFL,RSFL} with chiral symmetry. 
The winding number can be defined in terms of the sign function of the Hamiltonian. Following Connes' program \cite{Connes}, 
Prodan and Schulz-Baldes proved that the winding number is quantized to a nontrivial integer \cite{PSB}. 

The other class is topological insulators with time-reversal symmetry in two dimensions \cite{KaneMele,SB}. 
In the previous paper \cite{KatsuraKoma}, the authors proposed an alternative approach which is based on 
the method of a pair of projections introduced by Avron, Seiler and Simon \cite{ASS,ASS2}. 
Relying on their method, the noncommutative formula of the $\ze_2$ index for this class was obtained and
the robustness of the index against any time-reversal symmetric perturbation including disorder was proved. 

In this paper, we extend our previous approach to generic topological insulators and superconductors. 
As a result, we prove that the noncommutative index theorem gives a precise formula for the topological invariant for 
each element in the periodic table \cite{KitaevPT,SRFL,RSFL}. In other words, 
all of the indices in Table~\ref{Ptable} below can be explained by our noncommutative index theorem. 
The advantage of our approach is that it is not necessary for disordered systems 
to require the ergodicity of the probability measure \cite{EGS,Koma2}. 
Namely, the integral quantization of the topological invariant is proved 
to be robust against any symmetry preserving perturbation as long as the Fermi level lies 
in a spectral gap of the Hamiltonian or in a mobility gap. 

The present paper is organized as follows: In the next section, we define a generic tight-binding model 
on a hypercubic lattice in dimensions $d \ge 1$, and discuss the possible symmetries of the Hamiltonian. 
In Sec.~\ref{Sec:MainResults}, we present our main results in the form of the index theorem, 
which allows us to relate the analytic index to the topological invariant. 
The proofs of the integer-valued index theorems in even and odd dimensions are given in 
Sec.~\ref{sec:NosymmEvenD} and Sec.~\ref{sec:ChiralSymmOddD}, respectively. 
The detailed nature of the indices for all the classes of models in the periodic table are proved in Sec.~\ref{PeriodicTable}.     
In Sec.~\ref{HomoArg}, we show the robustness of the indices against perturbations using a homotopy argument. 
Our convention for the gamma matrices and their useful properties are summarized 
in Appendix~\ref{AppGamma}. Appendices~\ref{traceclassA}-\ref{proofTheorem:ChiralIndtheta}
are devoted to technical estimates. In Appendix~\ref{sec:LR}, we show 
that the chiral topological invariants can be interpreted as a linear-response coefficient in one dimensions. 
Appendices~\ref{EvenALO} and \ref{RelationIndices} are devoted to the proofs of useful propositions. 
In particular, we show the relation between the indices of classes DIII and AII in three dimensions 
in Appendix~\ref{RelationIndices}. 

\section{Tight-Binding Models on $\ze^d$}
\label{sec:Model}

We consider a tight-binding model of fermions on the $d$-dimensional hypercubic lattice $\ze^d$ with $d\ge 1$. 
The Hamiltonian $H$ reads
\begin{equation}
(H\varphi)_\alpha(x):=\sum_{y}\sum_{\beta}t_{x,y}^{\alpha,\beta}\varphi_\beta(y)\quad \mbox{for \ } x\in\ze^d,
\label{H}
\end{equation}
where the subscript $\beta$ of the wavefunction $\varphi_\beta\in \ell^2(\ze^d,\co^M)$ 
is the internal degree of freedom with the dimension $M$ such as spin or orbital; 
the hopping amplitudes $t_{x,y}^{\alpha,\beta}$ are complex numbers which satisfy 
the Hermitian conditions, 
$$
t_{y,x}^{\beta,\alpha}=\left(t_{x,y}^{\alpha,\beta}\right)^\ast.
$$
We note that the above Hamiltonian includes a tight-binding model on a more general graph, 
because it can be embedded into the model on $\ze^d$ with suitably chosen hopping amplitudes.
We assume that the hopping amplitudes are of finite range, 
and that all of the strengths are uniformly bounded as 
$$
\left|t_{y,x}^{\beta,\alpha}\right|\le t_0 
$$
with some positive constant $t_0$. Throughout the present paper, we require the following assumption: 

\begin{ass}
\label{Assumption} 
The Fermi level $E_{\rm F}$ lies in the spectral gap of the Hamiltonian $H$ 
or in the localization regime. More precisely, we assume that the resolvent $(E_{\rm F}-H)^{-1}$ 
exponentially decays with distance as    
\begin{equation}
\label{decayresolvent}
\sup_{\varepsilon>0}\left\Vert\chi_{\{x\}}(E_{\rm F}+i\varepsilon-H)^{-1}\chi_{\{y\}}\right\Vert
\le \mathcal{C}_0\exp[-|x-y|/\xi_0],
\end{equation}
where $\chi_{\{x\}}$ is the characteristic function of the site $x$, i.e., 
\begin{equation}
\label{chix}
\chi_{\{x\}}(x'):=\begin{cases} 1, & \text{$x=x'$};\\
0, & \text{$x\ne x'$},
\end{cases}
\end{equation}
and the positive constants, $\mathcal{C}_0$ and $\xi_0$, depend only on the parameters of the tight-binding model. 
\end{ass}

\begin{rem}
The bound (\ref{decayresolvent}) holds for random Hamiltonians with probability one as in Theorems~2.4 and 2.5 in \cite{CH}.
Instead of (\ref{decayresolvent}), we can assume the Aizenman-Molchanov bound \cite{AM}. 
In this case, the corresponding statements about the index theorem are proved along the lines of \cite{Koma2}. 
The bound (\ref{decayresolvent}) yields that the matrix elements of  
the Fermi projection $P_{\rm F}$ also decays exponentially with distance as in (\ref{decayPF}) below. 
\end{rem}

We can treat random models without making any assumption on
the distribution because we assume only the exponential decay of the resolvent for a fixed configuration.
For the same reason, our method is applicable to deterministic models,
e.g., tight-binding models on Sierpinski carpet \cite{Sierpinski} or Penrose tiling \cite{Penrose}.
Actually, by removing appropriate bonds from the bonds of the square lattice,
the Sierpinski carpet can be constructed by starting from the square
lattice. As for the Penrose tiling, this can be embedded into
the two-dimensional plane. Therefore, by using the Euclidean distance,
we can define the Dirac operator in Sec.~\ref{Sec:MainResults} below for both cases. 
As a consequence, we can apply our method to the corresponding models. 

We can also treat superconductors which are described by the Bogoliubov-de Gennes Hamiltonian 
which is given by a quadratic form of creation and annihilation fermion operators. 
Then, the matrix elements of the quadratic form define the corresponding Hamiltonian (\ref{H}). 
Namely, it is reduced to that for a single-body tight-binding Hamiltonian.  
The ``Fermi energy" is set to $E_{\rm F}=0$ by definition. 

\subsection{Time-Reversal Transformation}

As usual, we define the complex conjugate for the present wavefunctions, $\varphi\in\ell^2(\ze^d,\co^M)$, by 
$$
(\overline{\varphi})_\alpha(x)=\overline{\varphi_\alpha(x)}\quad \mbox{for \ } x\in\ze^d,
$$
where $\alpha$ is an index of the internal degree of freedom. 
We introduce a time-reversal transformation $\Theta$ for the wavefunctions as 
\begin{equation}
\label{TRS}
\varphi^{\Theta}:=\Theta\varphi=U^{\Theta}\overline{\varphi}\quad \mbox{for \ }\varphi\in\ell^2(\ze^d,\co^M).
\end{equation}
Here, for simplicity, we assume that $U^\Theta$ is a unitary operator which can be written in 
a direct sum as $U^\Theta=\bigoplus_i U_i^\Theta$ with a period on the lattice $\ze^d$, 
where the unitary operator $U_i^\Theta$ acts on a local state on the $i$-th unit cell with a compact support. 
Since all of the states on some cells can be identified with an internal degree of freedom at a single site 
of a lattice, the unitary operator $U^\Theta$ can be taken to be independent of the lattice site. 

If a time-reversal transformation $\Theta$ satisfies
$$
\Theta^2\varphi=-\varphi\quad \mbox{for any \ } \varphi\in\ell^2(\ze^d,\co^M), 
$$
then we say that $\Theta$ is an odd time-reversal transformation. In the case where
$$
\Theta^2\varphi=\varphi\quad \mbox{for any \ } \varphi\in\ell^2(\ze^d,\co^M), 
$$
we say that $\Theta$ is an even time-reversal transformation.

Let $\mathcal{A}$ be an operator on the Hilbert space $\ell^2(\ze^d,\co^M)$. We say that 
the operator $\mathcal{A}$ is time-reversal symmetric with respect to $\Theta$ 
if the following condition holds: 
$$
\Theta(\mathcal{A}\varphi)=\mathcal{A}\varphi^{\Theta}\quad \mbox{for any \ }\varphi\in\ell^2(\ze^d,\co^M).
$$
If the present Hamiltonian $H$ satisfies this condition, we say that the Hamiltonian $H$ is 
time-reversal symmetric with respect to $\Theta$.  
In addition, if $\Theta$ is an odd time-reversal transformation, we say that the Hamiltonian $H$ is 
odd time-reversal symmetric. 

\subsection{Particle-Hole Transformation}
\label{sec:ChiralModel}

Next, we introduce a particle-hole transformation $\Xi$ for the wavefunctions as 
\begin{equation}
\varphi^{\Xi}:=\Xi\varphi=U^{\Xi}\overline{\varphi}\quad \mbox{for \ }\varphi\in\ell^2(\ze^d,\co^M),
\end{equation}
where $U^\Xi$ is a unitary operator, 
similarly to the case of the time-reversal transformation $\Theta$. 
In the same way as in the time-reversal transformation, one can define odd and even particle-hole transformations. 
The only difference is that the following condition for the Hamiltonian $H$ is required: 
\begin{equation}
\label{PHconditionH}
\Xi(H\varphi)=-H\varphi^{\Xi}\quad \mbox{for any \ } \varphi\in\ell^2(\ze^d,\co^M), 
\end{equation}
in which case we say that the Hamiltonian $H$ is particle-hole symmetric with respect to $\Xi$. 

The condition (\ref{PHconditionH}) implies that, if a positive energy, $E>0$, is in the spectrum of the Hamiltonian $H$, 
then the corresponding negative energy, $-E$, is also in the spectrum of $H$.
Namely, the upper energy bands are mapped to the lower energy bands by the particle-hole transformation $\Xi$. 
This is, in fact, a built-in symmetry of the Bogoliubov-de Gennes Hamiltonians.
For a class of Hamiltonians with this symmetry, we set $E_{\rm F}=0$
and deal with the following two situations: 
(i) There is a spectral gap between the upper and lower bands in 
the spectrum of the particle-hole symmetric Hamiltonian $H$. 
(ii) Consider the situation that the Fermi level $E_{\rm F}=0$ lies in 
a localization regime. If $E_{\rm F}=0$ is an eigenvalue of the Hamiltonian $H$, 
we cannot handle the resolvent $(E_{\rm F}-H)^{-1}$. 
Fortunately, it is well known that the corresponding event does not occur with probability one 
in many of random systems. In particular, a matrix-valued random potential \cite{Bourgain,KMM,GM,ChapStolz} 
meets the present conditions. 
In fact, for such random potentials, the decay bound (\ref{decayresolvent}) in Assumption~\ref{Assumption} 
holds with probability one \cite{CH}, provided that both an initial decay estimate for the resolvent and a Wegner estimate 
for the density of states hold. This decay bound implies that $E_{\rm F}$ is not an eigenvalue of $H$ with probability one. 
Thus, we assume that $E_{\rm F}=0$ is not an eigenvalue of the particle-hole symmetric Hamiltonian $H$ with 
probability one with respect to configurations of the random quantities of the system. 
In such a system, the matrix elements of the Fermi projection $P_{\rm F}$ also decays exponentially with distance 
as in (\ref{decayPF}) below. 

\subsection{Chiral Transformation}

We say that the Hamiltonian $H$ is chiral symmetric with respect to a chiral operator $S$ 
if the following conditions hold: 
\begin{itemize}
\item $S$ is self-adjoint, i.e., $S^\ast=S$ 
\item $S$ is unitary, i.e., $S^2=1$ 
\item The Hamiltonian $H$ is transformed by $S$ as   
\begin{equation}
\label{chiralcondition}
SHS=-H. 
\end{equation}
\end{itemize}
Because of $S^2=1$, there is no property about even-oddness in this case.  
The condition (\ref{chiralcondition}) implies that the 
positive-energy bands are mapped to the corresponding negative-energy ones, 
as in the case of the particle-hole symmetric Hamiltonian.  
Therefore, we set $E_{\rm F}=0$, and deal with the situation that 
the matrix elements of the Fermi projection $P_{\rm F}$ 
decays exponentially with distance as in (\ref{decayPF}) below. 

\section{Main Results}
\label{Sec:MainResults}

In this paper, we prove that the noncommutative index theorem gives the desired index of 
the topological invariant for 
each element in the periodic table (Table~\ref{Ptable}).   
The analysis of the nature of the indices reproduces the results which were claimed in previous works \cite{KitaevPT,SRFL,RSFL}.

In Table~\ref{Ptable}, possible topological insulators and superconductors are sorted out according 
to their symmetries and spatial dimensions $d$. 
The first column refers to the Cartan-Altland-Zirnbauer (CAZ) 
classification scheme \cite{AZ} for symmetry of the Hamiltonian $H$. 
The columns, TRS, PHS and CHS, denote the character of the Hamiltonian $H$ in relation to time-reversal, particle-hole and 
chiral symmetries, respectively. If the Hamiltonian $H$ does not have 
the corresponding symmetry, then we write $0$. We write $\pm 1$ for even- and oddness for symmetry, respectively. 
Since a chiral operator $S$ always satisfies $S^2=1$, we write $1$ for chiral symmetry classes. 
The symbols $\ze$ and $\ze_2$ indicate whether the model belonging to an entry takes an integer-valued 
or a $\ze_2$-valued  index, respectively. 
The symbol $2\ze$ denotes that the integer-valued index is always even. 
The empty entries denote the case where the corresponding model has no index.
Due to the periodicity with the period $8$ with respect to the spatial dimension $d$, we omit the cases for 
nine or higher dimensions. This fact which is often called ``Bott periodicity" will be proved in Sec.~\ref{PeriodicTable}. 
Roughly speaking, the periodicity $8$ is a consequence of the periodicity of the Dirac-Clifford algebra. 

\begin{table}
\begin{center}
\begin{tabular}{cccc||c c c c c c c c}
\hline\hline
\multicolumn{4}{c||}{Symmetry} & \multicolumn{8}{|c}{Spatial dimension $d$}\\
\hline
\multicolumn{1}{c|}{$\!$CAZ$\!$} & $\!$TRS$\!$ & $\!$PHS$\!$ & $\!$CHS$\!$ &  $1$ & $2$ & $3$ & $4$ & $5$ &  $6$ & $7$ & $8$ \\
\hline
\multicolumn{1}{c|}{A} &  $0$ &$0$&$0$ &  & $\ZM$ & & $\ZM$ & & $\ZM$  & & $\ZM$\\
\multicolumn{1}{c|}{AIII}  & $0$&$0$&$ 1$ & $\ZM$ & & $\ZM$ & & $\ZM$ & & $\ZM$  & \\
\hline
\multicolumn{1}{c|}{AI} & $+1$&$0$&$0$ &  & & & $2\,\ZM$ & & $\ZM_2$ & $\ZM_2$ & $\ZM$\\
\multicolumn{1}{c|}{BDI} & $+1$&$+1$&$1$ &  $\ZM$ & & & & $2\,\ZM$ & & $\ZM_2$ & $\ZM_2$ \\
\multicolumn{1}{c|}{D} & $0$ &$+1$&$0$ &  $\ZM_2$ & $\ZM$ & & & & $2\,\ZM$ & & $\ZM_2$ \\
\multicolumn{1}{c|}{DIII} & $-1$&$+1$&$1$ &  $\ZM_2$ & $\ZM_2$ & $\ZM$ & & & & $2\,\ZM$ & \\
\multicolumn{1}{c|}{AII} & $-1$&$0$&$0$ &  & $\ZM_2$ & $\ZM_2$ & $\ZM$ & & & & $2\,\ZM$ \\
\multicolumn{1}{c|}{CII} & $-1$&$-1$&$1$ &  $2\,\ZM$ & & $\ZM_2$ & $\ZM_2$ & $\ZM$ & & & \\
\multicolumn{1}{c|}{C} & $0$ &$-1$&$0$ & & $2\,\ZM$ & & $\ZM_2$ & $\ZM_2$ & $\ZM$  & & \\
\multicolumn{1}{c|}{CI} & $+1$&$-1$&$1$ &  & & $2\,\ZM$ & & $\ZM_2$ & $\ZM_2$ & $\ZM$ & \\
\hline\hline
\end{tabular}
\bigskip
\caption{\sl The periodic table of topological insulators and superconductors, 
listing the possible values of the indices as the strong topological invariants 
for the CAZ classes in dependence of the spatial dimension $d$.}
\label{Ptable}
\end{center}
\end{table}

\subsection{Index Theorems in Even Dimensions}

Let us consider first the Hamiltonian (\ref{H}) on $\ze^d$ 
with even dimensions, $d=2n$, ($n=1,2,\ldots$). 
We write $P_{\rm F}$ for the projection onto the states 
whose energy is less than the Fermi energy $E_{\rm F}$, and refer to it as the Fermi projection. 
By using a contour integral in the complex plane, it can be written as 
\begin{equation}
\label{PFcontour}
P_{\rm F}=\frac{1}{2\pi i}\oint dz \frac{1}{z-H}.
\end{equation}

Let $\gamma^{(1)},\gamma^{(2)},\ldots,\gamma^{(2n)},\gamma^{(2n+1)}$ be 
the gamma matrices acting on the $2^n$-dimensional Hilbert space $\co^{2^n}$ 
which represents an 
auxiliary degrees of freedom.  
The gamma matrices obey the anticommutation relations, 
\begin{equation}
\label{gammaACR} 
\gamma^{(i)}\gamma^{(j)}+\gamma^{(j)}\gamma^{(i)}=2\delta_{ij},
\end{equation}
for $i,j\in\{1,2,\ldots,2n+1\}$ with the Kronecker delta $\delta_{ij}$. 
The concrete expression for the gamma matrices and their basic properties are given in 
Appendix~\ref{AppGamma}. 
We write $\gamma=(\gamma^{(1)},\gamma^{(2)},\ldots,\gamma^{(2n)})$, i.e., $\gamma$ is a $2n$-component vector 
whose $i$-th component is given by the gamma matrix $\gamma^{(i)}$. 
We denote the Euclidean distance on $\re^d$ by  
$$
|y|:=\left[\sum_{j=1}^{2n}(y^{(j)})^2\right]^{1/2}\quad \mbox{for \ }
y=(y^{(1)},y^{(2)},\ldots,y^{(2n)})\in\re^d.
$$
We introduce a Dirac operator as \cite{Connes,CM,PLB}
\begin{equation}
\label{defDiracEven}
D_a(x):=\frac{1}{|x-a|}(x-a)\cdot \gamma\quad \mbox{for \ }x\in\ze^d,
\end{equation}
where we have used the shorthand notation, 
$$
(x-a)\cdot \gamma=\sum_{i=1}^{2n}(x^{(i)}-a^{(i)})\gamma^{(i)},
$$
and the $2n$-component vector $a=(a^{(1)},a^{(2)},\ldots,a^{(2n)})$ satisfies $a\in\re^d\backslash\ze^d$. 
The Dirac operator $D_a$ acts on the Hilbert space $\ell^2(\ze^d,\co^M)\otimes\co^{2^n}$,  
where the position operator $x$ and the gamma matrix $\gamma^{(i)}$ is identified with 
$x\otimes 1$ and $1\otimes \gamma^{(i)}$, 
respectively. 
The Hamiltonian $H$ of (\ref{H}) is identified 
with $H\otimes 1$ on $\ell^2(\ze^d,\co^M)\otimes\co^{2^n}$. 
The Dirac operator of (\ref{defDiracEven}) is a noncommutative analogue \cite{Connes,CM} 
of the partial differential operators in the Atiyah-Singer index theorem.

In order to describe our noncommutative index theorems, we introduce the current operator $J_a^{(j)}$ 
in the $j$-th direction with the kink at $a^{(j)}$ by \cite{Koma5}
\begin{equation}
\label{currentJaj}
J_a^{(j)}:=i[H,\vartheta_a^{(j)}],
\end{equation}
where 
\begin{equation}
\label{varthetaaj}
\vartheta_a^{(j)}(x):=\vartheta(x^{(j)}-a^{(j)})\ \ \mbox{for \ } x\in\ze^d 
\end{equation}
with the step function, 
\[
\vartheta(b):=\begin{cases} 1, & \text{$b\ge 0$};\\
0, & \text{$b<0$},
\end{cases}
\]
for $b\in\re$. 

\begin{thm}
\label{NonCommuIntThmEven}
In even dimensions, the following relation between the integer-valued index and the topological invariant 
(generalized Chern number) is valid:   
\begin{multline*}
\frac{1}{2}\Bigl\{{\rm dim}\;{\rm ker}\;[\gamma^{(2n+1)}(P_{\rm F}-D_aP_{\rm F}D_a)-1]\\
-{\rm dim}\;{\rm ker}\;[\gamma^{(2n+1)}(P_{\rm F}-D_aP_{\rm F}D_a)+1]\Bigr\}\\
=\frac{(-1)^{n-1}}{n!(2\pi i)^n}\sum_\sigma(-1)^\sigma 
\oint dz_1\oint dz_2\cdots \oint dz_{2n}{\rm Tr}\; 
P_{\rm F}\frac{1}{z_1-H}J_a^{(\sigma_1)}\frac{1}{z_1-H}\\
\times \frac{1}{z_2-H}J_a^{(\sigma_2)}\frac{1}{z_2-H}\cdots\frac{1}{z_{2n}-H}J_a^{(\sigma_{2n})}\frac{1}{z_{2n}-H},
\end{multline*}
where ${\rm dim}\;{\rm ker}\; O$ stands for the dimension of the kernel of an operator $O$, ${\rm Tr}\;$ denotes 
the trace of operator on the Hilbert space $\ell^2(\ze^d,\co^M)$, and $\sigma$ is a permutation given by 
$$
\sigma=\left(\begin{matrix}1,2,\ldots,2n\cr \sigma_1,\sigma_2,\ldots,\sigma_{2n}\cr \end{matrix}\right).
$$
with the signature $(-1)^\sigma$. Further, the integer-valued index is continuous with respect to 
the norm of any perturbation if the perturbed Hamiltonian has the same symmetry as that of the unperturbed Hamiltonian. 
In other words, as long as the Fermi level $E_{\rm F}$ lies in a spectral gap or a mobility gap, 
the integer-valued index is robust against any perturbation which preserves the same symmetry 
as that of the unperturbed Hamiltonian.   
\end{thm}

The proof of the index formula for the most typical class A, which has no symmetry, is given in Sec.~\ref{sec:NosymmEvenD}, 
and for the rest of the classes that have various symmetries, the proofs are given in Sec.~\ref{PTableEvenD}.   
The continuity of their indices is proved by combining the argument of a certain supersymmetric structure \cite{ASS2} 
at the end of Sec.~\ref{NoSymEvenDInd} with the homotopy argument in Sec~\ref{HomoArg}.   

\begin{rem}
(i) The noncommutative index Theorem~\ref{NonCommuIntThmEven} gives the desired $\ze$ or $2\ze$ indices 
for the corresponding elements in Table~\ref{Ptable}. In particular, 
the classes with $2\ze$ index \cite{RSFL,TeoKane,ShiozakiSato,GSB} have a structure similar to 
the Kramers doublet in the spectrum of 
the operator $\gamma^{(2n+1)}(P_{\rm F}-D_aP_{\rm F}D_a)$ as we show in Sec.~\ref{PTableEvenD}. 
\smallskip

\noindent
(ii) In two or higher even dimensions, Prodan, Leung and Bellissard \cite{PLB} discussed 
the conditions under which both the quantization and homotopy invariance of the noncommutative Chern number hold 
against perturbations. 
Their underlying assumption which they call, strong disorder, is slightly stronger than ours. In addition, they required that 
the probability measure for disorder is ergodic with respect to spatial translations.   
On the other hand, our approach is based on the method of a pair of projections which was introduced 
by Avron, Seiler and Simon \cite{ASS,ASS2}. 
The advantage of our approach is that \cite{EGS,Koma2} 
it is not necessary for disordered systems to require the ergodicity of the probability measure. 
Namely, the integral quantization of the topological invariant is proved to be robust against any perturbation 
as long as the Fermi level lies in a spectral gap of the Hamiltonian or in a localization regime. 
\end{rem}

\begin{thm}
\label{thm:z2even}
In even spatial dimensions, 
for the entries specified by the symbol $\ze_2$ in Table~\ref{Ptable}, 
the integer-valued index given in Theorem~\ref{NonCommuIntThmEven} vanishes identically. 
However, the following $\ze_2$ index may be nonvanishing: 
$$
{\rm Ind}_2^{(2n)}(D_a,P_{\rm F})
:=\frac{1}{2}{\rm dim}\;{\rm ker}\;[\gamma^{(2n+1)}(P_{\rm F}-D_aP_{\rm F}D_a)-1]\ \ \mbox{\rm modulo}\ 2.
$$
More precisely, this $\ze_2$ index is continuous with respect to 
the norm of any perturbation if the perturbed Hamiltonian has the same symmetry as that of the unperturbed Hamiltonian. 
In other words, as long as the Fermi level lies in a spectral gap of the Hamiltonian or in a localization regime, 
the $\ze_2$ index is robust against any perturbation which preserves the same symmetry as that of the unperturbed Hamiltonian.  
\end{thm}

The proof is given in Sec.~\ref{PTableEvenD}. In particular, the continuity of the $\ze_2$ indices is 
proved at the end of the case of AII class in two dimensions by relying on the homotopy argument in Sec~\ref{HomoArg}.
A numerical demonstration of the above formula was presented in \cite{AKK}.

\begin{rem}
(i) Since the $\ze_2$ index is given by the dimension of the kernel of the operator, 
one might think that the $\ze_2$ index may be written in terms of an integral of some connection, 
similarly to the Chern number. In fact, 
the $\ze_2$ index which was defined by Kane and Mele 
can be written in terms of an integral of the same connection as that of the Chern number 
over one-half of the Brillouin zone 
for translationally invariant systems \cite{FuKane,MooreBalents,EssinMoore}. 
(See also related articles \cite{FuKane2,Roy,FukuiHatsugai,LeeRyu}.) 
In \cite{KatsuraKoma}, the authors showed that, for Kane-Mele model without disorder, the $\ze_2$ index 
in the form of an integral coincides with the dimension of the kernel of an operator. 
This index of the operator coincides with that in Theorem~\ref{thm:z2even}, too. 
\smallskip

\noindent
(ii)  Schulz-Baldes \cite{SB,GSB} defined 
$\ze_2$-indices for general odd symmetric Fredholm operators 
and treated the $\ze_2$ index for the Kane-Mele model as an example. 
(See also \cite{FukuiFujiwara,DNSB} for related articles.)
But his approach is different from ours. In fact, he defined the $\ze_2$ index 
by the parity of ${\rm dim}\; {\rm ker}\; \mathfrak{T}$ for the Fredholm operator $\mathfrak{T}$. 
As we will show in Sec.~\ref{NoSymEvenDInd} below, 
the dimension of the kernel of $\mathfrak{T}$ coincides with ours. 
On the one hand, the operator $\mathfrak{T}$ is noncompact and contains the essential spectrum.
On the other hand, our operator, say $\mathfrak{A}$, has only the discrete spectrum 
with finite multiplicity except for zero, because $\mathfrak{A}$ is compact 
which follows from the fact that $\mathfrak{A}^m$ with some positive integer $m$ is trace class. 
Therefore, a homotopy argument for $\mathfrak{A}$ is much easier to handle than that for $\mathfrak{T}$. 
It should be noted that 
the homotopy argument is indispensable for defining the $\ze_2$ index for disordered systems. 
\smallskip

\noindent
(iii)  Hastings and Loring \cite{HL1,HL2,LH,Loring} proposed an alternative noncommutative approach to 
the $\ze_2$ index of topological insulators.  

\end{rem}

\subsection{Index Theorems in Odd Dimensions}

Let us consider the Hamiltonian $H$ of (\ref{H}) on $\ze^d$ 
with odd dimensions, $d=2n+1$, ($n=0,1,2,\ldots$). 
In this case, we set 
\begin{equation}
\label{gammaOddDim}
\gamma=(\gamma^{(1)},\gamma^{(2)},\ldots,\gamma^{(2n)},\gamma^{(2n+1)}).
\end{equation}
The corresponding Dirac operator is given by  
\begin{equation}
\label{defDiracOdd}
D_a(x):=\frac{1}{|x-a|}(x-a)\cdot \gamma
\end{equation}
for $x=(x^{(1)},x^{(2)},\ldots,x^{(2n+1)})\in\ze^{2n+1}$ 
and $a=(a^{(1)},a^{(2)},\ldots,a^{(2n+1)})\in\re^d\backslash\ze^d$. 
Here, $|y|$ denotes the Euclidean distance given by 
$$
|y|=\left[\sum_{j=1}^{2n+1}(y^{(j)})^2\right]^{1/2} \ \mbox{for \ } y=(y^{(1)},y^{(2)},\ldots,y^{(2n+1)})\in\re^{2n+1}.
$$

Consider first the case that the Hamiltonian $H$ has no chiral symmetry. 

\begin{thm}
\label{thm:Z2oddNoS}
Assume that the Hamiltonian $H$ has no chiral symmetry in odd spatial dimensions. 
Then, all of the nontrivial indices for the models in the classes are given by $\ze_2$ index as in Table~\ref{Ptable}.  
Actually, the integer-valued index vanishes identically. 
However, the models preserve the $\ze_2$ index given by 
\begin{equation}
\label{Z2oddNoS}
{\rm Ind}_2^{(2n+1)}(D_a,P_{\rm F})
:={\rm dim}\;{\rm ker}\;[(P_{\rm F}-D_aP_{\rm F}D_a)-1]\ \ \mbox{\rm modulo}\ 2. 
\end{equation}
in the sense that the parity ($\ze_2$ index) is continuous with respect to 
the norm of any perturbation if the perturbed Hamiltonian has the same symmetry as that of the unperturbed Hamiltonian. 
In other words, as long as the Fermi level lies in the spectral gap of the Hamiltonian or in the localization regime, 
the $\ze_2$ index is robust against any perturbation which preserves the same symmetry as that of the unperturbed Hamiltonian.
\end{thm}

The proof is given in Sections~\ref{sec:NoSS}, \ref{sec:AIAII} and \ref{sec:CD}. The continuity of the $\ze_2$ indices 
can be proved in the same way as in the proof of Theorem~\ref{thm:z2even}.

\begin{rem} (i) For numerical calculations of the $\ze_2$ index in three-dimensional systems with disorder, 
see \cite{Guo,LeungProdan}. They used a method of twisted boundary conditions \cite{NTW,KaneMele}. 
For a scattering matrix approach to the $\ze_2$ index in three-dimensional systems with disorder, see \cite{FHA,SbBr} 
and a related article \cite{FHAB}. A numerical demonstration of the present formula (\ref{Z2oddNoS}) 
was presented in \cite{AKK}.   
\end{rem}

Next, consider the case where the Hamiltonian $H$ has chiral symmetry with a chiral operator $S$. 
As mentioned in Sec.~\ref{sec:ChiralModel}, we deal with the following two situations: 
(i) There is a spectral gap between the upper and lower bands in 
the spectrum of the chiral symmetric Hamiltonian $H$. (ii) When the Fermi level $E_{\rm F}=0$ lies in 
a localization regime, we assume that $E_{\rm F}=0$ is not an eigenvalue of the chiral Hamiltonian $H$ with 
probability one with respect to configurations of the random quantities of the system.  
In such a system, the resolvent $(E_{\rm F}-H)^{-1}$ decays exponentially with distance. 
We write $P_\pm$ for the spectral projections onto the upper and the lower bands, respectively.  
Using the two projections $P_\pm$, we introduce a unitary operator, 
\begin{equation}
\label{U}
U:=P_+-P_-.
\end{equation}
We also define a projection operator,
\begin{equation}
\label{projectionchiral}
\mathcal{P}_{\rm D}:=\frac{1}{2}(1+D_a),
\end{equation} 
in terms of the Dirac operator $D_a$. 

\begin{thm}
\label{NonCommuThmOdd}
Assume that the Hamiltonian $H$ has chiral symmetry with a chiral operator $S$ in odd dimensions. 
Then, the following relation between the integer-valued index and the topological invariant 
(generalized Chern number) is valid:
\begin{multline*}
\frac{1}{2}\Bigl\{{\rm dim}\;{\rm ker}\;[S(\mathcal{P}_{\rm D}-U\mathcal{P}_{\rm D}U)-1]\\
-{\rm dim}\;{\rm ker}\;[S(\mathcal{P}_{\rm D}-U\mathcal{P}_{\rm D}U)+1]\Bigr\}\\
=\frac{i^{n}}{(2n+1)!!\cdot 2\cdot \pi^{n+1}}\sum_\sigma(-1)^\sigma 
\oint dz_1\oint dz_2\cdots \oint dz_{2n+1}{\rm Tr}\; 
SU\frac{1}{z_1-H}\\ 
\times J_a^{(\sigma_1)}\frac{1}{z_1-H}
\frac{1}{z_2-H}J_a^{(\sigma_2)}\frac{1}{z_2-H}\cdots\frac{1}{z_{2n+1}-H}J_a^{(\sigma_{2n+1})}\frac{1}{z_{2n+1}-H}. 
\end{multline*}
Further, the integer-valued index is continuous with respect to 
the norm of any perturbation if the perturbed Hamiltonian has the same symmetry as that of the unperturbed Hamiltonian. 
Namely, as long as the Fermi level lies in a spectral gap of the Hamiltonian or in a localization regime, 
the integer-valued index 
is robust against any perturbation which preserves the same symmetry as that of the unperturbed Hamiltonian. 
\end{thm}

The proof of the index formula for the most typical class AIII, which has only chiral symmetry, 
is given in Sec.~\ref{sec:ChiralSymmOddD}. For the rest of the classes that have various symmetries, 
the proofs are given in Sec.~\ref{sec:SSO}. The continuity of the indices can be proved in 
the same way as in the proof of Theorem~\ref{NonCommuIntThmEven}.

\begin{rem}
(i) In \cite{MHSP}, it was shown numerically for disordered models that the noncommutative analogue 
of the Chern number for odd-dimensional systems with chiral symmetry is quantized to an integer.
See also \cite{SFP}. 
The mathematical proof was given by Prodan and Schulz-Baldes in \cite{PSB}. 
In addition to the assumption about the spectral gap or localization, they required that 
the probability measure for disorder is ergodic with respect to spatial translations. 
For the advantage of our approach, see Remark~(ii) of Theorem~\ref{NonCommuIntThmEven}.  
\smallskip

\noindent
(ii) As we will show in Sec.~\ref{subsec:ChiralIndex} below, the index in Theorem~\ref{NonCommuThmOdd} 
is equal to the index of a Fredholm operator which is given by a pairing of the Dirac operator and 
a unitary operator. Therefore, Theorem~\ref{NonCommuThmOdd} is closely related to the local index formula  
in Corollary~II.1 in \cite{CM} although the two expressions for the index formula seem to be totally different. 
See also \cite{Higson1,Higson2}. 
\end{rem}

\begin{thm}
Assume that the Hamiltonian $H$ has chiral symmetry with a chiral operator $S$ in odd dimensions. 
Then, for the entries specified by the symbol $\ze_2$ in Table~\ref{Ptable}, 
the integer-valued index given in Theorem~\ref{NonCommuThmOdd} vanishes identically. 
However, the following $\ze_2$ index may be nonvanishing:
$$
{\rm Ind}_2^{(2n)}(D_a,S,U)
:=\frac{1}{2}{\rm dim}\;{\rm ker}\;[S(\mathcal{P}_{\rm D}-U\mathcal{P}_{\rm D}U)-1]\ \ \mbox{\rm modulo}\ 2.
$$
More precisely, this $\ze_2$ index is continuous with respect to 
the norm of any perturbation if the perturbed Hamiltonian has the same symmetry as that of the unperturbed Hamiltonian. 
In other words, as long as the Fermi level lies in a spectral gap of the Hamiltonian or in a localization regime, 
the $\ze_2$ index is robust against any perturbation which preserves the same symmetry as that of the unperturbed Hamiltonian.  
\end{thm}

The proofs are given in Sec.~\ref{sec:SSO}. The continuity of the $\ze_2$ indices can be proved 
in the same way as in the proof of Theorem~\ref{thm:z2even}.

\begin{rem}
(i) Recently, Gro{\ss}mann and Schulz-Baldes \cite{GSB} investigated index pairings of projections and unitaries  
for understanding the indices in the periodic table. Their projection operators which are associated to 
the Dirac operator is slightly different from ours. In consequence, they succeeded in understanding 
part of the indices in the periodic table in terms of their classifying scheme. 
\smallskip

\noindent
(ii) Roughly speaking, one of the central issues of K-theoretic approaches 
to topological insulators and superconductors is   
to determine an algebraic structure of homotopically equivalent classes of gapped Hamiltonians. 
It is believed that 
there exists a one-to-one correspondence between topological phases of such gapped systems 
and abelian groups, $\ze, 2\ze, \ze_2$, in the periodic table. 
For related K-theoretic approaches to 
the periodic table for disordered systems, 
see \cite{Thiang,BCR,Kellendonk1,Kubota,BKR,Kellendonk2}.  
\end{rem}

\section{No Symmetry in Even Dimensions}
\label{sec:NosymmEvenD}
\subsection{Index}
\label{NoSymEvenDInd}

Let us first consider the present models in even dimensions, $d=2n$, $n=1,2,\ldots$. 
To begin with, we introduce  
\begin{equation}
\label{Ind(2n)DPA}
{\rm Ind}^{(2n)}(D_a,P_{\rm F}):=\frac{1}{2}\;{\rm Tr}\; A^{2n+1}
\end{equation}
with the operator, 
\begin{equation}
\label{AEvenD}
A=\gamma^{(2n+1)}(P_{\rm F}-D_aP_{\rm F}D_a). 
\end{equation}
In Appendix~\ref{traceclassA}, we show that 
the operator $A^{2n+1}$ in the right-hand side is trace class, 
and hence the quantity, ${\rm Ind}^{(2n)}(D_a,P_{\rm F})$, is well defined.   

Firstly, we show that the quantity, ${\rm Ind}^{(2n)}(D_a,P_{\rm F})$, is equal to 
the integer-valued index in Theorem~\ref{NonCommuIntThmEven} in the preceding section.  
For this purpose, we introduce another operator, 
$$
B=\gamma^{(2n+1)}(1-P_{\rm F}-D_aP_{\rm F}D_a).
$$
Note that the Dirac operator $D_a$ satisfies $(D_a)^2=1$ from the definition (\ref{defDiracEven}). 
{From} the anticommutation relation (\ref{gammaACR}), the Dirac operator $D_a$ anticommutes with $\gamma^{(2n+1)}$ as  
\begin{equation}
\label{anticommGammaD}
\{\gamma^{(2n+1)},D_a\}=0.
\end{equation}
By using these relations  and $(\gamma^{(2n+1)})^2=1$, 
one can easily show that the above two operators, $A$ and $B$, satisfy \cite{ASS2} 
\begin{equation}
\label{A2B21}
A^2+B^2=1
\end{equation}
and 
\begin{equation}
\label{anticommuABBA1}
AB+BA=0. 
\end{equation}
In addition, one has 
\begin{equation}
\label{ABcommugamma}
[A,\gamma^{(2n+1)}]=0\quad\mbox{and}\quad[B,\gamma^{(2n+1)}]=0.
\end{equation}
Following the argument of \cite{ASS2}, we show that, if the eigenvalue $\lambda$ of $A$ satisfies 
$\lambda\in(-1,0)\cup(0,1)$, then $\lambda$ and $-\lambda$ come in pairs of eigenvalues of $A$.   
Besides, their eigenvectors are mapped to each other by the operator $B$. 
Let $\varphi$ be an eigenvector of $A$ with eigenvalue $\lambda\in(0,1)$, i.e., $A\varphi=\lambda\varphi$. 
Then, the anticommutation relation (\ref{anticommuABBA1}) yields 
$$
AB\varphi=-BA\varphi=-\lambda B\varphi.
$$
Similarly, from the first relation (\ref{A2B21}), one has 
$$
B^2\varphi=(1-A^2)\varphi=(1-\lambda^2)\varphi.
$$
These imply that the operation $B$ induces an invertible map from an eigenvector $\varphi$ with eigenvalue 
$\lambda\in(0,1)$ to that with $-\lambda$. Thus, there exists a one-to-one correspondence between the vectors 
$\varphi$ and $B\varphi$, and their eigenvalues come in pairs $\pm\lambda$, provided $0<|\lambda|<1$. 
Combining these observations with the fact that $A^{2n+1}$ is trace class, one has \cite{ASS2,Simon} 
$$
{\rm Tr}\; A^{2n+1}={\rm dim}\;{\rm ker}(A-1)-{\rm dim}\;{\rm ker}\;(A+1). 
$$
Thus, we obtain the left-hand side of the integer-valued index Theorem~\ref{NonCommuIntThmEven}. 

We shall show that the quantity, ${\rm Ind}^{(2n)}(D_a,P_{\rm F})$, takes an integer value only. 
{From} the anticommutation relation (\ref{anticommGammaD}), one has the expression, 
\begin{equation}
\label{Da}
D_a=\left(\begin{matrix}0 & \mathfrak{D}_a^\ast \cr \mathfrak{D}_a & 0\cr\end{matrix}\right),
\end{equation}
in the basis which diagonalizes $\gamma^{(2n+1)}$ as 
$$
\gamma^{(2n+1)}=\left(\begin{matrix}1 & 0 \cr 0 & -1 \cr \end{matrix}\right).
$$ 
Here, $\mathfrak{D}_a$ is a unitary operator because $(D_a)^2=1$. 

Note that 
$$
D_aP_{\rm F}D_a=\left(\begin{matrix}0 & \mathfrak{D}_a^\ast \cr \mathfrak{D}_a & 0\cr\end{matrix}\right)
\left(\begin{matrix}P_{\rm F} & 0 \cr 0 & P_{\rm F}\cr\end{matrix}\right)
\left(\begin{matrix}0 & \mathfrak{D}_a^\ast \cr \mathfrak{D}_a & 0\cr\end{matrix}\right)
=\left(\begin{matrix}\mathfrak{D}_a^\ast P_{\rm F}\mathfrak{D}_a & 0\cr 0 
& \mathfrak{D}_a P_{\rm F}\mathfrak{D}_a^\ast \cr \end{matrix}\right).
$$
Therefore, one has 
$$
A=\left(\begin{matrix}P_{\rm F}-\mathfrak{D}_a^\ast P_{\rm F}\mathfrak{D}_a & 0\cr 
0 & -(P_{\rm F}-\mathfrak{D}_a P_{\rm F}\mathfrak{D}_a^\ast)\cr
\end{matrix}\right).
$$
Substituting this into the right-hand side of (\ref{Ind(2n)DPA}), we have 
\begin{equation}
\label{diagoInd2n}
{\rm Ind}^{(2n)}(D_a,P_{\rm F})=\frac{1}{2}\;{\rm Tr}
\left(\begin{matrix}(P_{\rm F}-\mathfrak{D}_a^\ast P_{\rm F}\mathfrak{D}_a)^{2n+1} & 0\cr 
0 & -(P_{\rm F}-\mathfrak{D}_a P_{\rm F}\mathfrak{D}_a^\ast)^{2n+1}\cr
\end{matrix}\right).
\end{equation}
Combining the above argument with the commutation relations (\ref{ABcommugamma}), one has 
\begin{multline*}
{\rm Tr}\>(P_{\rm F}-\mathfrak{D}_a^\ast P_{\rm F}\mathfrak{D}_a)^{2n+1}\\
={\rm dim}\>{\rm ker}\>(P_{\rm F}-\mathfrak{D}_a^\ast P_{\rm F}\mathfrak{D}_a-1) 
-{\rm dim}\>{\rm ker}\>(P_{\rm F}-\mathfrak{D}_a^\ast P_{\rm F}\mathfrak{D}_a+1) 
\end{multline*}
on the subspace. Since one has 
$$
-1\le P_{\rm F}-\mathfrak{D}_a^\ast P_{\rm F}\mathfrak{D}_a\le 1,
$$ 
one obtains that 
$$
(P_{\rm F}-\mathfrak{D}_a^\ast P_{\rm F}\mathfrak{D}_a)\varphi=\varphi
$$
if and only if $(1-P_{\rm F})\varphi=0$ and $P_{\rm F}\mathfrak{D}_aP_{\rm F}\varphi=0$. 
Therefore, the following relation is valid: \cite{ASS2} 
$$
{\rm ker}\;(P_{\rm F}-\mathfrak{D}_a^\ast P_{\rm F}\mathfrak{D}_a-1)={\rm ker}\;\mathfrak{T},
$$
where $\mathfrak{T}$ is the Fredholm operator given by 
$$
\mathfrak{T}=P_{\rm F}\mathfrak{D}_a P_{\rm F}+(1-P_{\rm F}).
$$
Similarly, one obtains 
$$
{\rm ker}\;(P_{\rm F}-\mathfrak{D}_a^\ast P_{\rm F}\mathfrak{D}_a+1)=\mathfrak{D}_a^\ast\>{\rm ker}\;\mathfrak{T}^\ast
$$
by using the fact that $\mathfrak{D}_a$ is a unitary operator. 
In consequence, these observations yield 
\begin{multline*}
{\rm Tr}\>(P_{\rm F}-\mathfrak{D}_a^\ast P_{\rm F}\mathfrak{D}_a)^{2n+1}\\
={\rm dim}\>{\rm ker}\>(P_{\rm F}-\mathfrak{D}_a^\ast P_{\rm F}\mathfrak{D}_a-1) 
-{\rm dim}\>{\rm ker}\>(P_{\rm F}-\mathfrak{D}_a^\ast P_{\rm F}\mathfrak{D}_a+1)\\
={\rm dim}\>{\rm ker}\>\mathfrak{T}-{\rm dim}\>{\rm ker}\>\mathfrak{T}^\ast  
\end{multline*}
on the subspace. Similarly, one has 
$$
{\rm Tr}\>(P_{\rm F}-\mathfrak{D}_a P_{\rm F}\mathfrak{D}_a^\ast)^{2n+1}
={\rm dim}\>{\rm ker}\>\mathfrak{T}^\ast-{\rm dim}\>{\rm ker}\>\mathfrak{T}. 
$$
Substituting these into the expression (\ref{diagoInd2n}) of the index, we have 
$$
{\rm Ind}^{(2n)}(D_a,P_{\rm F})
={\rm dim}\>{\rm ker}\; \mathfrak{T}-{\rm dim}\>{\rm ker}\; \mathfrak{T}^\ast 
$$
Thus, the index, ${\rm Ind}^{(2n)}(D_a,P_{\rm F})$, is equal to the index of the Fredholm operator $\mathfrak{T}$. 

Next, we show that the index, ${\rm Ind}^{(2n)}(D_a,P_{\rm F})$, is independent of 
the location of $a$ of the Dirac operator $D_a$. We write 
$$
\mathfrak{T}'=P_{\rm F}\mathfrak{D}_{a'}P_{\rm F}+(1-P_{\rm F})
$$
with $a'\ne a$. Then, one has 
$$
\mathfrak{T}'-\mathfrak{T}=P_{\rm F}(\mathfrak{D}_{a'}-\mathfrak{D}_a)P_{\rm F}.
$$
Therefore, it is sufficient to show the operator $\mathfrak{D}_{a'}-\mathfrak{D}_a$ is compact. 
Actually, as is well known, if the difference between two Fredholm operators is compact, 
then the two Fredholm operators give the same value of the index.   
{From} the definition (\ref{defDiracEven}) of the Dirac operator $D_a$, one has 
$$
\mbox{all the matrix elements of } \left[D_{a'}(x)-D_a(x)\right]\sim \frac{1}{|x-a|}
$$
for a large $|x-a|$. On the other hand, the $(2,1)$-component of the matrix representation (\ref{Da}) 
for the Dirac operator $D_a$ is written as 
$$
\mathfrak{D}_a=\left(\begin{matrix}0,1 \cr\end{matrix}\right)D_a
\left(\begin{matrix}1 \cr 0\cr\end{matrix}\right) 
$$
in terms of the two vectors. Clearly, one has 
$$
\mathfrak{D}_{a'}-\mathfrak{D}_a=\left(\begin{matrix}0,1 \cr\end{matrix}\right)(D_{a'}-D_a)
\left(\begin{matrix}1 \cr 0\cr\end{matrix}\right). 
$$
These observations imply that $(\mathfrak{D}_{a'}-\mathfrak{D}_a)^{2n+1}$ is trace class, 
and hence $\mathfrak{D}_{a'}-\mathfrak{D}_a$ is compact.  
Consequently, the index is independent of the location $a$. 

The rest of this section is devoted to the proof that 
the integer-valued index is continuous with respect to the norm of any perturbations. 
We write $\mathfrak{A}=P_{\rm F}-\mathfrak{D}_a^\ast P_{\rm F}\mathfrak{D}_a$. 
If the eigenvalues of $A$ change continuously 
under a continuous variation of the parameters of the Hamiltonian $H$, 
then the difference, ${\rm dim}\;{\rm ker}\;(\mathfrak{A}-1)-{\rm dim}\;{\rm ker}\;(\mathfrak{A}+1)$, must be invariant 
under the deformation of the Hamiltonian. 
Namely, we have two possibilities: (i) Some eigenvectors of $\mathfrak{A}$ are lifted from the sector spanned 
by the eigenvectors of $\mathfrak{A}$ with the eigenvalue $\lambda=1$; (ii) some eigenvectors of $\mathfrak{A}$ 
with eigenvalue $\lambda \ne 1$ become degenerate with the eigenvectors of $\mathfrak{A}$ with the eigenvalue $\lambda=1$. 
In both cases, due to the map $B$, the same number of the corresponding eigenvectors with the opposite value $-\lambda$ of eigenvalue  
are simultaneously lifted from or become degenerate with the sector spanned 
by the eigenvectors of $\mathfrak{A}$ with the eigenvalue $\lambda=-1$. 
Therefore, the difference, ${\rm dim}\;{\rm ker}\;(\mathfrak{A}-1)-{\rm dim}\;{\rm ker}\;(\mathfrak{A}+1)$, 
of the dimensions between the two sectors with $\lambda=1$ and $\lambda=-1$ is invariant under 
perturbations. 
In consequence, it is enough to prove the continuity of the eigenvalues of the operator $A$ 
under deformation of the Hamiltonian \cite{RSB,Koma2}, in order to establish that the integer-valued index is 
robust against generic perturbations. This is proved in Sec.~\ref{HomoArg}.

\subsection{Topological Invariant}
\label{TopoInvEven}

In this section, we show that the index defined by (\ref{Ind(2n)DPA}) is equal to 
the topological invariant, i.e., the Chern number.   
By using the anticommutation relation (\ref{anticommGammaD}) and $(\gamma^{(2n+1)})^2=1$, 
the index can be written as 
\begin{equation}
\label{Ind(2n)DP}
{\rm Ind}^{(2n)}(D_a,P_{\rm F})=\frac{1}{2}\;{\rm Tr}\;\gamma^{(2n+1)}(P_{\rm F}-D_aP_{\rm F}D_a)^{2n+1}.
\end{equation}
For the purpose of the present section, we introduce an approximate index as 
$$
{\rm Ind}^{(2n)}(D_a,P_{\rm F};R)
:=\frac{1}{2}\;{\rm Tr}\;\gamma^{(2n+1)}(P_{\rm F}-D_aP_{\rm F}D_a)^{2n+1}\chi_R^a,
$$
where $\chi_R^a$ is the cutoff function given by 
\begin{equation}
\label{chiRa}
\chi_R^a(x):=\begin{cases} 1, & \text{$|x-a|\le R$};\\
0, & \text{otherwise}
\end{cases}
\end{equation}
with a large positive $R$. 
Since the operator $(P_{\rm F}-D_aP_{\rm F}D_a)^{2n+1}$ is trace class, one has 
$$
{\rm Ind}^{(2n)}(D_a,P_{\rm F})=\lim_{R\nearrow\infty}{\rm Ind}^{(2n)}(D_a,P_{\rm F};R).
$$

Let $\Omega\subset\re^d$ be a $d$-dimensional rectangular region whose center of gravity is the origin of $\re^d$. 
Further, we introduce 
\begin{equation}
\label{IndDaPFOmegaR}
{\rm Ind}^{(2n)}(D_a,P_{\rm F};\Omega,R):=\frac{1}{2|\Omega|}\int_\Omega dv(a) 
\;{\rm Tr}\;\gamma^{(2n+1)}(P_{\rm F}-D_aP_{\rm F}D_a)^{2n+1}\chi_R^a,
\end{equation}
where $|\Omega|$ denotes the volume of the region $\Omega$, and $v$ is the usual Lebesgue measure.  
Since the index is independent of $a$ as we showed in the above, we have 
$$
\lim_{R\nearrow\infty}{\rm Ind}^{(2n)}(D_a,P_{\rm F};\Omega,R)
={\rm Ind}^{(2n)}(D_a,P_{\rm F}).
$$
Clearly, 
$$
\lim_{\Omega\nearrow\re^d}\lim_{R\nearrow\infty}{\rm Ind}^{(2n)}(D_a,P_{\rm F};\Omega,R)
={\rm Ind}^{(2n)}(D_a,P_{\rm F}).
$$
The order of this double limit is interchangeable as follows: 

\begin{lem}
We have the relation,  
\label{lem:DlimtROmegaInd}
$$
\lim_{R\nearrow\infty}\lim_{\Omega\nearrow\re^d}{\rm Ind}^{(2n)}(D_a,P_{\rm F};\Omega,R)
={\rm Ind}^{(2n)}(D_a,P_{\rm F}).
$$
\end{lem}

The proof is given in Appendix~\ref{Appendix:DlimtROmegaInd}. 

Let $\Lambda\subset\ze^d$ be a finite lattice such that $\Lambda$ satisfies 
$\Lambda\subset \Omega$ and $|\Lambda|=|\Omega|$, where $|\Lambda|$ denotes 
the number of the sites in the lattice $\Lambda$.  
We introduce another approximate index as 
\begin{equation}
\label{IndDaPFLambdaR}
\widetilde{{\rm Ind}}^{(2n)}(D_a,P_{\rm F};\Lambda,R):=
\frac{1}{2|\Lambda|}\int_{\re^d}dv(a)\>{\rm Tr}\> \gamma^{(2n+1)}
(P_{\rm F}-D_aP_{\rm F}D_a)^{2n+1}\chi_R^a\chi_\Lambda,
\end{equation}
where
\[
\chi_\Lambda(u):=\begin{cases} 1, & \text{$u\in \Lambda$};\\
0, & \text{otherwise}.
\end{cases}
\]
Here, we stress that the integral with respect to $a$ over $\re^d$ is well defined 
because of the two cutoff function $\chi_R^a$ and $\chi_\Lambda$. 
Then, we have: 

\begin{lem}
\label{lem:IndtildeInd}
For any fixed $R$, the following two limits coincide with each other as 
\begin{equation}
\lim_{\Omega\nearrow\re^d}{\rm Ind}^{(2n)}(D_a,P_{\rm F};\Omega,R)=
\lim_{\Lambda\nearrow\ze^d}\widetilde{{\rm Ind}}^{(2n)}(D_a,P_{\rm F};\Lambda,R)
\end{equation}
in the sense of the accumulation points.
\end{lem}

The proof is given in Appendix~\ref{proof:lem:IndtildeInd}. 

Since $\chi_\Lambda$ is trace class, one has  
\begin{multline*}
\widetilde{{\rm Ind}}^{(2n)}(D_a,P_{\rm F};\Lambda,R)=
\frac{1}{2|\Lambda|}\int_{\re^d}dv(a)\>{\rm Tr}\> \gamma^{(2n+1)}
(P_{\rm F}-D_aP_{\rm F}D_a)^{2n+1}\chi_R^a\chi_\Lambda\\
=\frac{1}{2|\Lambda|}\int_{\re^d}dv(a)\>{\rm Tr}\> \gamma^{(2n+1)}
P_{\rm F}(P_{\rm F}-D_aP_{\rm F}D_a)^{2n}\chi_R^a\chi_\Lambda\\
-\frac{1}{2|\Lambda|}\int_{\re^d}dv(a)\>{\rm Tr}\> \gamma^{(2n+1)}
D_aP_{\rm F}D_a(P_{\rm F}-D_aP_{\rm F}D_a)^{2n}\chi_R^a\chi_\Lambda.
\end{multline*}
The second term in the right-hand side is written 
\begin{align*}
&\frac{1}{2|\Lambda|}\int_{\re^d}dv(a)\>{\rm Tr}\> \gamma^{(2n+1)}
D_aP_{\rm F}D_a(P_{\rm F}-D_aP_{\rm F}D_a)^{2n}\chi_R^a\chi_\Lambda\\
&=\frac{1}{2|\Lambda|}\int_{\re^d}dv(a)\>{\rm Tr}\> \gamma^{(2n+1)}
D_aP_{\rm F}(P_{\rm F}-D_aP_{\rm F}D_a)^{2n}D_a\chi_R^a\chi_\Lambda\\
&=\frac{1}{2|\Lambda|}\int_{\re^d}dv(a)\>{\rm Tr}\> D_a\gamma^{(2n+1)}
D_aP_{\rm F}(P_{\rm F}-D_aP_{\rm F}D_a)^{2n}\chi_R^a\chi_\Lambda\\
&=-\frac{1}{2|\Lambda|}\int_{\re^d}dv(a)\>{\rm Tr}\> \gamma^{(2n+1)}
P_{\rm F}(P_{\rm F}-D_aP_{\rm F}D_a)^{2n}\chi_R^a\chi_\Lambda,
\end{align*}
where we have used 
$$
D_a(P_{\rm F}-D_aP_{\rm F}D_a)=-(P_{\rm F}-D_aP_{\rm F}D_a)D_a,
$$
the property of the trace, the anticommutation relation (\ref{anticommGammaD}), and $D_a^2=1$. 
Therefore, we have 
\begin{equation}
\label{tildeIndDaPFLambdaR}
\widetilde{{\rm Ind}}^{(2n)}(D_a,P_{\rm F};\Lambda,R)=
\frac{1}{|\Lambda|}\int_{\re^d}dv(a)\>{\rm Tr}\> \gamma^{(2n+1)}
P_{\rm F}(P_{\rm F}-D_aP_{\rm F}D_a)^{2n}\chi_R^a\chi_\Lambda.
\end{equation}
Here, we stress the following: Since the operator $(P_{\rm F}-D_aP_{\rm F}D_a)^{2n}$ may not be trace class, 
we need the cutoff $\chi_R^a$. Instead of the cutoff function, an assumption of strong disorder 
about localization length was required in \cite{PLB}. 

In order to handle the right-hand side of (\ref{tildeIndDaPFLambdaR}), we introduce the complete orthonormal system of 
wavefunctions $\zeta_u^\alpha$ which are defined by  
\begin{equation}
\label{ONSzeta}
\zeta_u^\alpha:=\chi_{\{u\}}\otimes \Phi^\alpha,\quad \mbox{for \ } u\in\ze^d,\ \alpha=1,2,\ldots,d_{\rm s},   
\end{equation}
where $\chi_{\{u \}}$ is the characteristic function (\ref{chix}) of the single lattice site $u$, 
and the wavefunctions $\Phi^\alpha$ for the degree of freedom for spin or orbital 
are an orthonormal basis whose dimension of the Hilbert space is given by $d_{\rm s}$. 
Note that 
$$
P_{\rm F}-D_aP_{\rm F}D_a=[P_{\rm F},D_a]D_a=-D_a[P_{\rm F},D_a].
$$
By using this identity, $D_a^2=1$ and the identity (\ref{commuPD}) in Appendix~\ref{traceclassA}, 
the right-hand side of the index of (\ref{tildeIndDaPFLambdaR}) is written as  
\begin{align}
&\hspace{-0.3cm}\widetilde{{\rm Ind}}^{(2n)}(D_a,P_{\rm F};\Lambda,R)\nonumber\\ \nonumber 
&=\frac{1}{|\Lambda|}\int_{\re^d}dv(a)\>{\rm Tr}\> \gamma^{(2n+1)}
P_{\rm F}(P_{\rm F}-D_aP_{\rm F}D_a)^{2n}\chi_R^a\chi_\Lambda\\ \nonumber
&=\frac{(-1)^n}{|\Lambda|}\int_{\re^d}dv(a)\>{\rm Tr}\> \gamma^{(2n+1)}
P_{\rm F}[P_{\rm F},D_a]^{2n}\chi_R^a\chi_\Lambda\\ \nonumber 
&=\frac{(-1)^n}{|\Lambda|}\int_{\re^d}dv(a)\>\sum_{u_1,\ldots,u_{2n+1}}\sum_{\alpha_1,\ldots,\alpha_{2n+1}}
\langle\zeta_{u_{2n+1}}^{\alpha_{2n+1}},P_{\rm F}\zeta_{u_1}^{\alpha_1}\rangle\\ \nonumber 
&\times \langle\zeta_{u_1}^{\alpha_1},P_{\rm F}\zeta_{u_2}^{\alpha_2}\rangle \cdots 
\langle\zeta_{u_{2n}}^{\alpha_{2n}},P_{\rm F}\zeta_{u_{2n+1}}^{\alpha_{2n+1}}\rangle
\chi_R^a(u_{2n+1})\chi_\Lambda(u_{2n+1})\\ \nonumber 
&\times{\rm tr}^{(\gamma)}\gamma^{(2n+1)}
\left(\frac{u_1-a}{|u_1-a|}-\frac{u_2-a}{|u_2-a|}\right)\cdot\gamma\\ 
&\times \left(\frac{u_2-a}{|u_2-a|}-\frac{u_3-a}{|u_3-a|}\right)\cdot\gamma
\cdots \left(\frac{u_{2n}-a}{|u_{2n}-a|}-\frac{u_{2n+1}-a}{|u_{2n+1}-a|}\right)\cdot\gamma,
\label{tildeInd}
\end{align}
where the trace with respect to the gamma matrices is denoted by ${\rm tr}^{(\gamma)}$. 

We write 
\begin{align}
\mathcal{T}^{(\gamma)}&={\rm tr}^{(\gamma)}\gamma^{(2n+1)}
\left(\frac{u_1-a}{|u_1-a|}-\frac{u_2-a}{|u_2-a|}\right)\cdot\gamma 
\left(\frac{u_2-a}{|u_2-a|}-\frac{u_3-a}{|u_3-a|}\right)\cdot\gamma\nonumber\\ 
&\qquad\qquad\qquad\cdots \left(\frac{u_{2n}-a}{|u_{2n}-a|}-\frac{u_{2n+1}-a}{|u_{2n+1}-a|}\right)\cdot\gamma
\label{tracegamma}
\end{align}
for short. Following \cite{PLB}, we calculate this quantity, $\mathcal{T}^{(\gamma)}$ in the expression (\ref{tildeInd}) 
of the index. To begin with, let us list useful identities for the gamma matrices:\cite{PSB} 
\begin{itemize}
\item ${\rm tr}^{(\gamma)}\gamma^{(j_1)}\gamma^{(j_2)}\cdots\gamma^{(j_{2m+1})}=0$ if $m<n$;
\item ${\rm tr}^{(\gamma)}\gamma^{(\sigma_1)}\gamma^{(\sigma_2)}\cdots\gamma^{(\sigma_{2n+1})}=(-1)(-2i)^n(-1)^\sigma$ for  
the permutation, 
\begin{equation}
\label{permutationsigma}
\sigma=\left(\begin{matrix}1,2,\ldots,2n+1\cr \sigma_1,\sigma_2,\ldots,\sigma_{2n+1}\cr \end{matrix}\right).
\end{equation}
\end{itemize}
Using these identities, we have 
\begin{multline*}
\mathcal{T}^{(\gamma)}=\sum_{\ell=1}^{2n+1}(-1)^{\ell-1}{\rm tr}^{(\gamma)}\gamma^{(2n+1)}\frac{u_1-a}{|u_1-a|}\cdot 
\gamma\cdots 
\underline{\frac{u_\ell-a}{|u_\ell-a|}\cdot\gamma}\cdots\frac{u_{2n+1}-a}{|u_{2n+1}-a|}\cdot\gamma\\
=(-2i)^n\sum_{\ell=1}^{2n+1}(-1)^{\ell}\sum_\sigma(-1)^\sigma \frac{u_1^{(\sigma_1)}-a^{(\sigma_1)}}{|u_1-a|}
\frac{u_2^{(\sigma_2)}-a^{(\sigma_2)}}{|u_2-a|}\\
\cdots\frac{u_{\ell-1}^{(\sigma_{\ell-1})}-a^{(\sigma_{\ell-1})}}{|u_{\ell-1}-a|}
\frac{u_{\ell+1}^{(\sigma_{\ell})}-a^{(\sigma_{\ell})}}{|u_{\ell+1}-a|}
\cdots\frac{u_{2n+1}^{(\sigma_{2n})}-a^{(\sigma_{2n})}}{|u_{2n}-a|},
\end{multline*}
where the underline means that the factor is omitted, and the permutation $\sigma$ is given by 
(\ref{permutationsigma}). 
The determinant of the matrix which consists of the $2n$ vectors, $u_1-a,\cdots,\underline{u_\ell-a},\cdots,u_{2n+1}-a$, 
in the right-hand side is written 
\begin{align*}
&{\rm det}\left(\frac{u_1-a}{|u_1-a|},
\cdots,\underline{\frac{u_\ell-a}{|u_\ell-a|}},
\cdots,\frac{u_{2n+1}-a}{|u_{2n+1}-a|}\right)\\
&=
\sum_\sigma(-1)^\sigma \frac{u_1^{(\sigma_1)}-a^{(\sigma_1)}}{|u_1-a|}
\cdots\frac{u_{\ell-1}^{(\sigma_{\ell-1})}-a^{(\sigma_{\ell-1})}}{|u_{\ell-1}-a|}
\frac{u_{\ell+1}^{(\sigma_\ell)}-a^{(\sigma_\ell)}}{|u_{\ell+1}-a|}
\cdots\frac{u_{2n+1}^{(\sigma_{2n})}-a^{(\sigma_{2n})}}{|u_{2n+1}-a|}.
\end{align*}
According to \cite{Stein}, the determinant is equal to $(2n)!$ times the oriented volume of the simplex, 
$$
\left[0,\frac{u_1-a}{|u_1-a|},
\cdots,\underline{\frac{u_\ell-a}{|u_\ell-a|}},
\cdots,\frac{u_{2n+1}-a}{|u_{2n+1}-a|}\right].
$$
Therefore, the quantity, $\mathcal{T}^{(\gamma)}$, can be written as  
\begin{multline*}
\mathcal{T}^{(\gamma)}\\=(-2i)^n(2n)!\sum_{\ell=1}^{2n+1}(-1)^{\ell}{\rm Vol}
\left[0,\frac{u_1-a}{|u_1-a|},
\cdots,\underline{\frac{u_\ell-a}{|u_\ell-a|}},
\cdots,\frac{u_{2n+1}-a}{|u_{2n+1}-a|}\right],
\end{multline*}
where ${\rm Vol}[\cdots]$ denotes the oriented volume of the simplex. 
We note that 
\begin{align*}
&{\rm Vol}\left[0,\frac{u_1-a}{|u_1-a|},\cdots,\underline{\frac{u_\ell-a}{|u_\ell-a|}},
\cdots,\frac{u_{2n+1}-a}{|u_{2n+1}-a|}\right]\\
&={\rm Vol}\left[a,a+\frac{u_1-a}{|u_1-a|},\cdots,\underline{a+\frac{u_\ell-a}{|u_\ell-a|}},
\cdots,a+\frac{u_{2n+1}-a}{|u_{2n+1}-a|}\right]\\
&=(-1)^{\ell-1}{\rm Vol}\left[a+\frac{u_1-a}{|u_1-a|},\cdots,a,\cdots,a+\frac{u_{2n+1}-a}{|u_{2n+1}-a|}\right],
\end{align*}
where $a$ is located at the $\ell$-th position. We write 
$$
\mathfrak{S}_\ell(a)=\left[a+\frac{u_1-a}{|u_1-a|},\cdots,a,\cdots,a+\frac{u_{2n+1}-a}{|u_{2n+1}-a|}\right]
$$
for the simplex in the above right-hand side. By using these relations, $\mathcal{T}^{(\gamma)}$ can be expressed as 
\begin{equation}
\label{tracegammaVol}
\mathcal{T}^{(\gamma)}=-(-2i)^n(2n)!\sum_{\ell=1}^{2n+1}{\rm Vol}[\mathfrak{S}_\ell(a)] 
\end{equation}
in terms of the volume of the simplex $\mathfrak{S}_\ell(a)$. In order to calculate the right-hand side, 
we introduce the simplex,  
$$
\mathfrak{S}:=[u_1,u_2,\ldots,u_{2n+1}].
$$
If $a\in\mathfrak{S}$, then the orientation of the simplex $\mathfrak{S}_\ell(a)$ is the same as that of $\mathfrak{S}$. 
We write $\mathfrak{B}(a)$ for the unit ball centered at $a$, and define the intersection $\mathfrak{B}_\ell(a)$ between 
the ball and a large simplex as 
$$
\mathfrak{B}_\ell(a):=\lim_{L\nearrow\infty}\mathfrak{B}(a)\cap 
\left[a+L\frac{u_1-a}{|u_1-a|},\cdots,a,\cdots,a+L\frac{u_{2n+1}-a}{|u_{2n+1}-a|}\right]. 
$$ 
Here, we define the orientation of $\mathfrak{B}_\ell(a)$ by the orientation of the simplex in 
the right-hand side.  

Fix $u_{2n+1}$, and consider 
\begin{equation}
\label{frakI}
\mathfrak{I}:=\int_{\re^d}dv(a)\chi_R^a(u_{2n+1})\sum_{\ell=1}^{2n+1}{\rm Vol}[\mathfrak{S}_\ell(a)].
\end{equation}
We decompose this into two parts as 
$$
\mathfrak{I}=\mathfrak{I}_1+\mathfrak{I}_2
$$
with 
$$
\mathfrak{I}_1:=\int_{\re^d}dv(a)\chi_R^a(u_{2n+1})
\sum_{\ell=1}^{2n+1}\left\{{\rm Vol}[\mathfrak{S}_\ell(a)]-{\rm Vol}[\mathfrak{B}_\ell(a)]\right\}
$$
and 
$$
\mathfrak{I}_2:=\int_{\re^d}dv(a)\chi_R^a(u_{2n+1})
\sum_{\ell=1}^{2n+1}{\rm Vol}[\mathfrak{B}_\ell(a)].
$$
{From} (\ref{tracegamma}), (\ref{tracegammaVol}), and (\ref{frakI}), 
the index of (\ref{tildeInd}) can be decomposed into two parts as 
\begin{equation}
\label{tildeInd2}
\widetilde{{\rm Ind}}^{(2n)}(D_a,P_{\rm F};\Lambda,R) 
=-{(2i)^n(2n)!}(\mathfrak{L}_1+\mathfrak{L}_2)
\end{equation}
with 
\begin{multline}
\mathfrak{L}_j=\frac{1}{|\Lambda|}\sum_{u_1,\ldots,u_{2n+1}}\sum_{\alpha_1,\ldots,\alpha_{2n+1}}
\langle\zeta_{u_{2n+1}}^{\alpha_{2n+1}},P_{\rm F}\zeta_{u_1}^{\alpha_1}\rangle
\langle\zeta_{u_1}^{\alpha_1},P_{\rm F}\zeta_{u_2}^{\alpha_2}\rangle \\
\cdots 
\langle\zeta_{u_{2n}}^{\alpha_{2n}},P_{\rm F}\zeta_{u_{2n+1}}^{\alpha_{2n+1}}\rangle
\chi_\Lambda(u_{2n+1})\mathfrak{I}_j
\end{multline}
for $j=1,2$. 

Consider first the contribution of $\mathfrak{L}_2$. 
As shown in \cite{PLB}, the following holds: 
$$
\sum_{\ell=1}^{2n+1}{\rm Vol}[\mathfrak{B}_\ell(a)]
=\begin{cases} \pi^n/n! & \text{if $a$ inside $\mathfrak{S}$};\\
0 & \text{if $a$ outside $\mathfrak{S}$}.
\end{cases}
$$
Here, $\pi^n/n!$ is the volume of the $2n$-dimensional unit ball. 
Therefore, one has 
$$
\mathfrak{I}_2=\frac{\pi^n}{n!}{\rm Vol}[\mathfrak{S}^{(R)}],
$$
where $\mathfrak{S}^{(R)}:=\mathfrak{S}\cap \{u|\> |u-u_{2n+1}|\le R\}$. 
This implies that the contribution of $\mathfrak{L}_2$ 
in the index can be shown to converge to some value in the limit $R\nearrow\infty$ in the same way as in 
the proof of Lemma~\ref{lem:DlimtROmegaInd}. In particular, the correction is exponentially small in a finite large $R$. 

Next, we estimate the contribution of $\mathfrak{L}_1$ in the index of (\ref{tildeInd2}) for a large $R$. 
Let us consider ${\rm Vol}[\mathfrak{S}_\ell(a)]-{\rm Vol}[\mathfrak{B}_\ell(a)]$. 
Let $a^\ast$ be the inversion of $a$ relative to the center of 
the facet $u_1,\ldots,\underline{u_\ell},\ldots,u_{2n+1}$ of the simplex $\mathfrak{S}$. 
Then, the oriented volume ${\rm Vol}[\mathfrak{S}_\ell(a^\ast)]-{\rm Vol}[\mathfrak{B}_\ell(a^\ast)]$ 
has the opposite sign to that of ${\rm Vol}[\mathfrak{S}_\ell(a)]-{\rm Vol}[\mathfrak{B}_\ell(a)]$. 
Therefore, these two contributions cancel out each other in the integral of $\mathfrak{I}_1$ 
if both $a$ and $a^\ast$ are inside the cutoff region $\{u|\> |u-u_{2n+1}|\le R\}$. 
We can find $s$ and $s'$ such that $0<s<s'<1$, and that 
$a^\ast$ is inside the cutoff region $\{u|\> |u-u_{2n+1}|\le R\}$ 
for $u_i$ satisfying $|u_i-u_{2n+1}|\le sR$, $i=1,2,\ldots,\underline{\ell},\ldots,2n$, and for $a$ 
satisfying $|a-u_{2n+1}|\le s'R$. If some $u_i$ satisfies $|u_i-u_{2n+1}|\ge sR$, then 
the corresponding contribution can be estimated in the same way as in 
the proof of Lemma~\ref{lem:DlimtROmegaInd}, and the correction is exponentially small in a finite large $R$. 
Therefore, it is sufficient to treat the case that $a$ satisfies $|a-u_{2n+1}|\ge s'R$. 

Let us consider the triangle which consists of three points, $a$, $u_i$ and $u_j$. 
We write $r=|a-u_{2n+1}|$, and assume $r=|a-u_{2n+1}|\ge s'R$ with the above positive constant $s'$. 
Further, from the above observation, we assume that $|u_i-u_j|=o(R)$, where $o(\cdots)$ denotes the small order. 
We write $\theta=\angle(u_i,a,u_j)$ for the angle at the vertex $a$ in the triangle. 
For a large $R$, the angle $\theta$ behaves as 
$$
\theta\sim {|u_i-u_j|}/{r}. 
$$
Therefore, as in \cite{PLB}, one can show  
$$
\left|{\rm Vol}[\mathfrak{S}_\ell(a)]-{\rm Vol}[\mathfrak{B}_\ell(a)]\right|
\le {\rm Const.}\left[\frac{\max_{i,j}\{|u_i-u_j|\}}{r}\right]^{2n+1}.
$$
When integrating the upper bound with respect to $a$ over the cutoff region, the integrated value 
does not necessarily converge to a finite value in the limit $R\nearrow\infty$ 
because $|u_i-u_j|$ may tend to infinity as $r$ goes to infinity.  
But, from Assumption~\ref{Assumption}, one has 
\begin{equation}
\label{decayPF}
\left|\langle \zeta_u^\alpha,P_{\rm F}\zeta_v^\beta\rangle\right|\le{\rm Const.}e^{-\kappa|u-v|}
\end{equation}
with a positive constant $\kappa$. Clearly, the following bound holds:   
$$
|u_i-u_j|\exp[-\kappa'|u_i-u_j|]\le{\rm Const.}
$$
for a positive constant $\kappa'$. Therefore, one has 
\begin{align*}
&\int_{s'R\le r\le R} dv(a)\left\{\frac{\max_{i,j}\{|u_i-u_j|\exp[-\kappa'|u_i-u_j|]\}}{r}\right\}^{2n+1}\\
&\le {\rm Const.}\int_{s'R}^\infty dr r^{2n-1}\frac{1}{r^{2n+1}}
\le{{\rm Const.}}/{R}.
\end{align*}
Combining these observations with the argument in the proof of Lemma~\ref{lem:DlimtROmegaInd}, 
the contribution of $\mathfrak{L}_1$ in the index can be estimated as  
\begin{equation}
\label{L1bound}
|\mathfrak{L}_1|\le{\rm Const.}/{R}. 
\end{equation}
The constant in the right-hand side does not depend on the lattice $\Lambda$. 
Consequently, we obtain
\begin{multline}
\label{tildeIndRinfty}
\lim_{R\nearrow\infty}\widetilde{{\rm Ind}}^{(2n)}(D_a,P_{\rm F};\Lambda,R)\\
=-\frac{(2\pi i)^n(2n)!}{n!|\Lambda|}\sum_{u_1,\ldots,u_{2n}}\sum_{u_{2n+1}\in\Lambda}\sum_{\alpha_1,\ldots,\alpha_{2n+1}}
\langle\zeta_{u_{2n+1}}^{\alpha_{2n+1}},P_{\rm F}\zeta_{u_1}^{\alpha_1}\rangle \\ 
\times \langle\zeta_{u_1}^{\alpha_1},P_{\rm F}\zeta_{u_2}^{\alpha_2}\rangle 
\cdots 
\langle\zeta_{u_{2n}}^{\alpha_{2n}},P_{\rm F}\zeta_{u_{2n+1}}^{\alpha_{2n+1}}\rangle {\rm Vol}[\mathfrak{S}].
\end{multline}
{From} the above argument, (\ref{tildeInd}) and (\ref{tracegammaVol}), 
a reader might think that the index $\widetilde{{\rm Ind}}^{(2n)}(D_a,P_{\rm F};\Lambda,R)$ 
is vanishing in the limit $R\nearrow\infty$ because ${\rm Vol}[\mathfrak{S}_j(a)]$ and 
${\rm Vol}[\mathfrak{S}_j(a^\ast)]$ cancel out each other in the limit $R\nearrow\infty$. 
But one has 
$$
{\rm Vol}[\mathfrak{S}_j(a)]\sim \frac{1}{r^{2n-1}}\quad 
\mbox{for a large \ } r=\min_{\ell\ne j}\{|u_\ell-a|\}.
$$
Immediately, 
$$
\int_{\re^d}dv(a)\chi_R^a(u_{2n+1}){\rm Vol}[\mathfrak{S}_j(a)]
\sim \int_{R_0}^R dr r^{2n-1}\frac{1}{r^{2n-1}}=R-R_0
$$
with some positive constant $R_0$. Thus, the expression (\ref{tildeInd}) of the index is  
ill defined without the cutoff function $\chi_R^a$. Namely, the result strongly depends on  
how to sum up the sequence. In order to avoid this difficulty, we have introduced the cutoff function $\chi_R^a$. 

We write 
$$
\tilde{I}(\Lambda,R)=\widetilde{{\rm Ind}}^{(2n)}(D_a,P_{\rm F};\Lambda,R)
$$
and
$$
\tilde{I}(\Lambda,\infty)=\lim_{R\nearrow\infty}\widetilde{{\rm Ind}}^{(2n)}(D_a,P_{\rm F};\Lambda,R)
$$
for short. Then, the above result (\ref{L1bound}) implies 
\begin{equation}
\label{tildeIbound}
|\tilde{I}(\Lambda,R)-\tilde{I}(\Lambda,\infty)|\le {{\rm Const.}}/{R}
\end{equation}
for a large $R$, where the constant in the right-hands side does not depend on the lattice $\Lambda$. 
Further, we obtain:

\begin{lem}
The following relations are valid: 
\begin{align}
{\rm Ind}^{(2n)}(D_a,P_{\rm F})&=\lim_{R\nearrow\infty}\lim_{\Lambda\nearrow\ze^d}
\widetilde{{\rm Ind}}^{(2n)}(D_a,P_{\rm F};\Lambda,R)\nonumber\\  
&=\lim_{\Lambda\nearrow\ze^d}\lim_{R\nearrow\infty}
\widetilde{{\rm Ind}}^{(2n)}(D_a,P_{\rm F};\Lambda,R).
\label{IndtildeInd}
\end{align}
\end{lem}

\begin{proof}
The first equality of (\ref{IndtildeInd}) follows from Lemmas~\ref{lem:DlimtROmegaInd} 
and \ref{lem:IndtildeInd}.

In order to prove the second equality, we write 
$$
I(\Omega,R)={\rm Ind}^{(2n)}(D_a,P_{\rm F};\Omega,R)
$$
for short. For any given small $\epsilon>0$, any large $\Lambda$, and any large $\Omega$, 
we can find a large $R>0$ such that 
$$
|\tilde{I}(\Lambda,R)-\tilde{I}(\Lambda,\infty)|\le {{\rm Const.}}/{R}<\epsilon/3
$$
from (\ref{tildeIbound}), and that 
$$
|I(\Omega,R)-{\rm Ind}^{(2n)}(D_a,P_{\rm F})|\le{{\rm Const.}}/{R}<\epsilon/3
$$
from (\ref{IndRIndinfty}) and the estimate for the correction in Appendix~\ref{Appendix:DlimtROmegaInd}. 
For a fixed $R$ satisfying these conditions, we can find a large $\Lambda$ and a large $\Omega$ such that 
$$
|\tilde{I}(\Lambda,R)-I(\Omega,R)|<\epsilon/3
$$
as in the proof of Lemma~\ref{lem:IndtildeInd} in Appendix~\ref{proof:lem:IndtildeInd}. From these inequalities, 
we obtain  
\begin{align*}
|\tilde{I}(\Lambda,\infty)-{\rm Ind}^{(2n)}(D_a,P_{\rm F})|&\le |\tilde{I}(\Lambda,\infty)-\tilde{I}(\Lambda,R)|
+|\tilde{I}(\Lambda,R)-I(\Omega,R)|\\
&+|I(\Omega,R)-{\rm Ind}^{(2n)}(D_a,P_{\rm F})|\le \epsilon.
\end{align*}
This implies the second equality. 
\end{proof}

{From} this lemma, the index ${\rm Ind}^{(2n)}(D_a,P_{\rm F})$ can be derived from  
taking the limit $\Lambda\nearrow\ze^d$ for the expression of the right-hand side of (\ref{tildeIndRinfty}).  
The volume ${\rm Vol}[\mathfrak{S}]$ of the simplex in the expression is written 
\begin{align}
(2n)!{\rm Vol}[\mathfrak{S}]&={\rm det}[u_1-u_{2n+1},u_2-u_{2n+1},\ldots,u_{2n}-u_{2n+1}]\nonumber\\ \nonumber 
&={\rm det}[u_1,u_2,\ldots,u_{2n}]\\ \nonumber 
&\qquad\qquad\qquad-\sum_{\ell=1}^{2n}{\rm det}[u_1,u_2,\ldots,u_{\ell-1},u_{2n+1},u_{\ell+1},\ldots,u_{2n}]\\ \nonumber 
&=\sum_{\ell=1}^{2n}(-1)^{\ell-1}{\rm det}[u_1,u_2,\ldots,\underline{u_\ell},\ldots,u_{2n+1}]\\ \nonumber 
&={\rm det}[u_1-u_2,u_2-u_3,\ldots,u_{2n-1}-u_{2n},u_{2n}-u_{2n+1}]\\  
&=\sum_\sigma(-1)^\sigma\bigl(u_1^{(\sigma_1)}-u_2^{(\sigma_1)}\bigr)\bigl(u_2^{(\sigma_2)}-u_3^{(\sigma_2)}\bigr)
\cdots\bigl(u_{2n}^{(\sigma_{2n})}-u_{2n+1}^{(\sigma_{2n})}\bigr).
\label{Volsimplexdet}
\end{align}
In order to rewrite the right-hand side of (\ref{tildeIndRinfty}), 
we define the position operator $X=(X^{(1)},\ldots,X^{(2n)})$ by 
\begin{equation}
\label{X}
(X^{(\ell)}\varphi)_\alpha(x)=x^{(\ell)}\varphi_\alpha(x)\quad \mbox{for } x\in \ze^d,
\end{equation}
where $\varphi_\alpha\in\ell^2(\ze^d,\co^M)$ is a wavefunction. Combining this with the above observation, 
we have:

\begin{thm}
The index can be written as 
\begin{multline}
{\rm Ind}^{(2n)}(D_a,P_{\rm F})\\
=-\frac{(2\pi i)^n}{n!}\lim_{\Lambda\nearrow\ze^d}\frac{1}{|\Lambda|}\sum_\sigma (-1)^\sigma{\rm Tr}\> 
\chi_\Lambda P_{\rm F}[X^{(\sigma_1)},P_{\rm F}]\cdots[X^{(\sigma_{2n})},P_{\rm F}].
\end{multline}
\end{thm}

The $j$-th component of the position operator $X$ is approximated by 
$$
\hat{X}^{(j)}:=\sum_{a^{(j)}=-\tilde{L}+1}^{\tilde{L}}\vartheta_a^{(j)}-\tilde{L},
$$
with a large integer $\tilde{L}>0$, where $\vartheta_a^{(j)}$ is given by (\ref{varthetaaj}).  
Substituting this into the expression of the index, we have 
\begin{multline}
{\rm Ind}^{(2n)}(D_a,P_{\rm F})\\
=-\frac{(2\pi i)^n}{n!}\lim_{\Lambda\nearrow\ze^d}\frac{1}{|\Lambda|}
\sum_{a^{(1)}}\cdots \sum_{a^{(2n)}}
\sum_\sigma (-1)^\sigma{\rm Tr}\> 
\chi_\Lambda P_{\rm F}[\vartheta_a^{(\sigma_1)},P_{\rm F}]\\
\cdots[\vartheta_a^{(\sigma_{2n})},P_{\rm F}]
\end{multline}
because the matrix elements of $P_{\rm F}$ decay exponentially with large distance. 

We define an index as 
$$
{\rm Ind}^{(2n)}(\vartheta_a,P_{\rm F}):=-\frac{(2\pi i)^n}{n!}
\sum_\sigma (-1)^\sigma{\rm Tr}\> P_{\rm F}[\vartheta_a^{(\sigma_1)},P_{\rm F}]
\cdots[\vartheta_a^{(\sigma_{2n})},P_{\rm F}].
$$

\begin{thm}
The equality between the two indices is valid as 
\label{thm:IndDaIndtheta}
$$
{\rm Ind}^{(2n)}(D_a,P_{\rm F})={\rm Ind}^{(2n)}(\vartheta_a,P_{\rm F}).
$$
\end{thm}

The proof is given in Appendix~\ref{AppenProofthm:IndDaIndtheta}. 

In order to derive a useful expression of the index ${\rm Ind}^{(2n)}(\vartheta_a,P_{\rm F})$, 
we recall the expression (\ref{PFcontour}) of the projection $P_{\rm F}$ onto the Fermi sea 
in terms of the contour integral. 
Using this expression, the commutator in the index ${\rm Ind}^{(2n)}(\vartheta_a,P_{\rm F})$ is written   
$$
[\vartheta_a^{(j)},P_{\rm F}]=\frac{1}{2\pi i}\oint dz 
\left[\vartheta_a^{(j)}\frac{1}{z-H}-\frac{1}{z-H}\vartheta_a^{(j)}\right]. 
$$
The integrand in the right-hand side is computed as 
\begin{multline*}
\vartheta_a^{(j)}\frac{1}{z-H}-\frac{1}{z-H}\vartheta_a^{(j)}\\
=\frac{1}{z-H}(z-H)\vartheta_a^{(j)}\frac{1}{z-H}
-\frac{1}{z-H}\vartheta_a^{(j)}(z-H)\frac{1}{z-H}\\
=-\frac{1}{z-H}[H,\vartheta_a^{(j)}]\frac{1}{z-H}.
\end{multline*}
Therefore, one has 
$$
[\vartheta_a^{(j)},P_{\rm F}]=\frac{i}{2\pi i}\oint dz \frac{1}{z-H}J_a^{(j)}\frac{1}{z-H},
$$
where $J_a^{(j)}$ is the current operator which is given by (\ref{currentJaj}). 
By using the expression for the commutator, we have 
\begin{multline*} 
{\rm Ind}^{(2n)}(\vartheta_a,P_{\rm F})\\=\frac{(-1)^{n-1}}{n!(2\pi i)^n}\sum_\sigma(-1)^\sigma 
\oint dz_1\oint dz_2\cdots \oint dz_{2n}{\rm Tr}\; 
P_{\rm F}\frac{1}{z_1-H}J_a^{(\sigma_1)}\frac{1}{z_1-H}\\
\times \frac{1}{z_2-H}J_a^{(\sigma_2)}\frac{1}{z_2-H}\cdots\frac{1}{z_{2n}-H}J_a^{(\sigma_{2n})}\frac{1}{z_{2n}-H}.
\end{multline*}
This right-hand side is nothing but the generalized Chern number 
in the noncommutative index Theorem~\ref{NonCommuIntThmEven}.

It is often useful to express the Chern number in terms of finite-volume quantities, 
and hence we want to approximate this right-hand side by the operators on the finite-volume lattice $\Lambda$.  
We assume that the corresponding finite system with the periodic boundary condition 
has a spectral gap or a localization regime between the upper and lower bands, too. 
Then, relying on the locality of the current operator $J_a^{(j)}$ and Assumption~\ref{Assumption}, 
the above right-hand side can be approximated by 
the corresponding finite systems on the finite lattice $\Lambda$.  
We write $\varphi_{\pm,j}$ for the energy eigenvector of the corresponding finite-volume Hamiltonian 
with the eigenvalue $E_{\pm,j}$, where the subscript $\pm$ denotes the index for the upper and lower bands, 
respectively.  
Consequently, we obtain the relation between the index and the topological invariant: 
\begin{multline}
\label{currentChern} 
{\rm Ind}^{(2n)}(\vartheta_a,P_{\rm F})\\=\frac{(-1)^{n-1}(2\pi i)^n}{n!}\lim_{\Lambda\nearrow\ze^{2n}}
\sum_\sigma(-1)^\sigma \sum_{j_1,j_2,\ldots,j_{2n}} 
\frac{\langle \varphi_{-,j_1},J_a^{(\sigma_1)}\varphi_{+,j_2}\rangle}{E_{-,j_1}-E_{+,j_2}}\\
\times\frac{\langle \varphi_{+,j_2},J_a^{(\sigma_2)}\varphi_{-,j_3}\rangle}{E_{-,j_3}-E_{+,j_2}}
\cdots\frac{\langle \varphi_{-,j_{2n-1}},J_a^{(\sigma_{2n-1})}\varphi_{+,j_{2n}}\rangle}{E_{-,j_{2n-1}}-E_{+,j_{2n}}}
\frac{\langle \varphi_{+,j_{2n}},J_a^{(\sigma_{2n})}\varphi_{-,j_1}\rangle}{E_{-,j_1}-E_{+,j_{2n}}}. 
\end{multline}
This is an extension of the Hall conductance formula which leads to the Chern number 
in two dimensions \cite{Kohmoto}.

\subsection{Translationally Invariant Systems in Even Dimensions}

As an example, let us consider a translationally invariant Hamiltonian $H_\Lambda$ on 
the hypercubic box $\Lambda$ with the side length $L$. 
The energy eigenvectors of the Hamiltonian $H_\Lambda$ can be written as  $\varphi_{\pm,{\bf k}}$ 
with the energy eigenvalue $E_{\pm,{\bf k}}$ 
in terms of the momentum ${\bf k}:=(k^{(1)},k^{(2)},\ldots,k^{(2n)})$.   
We write 
$$
J^{(j)}:=\sum_{a_j}J_a^{(j)}. 
$$
Then, the index of (\ref{currentChern}) is written as 
\begin{multline*} 
{\rm Ind}^{(2n)}(\vartheta_a,P_{\rm F})\\=\lim_{\Lambda\nearrow\ze^{2n}}\frac{(-1)^{n-1}(2\pi i)^n}{n!L^{2n}}\sum_\sigma(-1)^\sigma
\sum_{{\bf k}} 
\frac{\langle \varphi_{-,{\bf k}},J^{(\sigma_1)}\varphi_{+,{\bf k}}\rangle}{E_{-,{\bf k}}-E_{+,{\bf k}}}\\
\times\frac{\langle \varphi_{+,{\bf k}},J^{(\sigma_2)}\varphi_{-,{\bf k}}\rangle}{E_{-,{\bf k}}-E_{+,{\bf k}}}
\cdots\frac{\langle \varphi_{-,{\bf k}},J^{(\sigma_{2n-1})}\varphi_{+,{\bf k}}\rangle}{E_{-,{\bf k}}-E_{+,{\bf k}}}
\frac{\langle \varphi_{+,{\bf k}},J^{(\sigma_{2n})}\varphi_{-,{\bf k}}\rangle}{E_{-,{\bf k}}-E_{+,{\bf k}}}\\
=\frac{(-1)^{n-1}i^n}{n!(2\pi)^{n}}\sum_\sigma(-1)^\sigma
\int dk_1dk_2\cdots dk_{2n} 
\frac{\langle \varphi_{-,{\bf k}},J^{(\sigma_1)}\varphi_{+,{\bf k}}\rangle}{E_{-,{\bf k}}-E_{+,{\bf k}}}\\
\times\frac{\langle \varphi_{+,{\bf k}},J^{(\sigma_2)}\varphi_{-,{\bf k}}\rangle}{E_{-,{\bf k}}-E_{+,{\bf k}}}
\cdots\frac{\langle \varphi_{-,{\bf k}},J^{(\sigma_{2n-1})}\varphi_{+,{\bf k}}\rangle}{E_{-,{\bf k}}-E_{+,{\bf k}}}
\frac{\langle \varphi_{+,{\bf k}},J^{(\sigma_{2n})}\varphi_{-,{\bf k}}\rangle}{E_{-,{\bf k}}-E_{+,{\bf k}}}.
\end{multline*}
As usual, we write $H({\bf k})$ for the Fourier transform of the Hamiltonian, and   
$$
P_{\rm F}({\bf k})=\frac{1}{2\pi i}\oint dz \frac{1}{z-H({\bf k})}
$$
for the projection onto the Fermi sea with the momentum ${\bf k}$. Then, one has 
$$
\langle \varphi_{\pm,{\bf k}},J^{(j)}\varphi_{\mp,{\bf k}}\rangle
=\langle\tilde{\varphi}_{\pm,{\bf k}},J^{(j)}({\bf k})\tilde{\varphi}_{\mp,{\bf k}}\rangle, 
$$
where $\tilde{\varphi}_{\pm,{\bf k}}$ is an eigenvector of $H({\bf k})$, and 
$$
J^{(j)}({\bf k}):=\frac{\partial}{\partial k^{(j)}}H({\bf k}). 
$$
Note that 
\begin{align*}
\Bigl\langle\tilde{\varphi}_{\pm,{\bf k}},\frac{\partial P_{\rm F}({\bf k})}{\partial k^{(j)}}
\tilde{\varphi}_{\mp,{\bf k}}\Bigr\rangle
&=\frac{-1}{2\pi i}\oint dz \Bigl\langle\tilde{\varphi}_{\pm,{\bf k}},\frac{1}{z-H({\bf k})}J^{(j)}({\bf k})
\frac{1}{z-H({\bf k})}\tilde{\varphi}_{\mp,{\bf k}}\Bigr\rangle\\
&=-\frac{\Bigl\langle\tilde{\varphi}_{\pm,{\bf k}},J^{(j)}({\bf k})
\tilde{\varphi}_{\mp,{\bf k}}\Bigr\rangle}{ E_{-,{\bf k}}-E_{+,{\bf k}} }.
\end{align*}
{From} these observations, one obtains 
\begin{multline}
\label{IndDelPeven} 
{\rm Ind}^{(2n)}(\vartheta_a,P_{\rm F})\\
=\frac{(-1)^{n-1}i^n}{n!(2\pi)^{n}}\sum_\sigma(-1)^\sigma 
\int dk^{(1)}dk^{(2)}\cdots dk^{(2n)}\; {\rm Tr}\; P_{\rm F}({\bf k})\\
\times \frac{\partial P_{\rm F}({\bf k})}{\partial k^{(\sigma_1)}}
\frac{\partial P_{\rm F}({\bf k})}{\partial k^{(\sigma_2)}}
\cdots\frac{\partial P_{\rm F}({\bf k})}{\partial k^{(\sigma_{2n})}}.
\end{multline}

More concretely, let us consider the Hamiltonian $H$ which is given by \cite{QWZ}
\begin{align*}
(H\varphi)(x)&=\sum_{j=1}^{2n} \frac{t_{\rm s}^{(j)}}{2i}
[\varphi(x+e^{(j)})-\varphi(x-e^{(j)})]\gamma^{(j)}\\
&+\Big\{m_0+\sum_{j=1}^{2n}\frac{t_{\rm c}^{(j)}}{2}[\varphi(x+e^{(j)})+\varphi(x-e^{(j)})]
\Big\}\gamma^{(2n+1)}, 
\end{align*}
where $t_{\rm s}^{(j)}$, $t_{\rm c}^{(j)}$ and $m_0$ are real constants, $e^{(j)}$ 
are the unit vectors in the $j$-th direction, and $\gamma^{(j)}$ are the gamma matrices.  
The Fourier transform is 
$$
H({\bf k})=\gamma\cdot\mathcal{E}({\bf k})=\sum_{i=1}^{2n+1} \gamma^{(i)}\mathcal{E}^{(i)}({\bf k}),
$$
where we have written 
$\mathcal{E}({\bf k}):=(\mathcal{E}^{(1)}({\bf k}),\mathcal{E}^{(2)}({\bf k}),\ldots,\mathcal{E}^{(2n+1)}({\bf k}))$ 
with\hfill\break $\mathcal{E}^{(j)}({\bf k})=t_{\rm s}^{(j)}\sin k^{(j)}$ for $j=1,2,\ldots,2n$, 
and 
$$
\mathcal{E}^{(2n+1)}({\bf k})=m_0+\sum_{j=1}^{2n}t_{\rm c}^{(j)}\cos k^{(j)}.
$$
{From} the properties of the gamma matrices, one has 
$$
H({\bf k})^2=|\mathcal{E}({\bf k})|^2\quad 
\mbox{with} \quad 
|\mathcal{E}({\bf k})|=\sqrt{\sum_{i=1}^{2n+1} [\mathcal{E}^{(i)}({\bf k})]^2}.
$$
Therefore, if $|\mathcal{E}({\bf k})|$ is nonvanishing for all ${\bf k}$, 
then there appears a nonvanishing spectral gap between the upper and lower bands. 
In the following, we consider such a situation. 

The projection onto the lower band can be written as 
$$
P_{\rm F}({\bf k})=\frac{1}{2}\Bigl[1-\frac{1}{|\mathcal{E}({\bf k})|}\gamma\cdot\mathcal{E}({\bf k})\Bigr].
$$
We define the unit vector $n({\bf k})$ with the $j$-th component, 
$$
n^{(j)}({\bf k}):=\frac{1}{|\mathcal{E}({\bf k})|}\mathcal{E}^{(j)}({\bf k}), 
$$
for $j=1,2,\ldots,2n+1$. Substituting these into the right-hand side of (\ref{IndDelPeven}) and 
using the properties of the gamma matrices, we obtain  
\begin{multline*} 
{\rm Ind}^{(2n)}(\vartheta_a,P_{\rm F})
=\frac{(-1)^{n-1}(2n)!}{n!2^{2n+1}\pi^n}\sum_\tau (-1)^\tau 
\int dk^{(1)}dk^{(2)}\cdots dk^{(2n)}\;
n^{(\tau_1)}({\bf k})\\
\times\frac{\partial n^{(\tau_2)}({\bf k})}{\partial k^{(1)}}
\frac{\partial n^{(\tau_3)}({\bf k})}{\partial k^{(2)}}
\cdots\frac{\partial n^{(\tau_{2n+1})}({\bf k})}{\partial k^{(2n)}}.
\end{multline*}
Here, the quantity in the right-hand side is equal to the surface area, $\left|\S^{2n}\right|$, 
of $2n$-dimensional unit sphere, $\S^{2n}$, multiplied by 
the winding number $\nu\in\ze$ as 
\begin{multline*}
\sum_\tau (-1)^\tau 
\int dk^{(1)}dk^{(2)}\cdots dk^{(2n)}\;
n^{(\tau_1)}({\bf k})\\
\times\frac{\partial n^{(\tau_2)}({\bf k})}{\partial k^{(1)}}
\frac{\partial n^{(\tau_3)}({\bf k})}{\partial k^{(2)}}
\cdots\frac{\partial n^{(\tau_{2n+1})}({\bf k})}{\partial k^{(2n)}}
=\nu \left|\S^{2n}\right|.
\end{multline*}
Since the surface area is given by 
$$
\left|\S^{2n}\right|={n!2^{2n+1}\pi^n}/{(2n)!},
$$
the index can be written in terms of the winding number $\nu$ 
as 
$$
{\rm Ind}^{(2n)}(\vartheta_a,P_{\rm F})=(-1)^{n-1}\nu\in\ze.
$$

\section{Chiral Symmetry in Odd Dimensions}
\label{sec:ChiralSymmOddD}

\subsection{Chiral Index}
\label{subsec:ChiralIndex}

Consider the Hamiltonian $H$ of (\ref{H}) on $\ze^d$ with odd dimensions, $d=2n+1$, $n=0,1,2,\ldots$. 
In the present section, we assume that the Hamiltonian $H$ is chiral symmetric with respect to a chiral operator $S$. 
Then, the condition (\ref{chiralcondition}) for chiral symmetry implies that, if a positive energy, $E>0$, is in 
the spectrum of the Hamiltonian $H$, then the corresponding negative energy, $-E$, is also in the spectrum of $H$. 
Therefore, the upper energy bands are mapped to the lower energy bands by the chiral operator $S$.  

As mentioned in Sec.~\ref{sec:ChiralModel}, we deal with the following two cases: 
(i) There is a spectral gap between the upper and lower bands in 
the spectrum of the chiral symmetric Hamiltonian $H$. (ii) When the Fermi level $E_{\rm F}=0$ lies in 
a localization regime, $E_{\rm F}=0$ is not an eigenvalue of the chiral Hamiltonian $H$ with 
probability one with respect to configurations of the random quantities of the system.  
In such a system, the matrix elements of the projection $P_{\rm F}$ onto the Fermi sea 
exponentially decay with large distance as in (\ref{decayPF}). 
 
In both of the situations, we can define the projections onto the upper and lower bands, 
and write, respectively, $P_\pm$ for the two projections. 
In terms of the projections, the unitary operator of (\ref{U}) is written as $U:=P_+-P_-$. 
Clearly, one has $U^\ast=U$ and $U^2=1$. Further, one has 
\begin{equation}
\label{anticommuteSU}
SUS=-U
\end{equation}
by the chiral condition (\ref{chiralcondition}) between the Hamiltonian $H$ and the chiral operator $S$. 

For simplicity, we assume that the Dirac operator $D_a$ of (\ref{defDiracOdd}) 
commutes with the chiral operator $S$, i.e., $[S,D_a]=0$. 
Actually, when the chiral operator $S$ is written in 
a direct sum as $S=\bigoplus_i S_i$ with a period on the lattice $\ze^d$, 
where the operator $S_i$ acts on a local state on the $i$-th unit cell, 
all of the states on some cells can be identified with an internal degree of freedom at a single site 
of a lattice. Therefore, the chiral operator $S$ can be taken to be independent of the lattice site for 
the lattice. This implies $[S,D_a]=0$. 
 
In order to define the chiral index, we recall the expression (\ref{projectionchiral}) of the projection operator,  
$\mathcal{P}_{\rm D}=(1+D_a)/2$. Let $\epsilon>0$. Then, the operator, 
$$
\left[(\mathcal{P}_{\rm D}-U\mathcal{P}_{\rm D}U)^2\right]^\epsilon
=:(\mathcal{P}_{\rm D}-U\mathcal{P}_{\rm D}U)^{2\epsilon},
$$
is well defined because $\mathcal{P}_{\rm D}-U\mathcal{P}_{\rm D}U$ is self adjoint. 
We define the chiral index by 
\begin{equation}
\label{Ind2n+1epsilonchiral}
{\rm Ind}^{(2n+1,\epsilon)}(D_a,S,U):=\frac{1}{2}{\rm Tr}\; 
S(\mathcal{P}_{\rm D}-U\mathcal{P}_{\rm D}U)^{2n+1+2\epsilon}.
\end{equation}

First, we show that the index is well defined, and takes the value of integer in Theorem~\ref{NonCommuThmOdd}. 
Note that one has 
$$
\mathcal{P}_{\rm D}-U\mathcal{P}_{\rm D}U=U[U,\mathcal{P}_{\rm D}]
=\frac{1}{2}U[U,D_a]=-U[P_-,D_a]
$$
from the definitions of $U$ and $\mathcal{P}_{\rm D}$ and $U^2=1$. 
Therefore, one can show that the operator $(\mathcal{P}_{\rm D}-U\mathcal{P}_{\rm D}U)^{2n+1+2\epsilon}$ is 
of trace class for $\epsilon>0$ in the same way as in Appendix~\ref{traceclassA}. 
Similarly to the case of even dimensions, we set 
$$
A=S(\mathcal{P}_{\rm D}-U\mathcal{P}_{\rm D}U)
\quad\mbox{and}\quad 
B=S(1-\mathcal{P}_{\rm D}-U\mathcal{P}_{\rm D}U).
$$
Then, one has $[S,A]=0$ {from} (\ref{anticommuteSU}), $S^2=1$ and $[S,\mathcal{P}_{\rm D}]=0$. 
Therefore, one has 
$$
{\rm Ind}^{(2n+1,\epsilon)}(D_a,S,U)=\frac{1}{2}{\rm Tr}\; A^{2n+1+2\epsilon}
$$
{from} the definition (\ref{Ind2n+1epsilonchiral}) of the chiral index. 
Further, since we have $[S,B]=0$, the pair of two operators, $A$ and $B$, 
satisfies the same relations $A^2+B^2=1$ {and} $AB+BA=0$ as in (\ref{A2B21}) and (\ref{anticommuABBA1}). 
Thus, in the same way as in the case of even dimensions, if an eigenvalue $\lambda$ of the operator $A$ 
satisfies $0<|\lambda|<1$, then $\lambda$ and $-\lambda$ come in pairs with the same multiplicity. 
Therefore, these eigenvalues are vanishing for taking the trace. In consequence, we obtain 
\begin{multline*}
{\rm Ind}^{(2n+1,\epsilon)}(D_a,S,U)\\
=\frac{1}{2}
\left\{{\rm dim}\;{\rm ker}\;[S(\mathcal{P}_{\rm D}-U\mathcal{P}_{\rm D}U)-1]
-{\rm dim}\;{\rm ker}\;[S(\mathcal{P}_{\rm D}-U\mathcal{P}_{\rm D}U)+1]\right\}.
\end{multline*}
This is equal to the integer-valued index of the noncommutative index Theorem~\ref{NonCommuThmOdd}. 

Next, let us show that the chiral index takes a value of integer. 
{From} the relation (\ref{anticommuteSU}), one has the expression, 
\begin{equation}
\label{Umatu}
U=\left(\begin{matrix}0 & \mathcal{U}^\ast \cr \mathcal{U} & 0\cr\end{matrix}\right),
\end{equation}
in the basis which diagonalizes the chiral operator $S$ as 
$$
S=\left(\begin{matrix}1 & 0 \cr 0 & -1 \cr \end{matrix}\right).
$$ 
Then, one has 
\begin{equation}
\label{PD-UPDUmat}
\mathcal{P}_{\rm D}-U\mathcal{P}_{\rm D}U=\left(\begin{matrix}
\mathcal{P}_{\rm D}-\mathcal{U}^\ast \mathcal{P}_{\rm D}\mathcal{U} & 0 \cr
0 & \mathcal{P}_{\rm D}-\mathcal{U}\mathcal{P}_{\rm D}\mathcal{U}^\ast \cr
\end{matrix}\right).
\end{equation}
Further, 
\begin{align*}
&{\rm Tr}\; S(\mathcal{P}_{\rm D}-U\mathcal{P}_{\rm D}U)^{2n+1+2\epsilon}\\
&={\rm Tr}\;\left(\begin{matrix}
(\mathcal{P}_{\rm D}-\mathcal{U}^\ast \mathcal{P}_{\rm D}\mathcal{U})^{2n+1+2\epsilon} & 0 \cr
0 & -(\mathcal{P}_{\rm D}-\mathcal{U}\mathcal{P}_{\rm D}\mathcal{U}^\ast)^{2n+1+2\epsilon} \cr
\end{matrix}
\right).
\end{align*}
Therefore, in the same way as in Section~\ref{NoSymEvenDInd}, we have 
\begin{equation}
\label{IndOddFredT}
{\rm Ind}^{(2n+1,\epsilon)}(D_a,S,U)={\rm dim}\;{\rm ker}\; \mathfrak{T}_\chi
-{\rm dim}\;{\rm ker}\;\mathfrak{T}_\chi^\ast,
\end{equation}
where the Fredholm operator is given by 
$$
\mathfrak{T}_\chi:=\mathcal{P}_{\rm D}\mathcal{U}\mathcal{P}_{\rm D}+1-\mathcal{P}_{\rm D}.
$$
Thus, the chiral index is independent of the parameter $\epsilon$. In the following, 
we write ${\rm Ind}^{(2n+1)}(D_a,S,U)$ for the chiral index by dropping  
the superscript $\epsilon$. 

The chiral index ${\rm Ind}^{(2n+1)}(D_a,S,U)$ is independent of the location of $a$ of 
the Dirac operator $D_a$, too. In order to show this fact, we write 
$$
\mathfrak{T}_\chi'=\mathcal{P}_{\rm D}'\mathcal{U}\mathcal{P}_{\rm D}'+(1-\mathcal{P}_{\rm D}')
$$
with the projection, 
$$
\mathcal{P}_{\rm D}'=\frac{1}{2}(1+D_{a'}). 
$$
Since $\mathfrak{T}_\chi'=\mathfrak{T}_\chi+(\mathfrak{T}_\chi'-\mathfrak{T}_\chi)$, 
it is sufficient to show that $\mathfrak{T}_\chi'-\mathfrak{T}_\chi$ is compact. 
Note that 
\begin{align*}
\mathfrak{T}_\chi'-\mathfrak{T}_\chi&=\mathcal{P}_{\rm D}'\mathcal{U}\mathcal{P}_{\rm D}'+(1-\mathcal{P}_{\rm D}')
-\mathcal{P}_{\rm D}\mathcal{U}\mathcal{P}_{\rm D}-(1-\mathcal{P}_{\rm D})\\
&=\mathcal{P}_{\rm D}'\mathcal{U}\mathcal{P}_{\rm D}'-\mathcal{P}_{\rm D}\mathcal{U}\mathcal{P}_{\rm D}'
+\mathcal{P}_{\rm D}\mathcal{U}\mathcal{P}_{\rm D}'-\mathcal{P}_{\rm D}\mathcal{U}\mathcal{P}_{\rm D}
-(\mathcal{P}_{\rm D}'-\mathcal{P}_{\rm D})\\
&=(\mathcal{P}_{\rm D}'-\mathcal{P}_{\rm D})\mathcal{U}\mathcal{P}_{\rm D}'+\mathcal{P}_{\rm D}\mathcal{U}
(\mathcal{P}_{\rm D}'-\mathcal{P}_{\rm D})
-(\mathcal{P}_{\rm D}'-\mathcal{P}_{\rm D}).
\end{align*}
Therefore, it is enough to show that $\mathcal{P}_{\rm D}'-\mathcal{P}_{\rm D}$ is compact. 
By definition, one has 
$$
\mathcal{P}_{\rm D}'-\mathcal{P}_{\rm D}=\frac{1}{2}(D_{a'}-D_a)\sim {\rm Const.}\frac{1}{|x-a'|}
\times\{\mbox{gamma matrix}\}
$$
for a large $|x-a'|$. This implies that $(\mathcal{P}_{\rm D}'-\mathcal{P}_{\rm D})^{2n+2}$ is trace class. 
Consequently, $\mathcal{P}_{\rm D}'-\mathcal{P}_{\rm D}$ is compact.

\subsection{Topological Invariant}
\label{Sec:ChiralInv}

In order to obtain the expression of the chiral index in terms of the topological invariant, 
we introduce an approximate index as 
$$
{\rm Ind}^{(2n+1,\epsilon)}(D_a,S,U;R):=\frac{1}{2}
{\rm Tr}\; S(\mathcal{P}_{\rm D}-U\mathcal{P}_{\rm D}U)^{2n+1+2\epsilon}\chi_R^a,
$$
where the cutoff function $\chi_R^a$ is given by (\ref{chiRa}). Clearly, one has 
$$
{\rm Ind}^{(2n+1)}(D_a,S,U)=\lim_{R\nearrow\infty} {\rm Ind}^{(2n+1,\epsilon)}(D_a,S,U;R). 
$$

Note that 
$$
\mathcal{P}_{\rm D}-U\mathcal{P}_{\rm D}U=U[U,\mathcal{P}_{\rm D}]=-[U,\mathcal{P}_{\rm D}]U
$$
and 
$$
[U,\mathcal{P}_{\rm D}]=\frac{1}{2}[U,D_a]=[P_+,D_a].
$$
Combining these with $U^2=1$, one has 
\begin{align*}
(\mathcal{P}_{\rm D}-U\mathcal{P}_{\rm D}U)^2&=-[U,\mathcal{P}_{\rm D}][U,\mathcal{P}_{\rm D}]\\
&=(i[U,\mathcal{P}_{\rm D}])^2=(i[P_+,D_a])^2.
\end{align*}
By using these relations, we obtain 
\begin{align}
{\rm Ind}^{(2n+1,\epsilon)}(D_a,S,U;R)&=\frac{1}{2}{\rm Tr}\; 
S(\mathcal{P}_{\rm D}-U\mathcal{P}_{\rm D}U)(\mathcal{P}_{\rm D}-U\mathcal{P}_{\rm D}U)^{2n+2\epsilon}\chi_R^a\nonumber\\  
&=\frac{1}{2}{\rm Tr}\; SU[P_+,D_a](i[P_+,D_a])^{2n+2\epsilon}\chi_R^a. 
\label{ApprochiralIndP+Da}
\end{align}

Further, we introduce 
$$
I_\pm(\epsilon,R):=\frac{1}{2}{\rm Tr}\; SD_a(P_\pm -D_aP_\pm D_a)^{2n+2+2\epsilon}\chi_R^a.
$$
Since the operator $(P_\pm -D_aP_\pm D_a)^{2n+2}$ is trace class, this is very tractable. 
Actually, the following relations are valid:  
$$
\lim_{R\nearrow\infty}I_\pm(\epsilon,R)=\frac{1}{2}{\rm Tr}\; SD_a(P_\pm -D_aP_\pm D_a)^{2n+2+2\epsilon}
$$
and 
\begin{equation}
\label{doubleLimtIpm}
\lim_{R\nearrow\infty}\lim_{\epsilon\downarrow 0}I_\pm(\epsilon,R)
=\lim_{\epsilon\downarrow 0}\lim_{R\nearrow\infty}I_\pm(\epsilon,R)
=\frac{1}{2}{\rm Tr}\; SD_a(P_\pm -D_aP_\pm D_a)^{2n+2}.
\end{equation}
Further, the identity $P_+=1-P_-$ implies 
\begin{equation}
\label{IDI+I-}
I_+(\epsilon,R)=I_-(\epsilon,R). 
\end{equation}

Let us show 
\begin{equation}
\label{IndchiralRIpmR}
{\rm Ind}^{(2n+1,\epsilon)}(D_a,S,U;R)=I_\pm(\epsilon,R).
\end{equation}
Since $\chi_R^a$ is trace class, we have 
\begin{align*}
I_\pm(\epsilon,R)&=\frac{1}{2}{\rm Tr}\; SD_a(P_\pm-D_aP_\pm D_a)(P_\pm-D_aP_\pm D_a)^{2n+1+2\epsilon}\chi_R^a\\ 
&=\frac{1}{2}{\rm Tr}\; SD_aP_\pm(P_\pm-D_aP_\pm D_a)^{2n+1+2\epsilon}\chi_R^a\\ 
&-\frac{1}{2}{\rm Tr}\; SP_\pm D_a(P_\pm-D_aP_\pm D_a)^{2n+1+2\epsilon}\chi_R^a,
\end{align*}
where we have used $D_a^2=1$. The first term in the right-hand side is written 
\begin{multline*}
{\rm Tr}\; SD_aP_\pm(P_\pm-D_aP_\pm D_a)^{2n+1+2\epsilon}\chi_R^a=\\
-{\rm Tr}\; SP_\pm D_a(P_\pm-D_aP_\pm D_a)^{2n+1+2\epsilon}\chi_R^a,
\end{multline*}
where we have used $[S,D_a]=0$, the property of the trace, and  
$$
D_a(P_\pm-D_aP_\pm D_a)=-(P_\pm -D_aP_\pm D_a)D_a.
$$
Therefore, one has 
$$
I_\pm(\epsilon,R)=
-{\rm Tr}\; SP_\pm D_a(P_\pm-D_aP_\pm D_a)^{2n+1+2\epsilon}\chi_R^a. 
$$
In the same way as the above, one has 
$$
(P_\pm-D_aP_\pm D_a)^2=(i[P_\pm,D_a])^2
$$
and 
$$
D_a(P_\pm-D_aP_\pm D_a)=-[P_\pm,D_a]. 
$$
{From} these observations, we obtain 
$$
I_\pm(\epsilon,R)={\rm Tr}\; SP_\pm [P_\pm,D_a](i[P_\pm,D_a])^{2n+2\epsilon}\chi_R^a.
$$
{From} this, (\ref{ApprochiralIndP+Da}), (\ref{IDI+I-}) and $P_-=1-P_+$, we obtain 
\begin{align}
I_+(\epsilon,R)&=\frac{1}{2}\left[I_+(\epsilon,R)+I_-(\epsilon,R)\right]\nonumber
\\ \nonumber
&=\frac{1}{2}{\rm Tr}\; SP_+ [P_+,D_a](i[P_+,D_a])^{2n+2\epsilon}\chi_R^a\\ \nonumber 
&+\frac{1}{2}{\rm Tr}\; SP_- [P_-,D_a](i[P_-,D_a])^{2n+2\epsilon}\chi_R^a\\ \nonumber
&=\frac{1}{2}{\rm Tr}\; SP_+ [P_+,D_a](i[P_+,D_a])^{2n+2\epsilon}\chi_R^a\\ \nonumber
&-\frac{1}{2}{\rm Tr}\; SP_- [P_+,D_a](i[P_+,D_a])^{2n+2\epsilon}\chi_R^a\\ \nonumber
&=\frac{1}{2}{\rm Tr}\; S(P_+-P_-) [P_+,D_a](i[P_+,D_a])^{2n+2\epsilon}\chi_R^a\\ 
&={\rm Ind}^{(2n+1,\epsilon)}(D_a,S,U;R).
\label{relationI+ChiralInd}
\end{align}
This is the desired result (\ref{IndchiralRIpmR}) because of (\ref{IDI+I-}).

We write 
$$
I_\pm(R):=\lim_{\epsilon\downarrow 0}I_\pm(\epsilon,R)=\frac{1}{2}{\rm Tr}\;
SD_a(P_\pm-D_aP_\pm D_a)^{2n+2}\chi_R^a.
$$
We introduce two approximate indices as 
$$
{\rm Ind}^{(2n+1)}(D_a,S,U;\Omega,R):=\frac{1}{|\Omega|}\int_\Omega dv(a)\;I_\pm(R)
$$
and 
\begin{multline}
\label{ApproTildeIndChiral}
\widetilde{{\rm Ind}}^{(2n+1)}(D_a,S,U;\Lambda,R)\\:=
\frac{1}{2|\Lambda|}\int_{\re^d}dv(a)\; {\rm Tr}\; SD_a(P_\pm -D_aP_\pm D_a)^{2n+2}\chi_R^a\chi_\Lambda. 
\end{multline} 
Since the chiral index is independent of the local $a$, one has 
\begin{align*}
\lim_{R\nearrow\infty}
{\rm Ind}^{(2n+1)}(D_a,S,U;\Omega,R)&=\frac{1}{|\Omega|}\int_\Omega dv(a)\;
\lim_{R\nearrow\infty}\lim_{\epsilon\downarrow 0}I_\pm(\epsilon,R)\\
&=\frac{1}{|\Omega|}\int_\Omega dv(a)\;\lim_{\epsilon\downarrow 0}
\lim_{R\nearrow\infty}I_\pm(\epsilon,R)\\
&=\lim_{\epsilon\downarrow 0}\frac{1}{|\Omega|}\int_\Omega dv(a)\;{\rm Ind}^{(2n+1,\epsilon)}(D_a,S,U)\\
&={\rm Ind}^{(2n+1)}(D_a,S,U), 
\end{align*}
where we have used (\ref{doubleLimtIpm}) and (\ref{IndchiralRIpmR}). 
Clearly, 
$$
\lim_{\Omega\nearrow\re^d}\lim_{R\nearrow\infty}
{\rm Ind}^{(2n+1)}(D_a,S,U,\Omega,R)={\rm Ind}^{(2n+1)}(D_a,S,U).
$$

The following lemma is an analogue of Lemma~\ref{lem:DlimtROmegaInd}: 

\begin{lem}
We have 
\label{doubleLimChiralInd}
$$
\lim_{R\nearrow\infty}\lim_{\Omega\nearrow\re^d}{\rm Ind}^{(2n+1)}(D_a,S,U,\Omega,R)=
{\rm Ind}^{(2n+1)}(D_a,S,U).
$$
\end{lem}

\begin{proof}
Note that 
\begin{align*}
\frac{1}{|\Omega|}\int_\Omega dv(a)\;I_\pm(R)&=
\frac{1}{2|\Omega|}\int_\Omega dv(a)\;{\rm Tr}\;SD_a(P_\pm-D_aP_\pm D_a)^{2n+2}\chi_R^a\\
&=\frac{1}{2|\Omega|}\int_\Omega dv(a)\;{\rm Tr}\;SD_a(P_\pm-D_aP_\pm D_a)^{2n+2}\\
&+\frac{1}{2|\Omega|}\int_\Omega dv(a)\;{\rm Tr}\;SD_a(P_\pm-D_aP_\pm D_a)^{2n+2}(\chi_R^a-1).
\end{align*}
The first term in the right-hand side is written as 
\begin{align*}
&\frac{1}{2|\Omega|}\int_\Omega dv(a)\;{\rm Tr}\;SD_a(P_\pm-D_aP_\pm D_a)^{2n+2}\\
&=\frac{1}{2|\Omega|}\int_\Omega dv(a)\;{\rm Tr}\;\lim_{\epsilon\downarrow 0}\lim_{R\nearrow \infty}I_\pm(\epsilon,R)\\
&=\lim_{\epsilon\downarrow 0}\frac{1}{2|\Omega|}\int_\Omega dv(a)\;{\rm Ind}^{(2n+1,\epsilon)}(D_a,S,U)
={\rm Ind}^{(2n+1)}(D_a,S,U),
\end{align*}
where we have used (\ref{doubleLimtIpm}) and (\ref{IndchiralRIpmR}), and the fact that 
the chiral index is independent of the local $a$. The absolute value of the integrand in 
the second term can be bounded by ${\rm Const.}/R$ uniformly with respect to $a$ in the same way as 
in the proof of Lemma~\ref{lem:TrAnRbound} in Appendix~\ref{Appendix:DlimtROmegaInd}. 
Thus, the second term is vanishing in the double limit, and the desired relation is obtained. 
\end{proof}

Further, we can obtain an analogue of Lemma~\ref{lem:IndtildeInd} in the same way as follows:

\begin{lem}
\label{ChiralIndLimOmegaLambda}
For any fixed $R$, the following two limits coincide with each other as 
$$
\lim_{\Omega\nearrow\re^d}{\rm Ind}^{(2n+1)}(D_a,S,U,\Omega,R)=
\lim_{\Lambda\nearrow\ze^d}\widetilde{{\rm Ind}}^{(2n+1)}(D_a,S,U,\Lambda,R).
$$
in the sense of the accumulation points. 
\end{lem}

By replacing the cutoff function $\chi_R^a$ with $\chi_R^a\chi_\Lambda$ 
in the derivation of (\ref{relationI+ChiralInd}), a similar relation is obtained as  
$$
{\rm Tr}\; SD_a(P_\pm-D_aP_\pm D_a)^{2n+2}\chi_R^a\chi_\Lambda
=(-1)^{n-1}{\rm Tr}\; SU([P_-,D_a])^{2n+1}\chi_R^a\chi_\Lambda.
$$
By using this relation, the approximate index of (\ref{ApproTildeIndChiral}) can be written 
\begin{multline}
\label{tildeIndDaSULambdaR}
\widetilde{{\rm Ind}}^{(2n+1)}(D_a,S,U;\Lambda,R)\\
=\frac{(-1)^{n-1}}{2|\Lambda|}\int_{\re^d}dv(a)\; {\rm Tr}\; SU([P_-,D_a])^{2n+1}\chi_R^a\chi_\Lambda. 
\end{multline}
Here, we stress that, since the operator $([P_-,D_a])^{2n+1}$ in the right-hand side may not be trace class, 
we need the cutoff function $\chi_R^a$. In \cite{PSB}, the corresponding cutoff function was not 
introduced in their proof of Theorem~4.3 in \cite{PSB}. 

By using the basis (\ref{ONSzeta}) and (\ref{commuPD}) in Appendix~\ref{traceclassA}, 
the above approximate index (\ref{tildeIndDaSULambdaR}) can be written as  
\begin{align}
&\widetilde{{\rm Ind}}^{(2n+1)}(D_a,S,U;\Lambda,R)\nonumber\\ \nonumber 
&=
\frac{(-1)^{n}}{2|\Lambda|}\int_{\re^d}dv(a)\;\sum_{u_1,u_2,\ldots,u_{2n+1},u_{2n+2}}
\sum_{\alpha_1,\alpha_2,\ldots,\alpha_{2n+1},\alpha_{2n+2}}
\langle \zeta_{u_{2n+2}}^{\alpha_{2n+2}},SU\zeta_{u_1}^{\alpha_1}\rangle\\ 
&\times\langle\zeta_{u_1}^{\alpha_1},P_-\zeta_{u_2}^{\alpha_2}\rangle
\langle\zeta_{u_2}^{\alpha_2},P_-\zeta_{u_3}^{\alpha_3}\rangle
\cdots
\langle \zeta_{u_{2n+1}}^{\alpha_{2n+1}}P_-\zeta_{u_{2n+2}}^{\alpha_{2n+2}}\rangle
\chi_R^a(u_{2n+2})\chi_\Lambda(u_{2n+2})\mathcal{T}_\chi^{(\gamma)},
\label{ChiralIndeONS}
\end{align}
where 
\begin{align*}
\mathcal{T}_\chi^{(\gamma)}&:={\rm tr}^{(\gamma)}\left(\frac{u_1-a}{|u_1-a|}-\frac{u_2-a}{|u_2-a|}\right)\cdot\gamma
\left(\frac{u_2-a}{|u_2-a|}-\frac{u_3-a}{|u_3-a|}\right)\cdot\gamma\\
&\times \left(\frac{u_{2n}-a}{|u_{2n}-a|}-\frac{u_{2n+1}-a}{|u_{2n+1}-a|}\right)\cdot\gamma
\left(\frac{u_{2n+1}-a}{|u_{2n+1}-a|}-\frac{u_{2n+2}-a}{|u_{2n+2}-a|}\right)\cdot\gamma.
\end{align*}
{From} the properties of the gamma matrices and the same argument as in Section~\ref{TopoInvEven}, one has 
\begin{align*} 
\mathcal{T}_\chi^{(\gamma)}&=\sum_{j=1}^{2n+2}(-1)^j{\rm tr}^{(\gamma)}
\frac{u_1-a}{|u_1-a|}\cdot\gamma\cdots\underline{\frac{u_j-a}{|u_j-a|}\cdot\gamma}
\cdots\frac{u_{2n+2}-a}{|u_{2n+2}-a|}\cdot\gamma \\
&=\sum_{j=1}^{2n+2}(-1)^{j-1}(-2i)^n\sum_\sigma (-1)^\sigma 
\frac{u_1^{(\sigma_1)}-a^{(\sigma_1)}}{|u_1-a|}\cdots\underline{\frac{u_j-a}{|u_j-a|}}\\ 
&\qquad\qquad\qquad\qquad\qquad\qquad \cdots\frac{u_{2n+2}^{(\sigma_{2n+1})}-a^{(\sigma_{2n+1})}}{|u_{2n+2}-a|}\\
&=(-2i)^n(2n+1)!\sum_{j=1}^{2n+2}{\rm Vol}[\mathfrak{S}_j(a)],
\end{align*}
where 
$$
\mathfrak{S}_j(a)=\left[a+\frac{u_1-a}{|u_1-a|},\cdots, a , \cdots, a+\frac{u_{2n+2}-a}{|u_{2n+2}-a|}\right].
$$

In the same way as in the calculations in the case of the even dimensions, one has 
$$
\int_{\re^d}dv(a)\chi_R^a(u_{2n+2})\mathcal{T}_\chi^{(\gamma)}
\sim (-2i)^n(2n+1)!{\rm Vol}[\mathfrak{B}^{(2n+1)}]{\rm Vol}[\mathfrak{S}] 
$$
for a large $R$ in the expression (\ref{ChiralIndeONS}) of the chiral index. Here, $\mathfrak{B}^{(2n+1)}$ is 
the $(2n+1)$-dimensional unit ball whose volume is given by  
$$
{\rm Vol}[\mathfrak{B}^{(2n+1)}]=\frac{2^{n+1}\pi^n}{(2n+1)!!}
$$
and $\mathfrak{S}$ is the simplex $[u_1,u_2,\cdots,u_{2n+2}]$ whose volume is given by 
$$
{\rm Vol}[\mathfrak{S}]=\frac{1}{(2n+1)!}{\rm det}[u_1-u_{2n+2},u_2-u_{2n+2},\cdots,u_{2n+1}-u_{2n+2}].
$$
 
We write 
$$
\widetilde{{\rm Ind}}^{(2n+1)}(D_a,S,U;\Lambda)
:=\lim_{R\nearrow\infty}\widetilde{{\rm Ind}}^{(2n+1)}(D_a,S,U;\Lambda,R). 
$$
{From} the these observations, the chiral index (\ref{ChiralIndeONS}) can be written as 
\begin{align}
&\widetilde{{\rm Ind}}^{(2n+1)}(D_a,S,U;\Lambda)\nonumber
\\ \nonumber 
&=\frac{2^{2n}(\pi i)^n}{(2n+1)!!|\Lambda|}
\sum_{u_1,u_2,\ldots,u_{2n+1},u_{2n+2}}
\sum_{\alpha_1,\alpha_2,\ldots,\alpha_{2n+1},\alpha_{2n+2}}
\langle \zeta_{u_{2n+2}}^{\alpha_{2n+2}},SU\zeta_{u_1}^{\alpha_1}\rangle\\ \nonumber
&\times\langle\zeta_{u_1}^{\alpha_1},P_-\zeta_{u_2}^{\alpha_2}\rangle
\langle\zeta_{u_2}^{\alpha_2},P_-\zeta_{u_3}^{\alpha_3}\rangle
\cdots
\langle \zeta_{u_{2n+1}}^{\alpha_{2n+1}}P_-\zeta_{u_{2n+2}}^{\alpha_{2n+2}}\rangle
\chi_\Lambda(u_{2n+2})\\ 
&\times {\rm det}[u_1-u_{2n+2},u_2-u_{2n+2},\cdots,u_{2n+1}-u_{2n+2}].
\label{ChiralIndeONS2}
\end{align}

For the same reason as in the case of the even dimensions, we have 
\begin{equation}
\label{tildeChiralIndDiff}
\left|\widetilde{{\rm Ind}}^{(2n+1)}(D_a,S,U;\Lambda,R)
-\widetilde{{\rm Ind}}^{(2n+1)}(D_a,S,U;\Lambda)\right|\le \frac{{\rm Const.}}{R}
\end{equation}
for a large $R$, where the constant in the right-hand side is independent of the lattice $\Lambda$. 

\begin{lem}
The following equalities are valid: 
\begin{align}
{\rm Ind}^{(2n+1)}(D_a,S,U)&=\lim_{R\nearrow\infty}\lim_{\Lambda\nearrow\ze^d}
\widetilde{{\rm Ind}}^{(2n+1)}(D_a,S,U;\Lambda,R)\nonumber\\ 
&=\lim_{\Lambda\nearrow\ze^d}\lim_{R\nearrow\infty}
\widetilde{{\rm Ind}}^{(2n+1)}(D_a,S,U;\Lambda,R).
\label{exchangeLimChiralInd} 
\end{align}
\end{lem}

\begin{proof}
The first equality follows from Lemmas~\ref{doubleLimChiralInd} and \ref{ChiralIndLimOmegaLambda}. 

Let us prove the second equality. Let $\varepsilon$ be a small positive number. 
{From} the proof of Lemma~\ref{doubleLimChiralInd} and the above inequality (\ref{tildeChiralIndDiff}), 
there exists a large $R$ such that  
$$
\left|{\rm Ind}^{(2n+1)}(D_a,S,U;\Omega,R)-{\rm Ind}^{(2n+1)}(D_a,S,U)
\right|\le{{\rm Const.}}/{R}\le \varepsilon/3
$$
for any large $\Omega$, and that 
\begin{equation}
\left|\widetilde{{\rm Ind}}^{(2n+1)}(D_a,S,U;\Lambda,R)
-\widetilde{{\rm Ind}}^{(2n+1)}(D_a,S,U;\Lambda)\right|\le {{\rm Const.}}/{R}\le \varepsilon/3
\end{equation}
for any large $\Lambda$. For a fixed $R$ satisfying these conditions, we can find a large $\Omega$ and 
a large $\Lambda$ such that 
$$
\left|{\rm Ind}^{(2n+1)}(D_a,S,U;\Omega,R)-\widetilde{{\rm Ind}}^{(2n+1)}(D_a,S,U;\Lambda,R)\right|
\le \varepsilon/3
$$
from Lemma~\ref{ChiralIndLimOmegaLambda}. Combining these three inequalities, one has 
\begin{align*}
&\left|{\rm Ind}^{(2n+1)}(D_a,S,U)-\widetilde{{\rm Ind}}^{(2n+1)}(D_a,S,U;\Lambda)\right|\\
&\le\left|{\rm Ind}^{(2n+1)}(D_a,S,U)-{\rm Ind}^{(2n+1)}(D_a,S,U;\Omega,R)\right|\\
&+\left|{\rm Ind}^{(2n+1)}(D_a,S,U;\Omega,R)-\widetilde{{\rm Ind}}^{(2n+1)}(D_a,S,U;\Lambda,R)\right|\\
&+\left|\widetilde{{\rm Ind}}^{(2n+1)}(D_a,S,U;\Lambda,R)-\widetilde{{\rm Ind}}^{(2n+1)}(D_a,S,U;\Lambda)\right|
\le \varepsilon.  
\end{align*}
This implies that the second equality holds. 
\end{proof}

By definition, the second equality of (\ref{exchangeLimChiralInd}) implies 
$$
{\rm Ind}^{(2n+1)}(D_a,S,U)=\lim_{\Lambda\nearrow\ze^d}
\widetilde{{\rm Ind}}^{(2n+1)}(D_a,S,U,\Lambda).
$$ 
Thus, it is sufficient to calculate the right-hand side of (\ref{ChiralIndeONS2}) for deriving the expression of 
the chiral index in terms of the topological invariant, i.e., a winding number. 
The determinant in the right-hand side of (\ref{ChiralIndeONS2}) can be computed as 
\begin{multline*}
{\rm det}[u_1-u_{2n+2},u_2-u_{2n+2},\cdots,u_{2n+1}-u_{2n+2}]\\
=\sum_\sigma (-1)^\sigma (u_1^{(\sigma_1)}-u_2^{(\sigma_1)})(u_2^{(\sigma_2)}-u_3^{(\sigma_2)})\cdots
(u_{2n+1}^{(\sigma_{2n+1})}-u_{2n+2}^{(\sigma_{2n+1})})
\end{multline*}
in the same way as in (\ref{Volsimplexdet}). Substituting this into the right-hand side of (\ref{ChiralIndeONS2}), 
we obtain 
\begin{multline*}
\widetilde{{\rm Ind}}^{(2n+1)}(D_a,S,U,\Lambda)\\
=\frac{2^{2n}(\pi i)^n}{(2n+1)!!|\Lambda|}\sum_\sigma {\rm Tr}\; \chi_\Lambda SU[X^{(\sigma_1)},P_-]
\cdots[X^{(\sigma_{2n+1})},P_-],
\end{multline*}
where the position operator $X=(X^{(1)},\ldots,X^{(2n+1)})$ is given by (\ref{X}) 
with replacing the spatial dimension $d=2n$ by $(2n+1)$. 
 
\begin{thm}
\label{thm:ChiralIndtheta}
The chiral index can be expressed as 
\begin{multline}
\label{ChiralInvariant}
{\rm Ind}^{(2n+1)}(D_a,S,U)\\
=\frac{2^{2n}(\pi i)^n}{(2n+1)!!}\sum_\sigma (-1)^\sigma    
{\rm Tr}\; SU[\vartheta_a^{(\sigma_1)},P_-][\vartheta_a^{(\sigma_2)},P_-]\cdots[\vartheta_a^{(\sigma_{2n+1})},P_-].
\end{multline}
\end{thm}

The proof is slightly different from that of Theorem~\ref{thm:IndDaIndtheta}. 
The details are given in Appendix~\ref{proofTheorem:ChiralIndtheta}. 

In the following, we write 
$$
{\rm Ind}^{(2n+1)}(\vartheta_a,S,U)={\rm Ind}^{(2n+1)}(D_a,S,U). 
$$

\begin{rem}
It is instructive to note the following: 
As an example, let us consider a one-dimensional finite system.
Then, all of the operators on the corresponding Hilbert space are a finite dimensional matrix. 
Therefore, the chiral index can be computed as  
\begin{align*}
{\rm Ind}^{(1)}(\vartheta_a,S,U)&={\rm Tr}\; SU[\vartheta_a,P_-]\\
&={\rm Tr}\; S(P_+-P_-)(\vartheta_a P_--P_-\vartheta_a)\\
&={\rm Tr}\; P_-S(P_+-P_-)\vartheta_a+{\rm Tr}\; SP_-\vartheta_a\\
&={\rm Tr}\; SP_+\vartheta_a+{\rm Tr}\; SP_-\vartheta_a={\rm Tr}\; S\vartheta_a=0, 
\end{align*}
where we have used the property of the trace, the property of the chiral operator $S$ and $P_-S=SP_+$. 
Namely, the chiral index is always vanishing for a finite system. In other words, this 
is nothing but a consequence of the well-known fact that the index of a Fredholm operator on a finite-dimensional 
Banach space is vanishing. Thus, the expression (\ref{ChiralInvariant}) of the chiral index is not 
necessarily convenient to numerically compute the value of the index.  
\end{rem}

We can derive a useful expression of the above chiral index. 
Since the projection $P_-$ onto the lower band is written in terms of the contour integral,  
one has 
$$
[\vartheta_a^{(j)},P_-]=\frac{i}{2\pi i}\oint dz \frac{1}{z-H}J_a^{(j)}\frac{1}{z-H}
$$
in the same way as in the case of even dimensions, where $J_a^{(j)}:=i[H,\theta_\ell]$ is the local current operator. 
Substituting this into the expression (\ref{ChiralInvariant}) of the chiral index, we obtain 
\begin{multline*} 
{\rm Ind}^{(2n+1)}(\vartheta_a,S,U)\\=\frac{i^{n}}{(2n+1)!!\cdot 2\cdot \pi^{n+1}}\sum_\sigma(-1)^\sigma 
\oint dz_1\oint dz_2\cdots \oint dz_{2n+1}{\rm Tr}\; 
SU\frac{1}{z_1-H}\\ \times J_a^{(\sigma_1)}\frac{1}{z_1-H}
\frac{1}{z_2-H}J_a^{(\sigma_2)}\frac{1}{z_2-H}\cdots\frac{1}{z_{2n+1}-H}J_a^{(\sigma_{2n+1})}\frac{1}{z_{2n+1}-H}.
\end{multline*}
This right-hand side is nothing but the generalized Chern number in Theorem~\ref{NonCommuThmOdd}. 

In the same way as in the case of even dimensions, 
we want to approximate the above right-hand side by the operators on the finite-volume lattice $\Lambda$.  
We assume that the corresponding finite system with the periodic boundary condition 
has a spectral gap or a localization regime between the upper and lower bands. 
Then, relying on the locality of the current operator $J_a^{(j)}$ and Assumption~\ref{Assumption}, 
the above right-hand side can be approximated by 
the corresponding finite systems on the finite lattice $\Lambda$.  
We write $\varphi_{\pm,j}$ for the energy eigenvector of the corresponding finite-volume Hamiltonian 
with the eigenvalue $E_{\pm,j}$, where the subscript $\pm$ denotes the index for the upper and lower bands, 
respectively.  
Consequently, we obtain the relation between the index and the topological invariant: 
\begin{multline}
\label{currentChernChiral} 
{\rm Ind}^{(2n+1)}(\vartheta_a,S,U)=\frac{2^{2n}i^{3n+1}\pi^n}{(2n+1)!!}\\
\times \lim_{\Lambda\nearrow\ze^{2n+1}}
\sum_\sigma(-1)^\sigma \sum_{j_0,j_1,\ldots,j_{2n+1}}\Biggl\{ \langle\varphi_{-,j_0},S\varphi_{+,j_1}\rangle
\frac{\langle \varphi_{+,j_1},J_a^{(\sigma_1)}\varphi_{-,j_2}\rangle}{E_{-,j_2}-E_{+,j_1}}\\
\times \frac{\langle \varphi_{-,j_2},J_a^{(\sigma_2)}\varphi_{+,j_3}\rangle}{E_{-,j_2}-E_{+,j_3}}
\cdots\frac{\langle \varphi_{-,j_{2n}},J_a^{(\sigma_{2n})}\varphi_{+,j_{2n+1}}\rangle}{E_{-,j_{2n}}-E_{+,j_{2n+1}}}
\frac{\langle \varphi_{+,j_{2n+1}},J_a^{(\sigma_{2n+1})}\varphi_{-,j_0}\rangle}{E_{-,j_0}-E_{+,j_{2n+1}}}\\
-\langle\varphi_{+,j_0},S\varphi_{-,j_1}\rangle
\frac{\langle \varphi_{-,j_1},J_a^{(\sigma_1)}\varphi_{+,j_2}\rangle}{E_{-,j_1}-E_{+,j_2}}
\frac{\langle \varphi_{+,j_2},J_a^{(\sigma_2)}\varphi_{-,j_3}\rangle}{E_{-,j_3}-E_{+,j_2}}\\
\cdots\frac{\langle \varphi_{+,j_{2n}},J_a^{(\sigma_{2n})}\varphi_{-,j_{2n+1}}\rangle}{E_{-,j_{2n+1}}-E_{+,j_{2n}}}
\frac{\langle \varphi_{-,j_{2n+1}},J_a^{(\sigma_{2n+1})}\varphi_{+,j_0}\rangle}{E_{-,j_{2n+1}}-E_{+,j_0}}\Biggr\}. 
\end{multline}
In one dimension $(n=0)$, the right-hand side of the chiral index (\ref{currentChernChiral}) coincides 
with the imaginary part of the linear response coefficient for 
the ${\rm AC}$ conductance except for the prefactor $2$ as we show in Appendix~\ref{sec:LR}.  

\subsection{Translationally Invariant Systems in Odd dimensions}

As an example, let us consider a translationally invariant Hamiltonian $H_\Lambda$ on 
the hypercubic box $\Lambda$ with the side length $L$. 
The energy eigenvectors of the Hamiltonian $H_\Lambda$ can be written as  $\varphi_{\pm,{\bf k}}$ 
with the energy eigenvalue $E_{\pm,{\bf k}}$ 
in terms of the momentum ${\bf k}:=(k^{(1)},k^{(2)},\ldots,k^{(2n+1)})$.   
In the same way as in the case of even dimensions, the chiral index of (\ref{currentChernChiral}) is written as 
\begin{multline*}
{\rm Ind}^{(2n+1)}(\vartheta_a,S,U)=\frac{i^{3n+1}}{(2n+1)!!\cdot 2\cdot \pi^{n+1}}\\
\times 
\sum_\sigma(-1)^\sigma \oint dk^{(1)}\cdots dk^{(2n+1)}
\Biggl\{ \langle\varphi_{-,{\bf k}},S\varphi_{+,{\bf k}}\rangle
\frac{\langle \varphi_{+,{\bf k}},J^{(\sigma_1)}\varphi_{-,{\bf k}}\rangle}{E_{-,{\bf k}}-E_{+,{\bf k}}}\\
\times \frac{\langle \varphi_{-,{\bf k}},J^{(\sigma_2)}\varphi_{+,{\bf k}}\rangle}{E_{-,{\bf k}}-E_{+,{\bf k}}}
\cdots\frac{\langle \varphi_{-,{\bf k}},J^{(\sigma_{2n})}\varphi_{+,{\bf k}}\rangle}{E_{-,{\bf k}}-E_{+,{\bf k}}}
\frac{\langle \varphi_{+,{\bf k}},J^{(\sigma_{2n+1})}\varphi_{-,{\bf k}}\rangle}{E_{-,{\bf k}}-E_{+,{\bf k}}}\\
-\langle\varphi_{+,{\bf k}},S\varphi_{-,{\bf k}}\rangle
\frac{\langle \varphi_{-,{\bf k}},J^{(\sigma_1)}\varphi_{+,{\bf k}}\rangle}{E_{-,{\bf k}}-E_{+,{\bf k}}}
\frac{\langle \varphi_{+,{\bf k}},J^{(\sigma_2)}\varphi_{-,{\bf k}}\rangle}{E_{-,{\bf k}}-E_{+,{\bf k}}}\\
\cdots\frac{\langle \varphi_{+,{\bf k}},J^{(\sigma_{2n})}\varphi_{-,{\bf k}}\rangle}{E_{-,{\bf k}}-E_{+,{\bf k}}}
\frac{\langle \varphi_{-,{\bf k}},J^{(\sigma_{2n+1})}\varphi_{+,{\bf k}}\rangle}{E_{-,{\bf k}}-E_{+,{\bf k}}}\Biggr\}. 
\end{multline*}
Further, this can be written 
\begin{multline}
\label{IndChiralPDel}
{\rm Ind}^{(2n+1)}(\vartheta_a,S,U)=\frac{-i^{3n+1}}{(2n+1)!!\cdot 2\cdot \pi^{n+1}}\\
\times 
\sum_\sigma(-1)^\sigma \oint dk^{(1)}\cdots dk^{(2n+1)}\;
{\rm Tr}\; S(1-2P_-({\bf k}))\frac{\partial P_-({\bf k})}{\partial k^{(\sigma_1)}}\cdots
\frac{\partial P_-({\bf k})}{\partial k^{(\sigma_{2n+1})}}
\end{multline}
in terms of the projection $P_-({\bf k})$ onto the lower band in the momentum space. 

More concretely, let us consider the Hamiltonian $H$ which is given by 
\begin{align*}
(H\varphi)(x)&=\sum_{j=1}^{2n+1} \frac{t_{\rm s}^{(j)}}{2i}
[\varphi(x+e^{(j)})-\varphi(x-e^{(j)})]\gamma^{(j)}\\
&+\Big\{m_0+\sum_{j=1}^{2n+1}\frac{t_{\rm c}^{(j)}}{2}[\varphi(x+e^{(j)})+\varphi(x-e^{(j)})]
\Big\}\gamma^{(2n+2)}, 
\end{align*}
where $t_{\rm s}^{(j)}$, $t_{\rm c}^{(j)}$ and $m_0$ are real constants, $e^{(j)}$ 
are the unit vectors in the $j$-th direction, and $\gamma^{(j)}$ are the gamma matrices. 
Clearly, this Hamiltonian $H$ is chiral symmetric with respect to the chiral operator $S=\gamma^{(2n+3)}$. 
The three-dimensional model is given by the equation (82) of \cite{RSFL}.
More generally, the models of this type are generally called odd dimensional 
Pauli-Dirac theory \cite{So,IM}. As is well known, these exhibit a nontrivial topological 
invariant, winding number \cite{MP}. Recently, these models have been intensively 
investigated again in \cite{Kitaev,SRFL,RSFL,MHSP,PSB,GSB}. 

The Fourier transform of the Hamiltonian $H$ is given by 
$$
H({\bf k})=\gamma\cdot\mathcal{E}({\bf k})=\sum_{i=1}^{2n+2} \gamma^{(i)}\mathcal{E}^{(i)}({\bf k}),
$$
where we have written 
$\mathcal{E}({\bf k}):=(\mathcal{E}^{(1)}({\bf k}),\mathcal{E}^{(2)}({\bf k}),\ldots,\mathcal{E}^{(2n+2)}({\bf k}))$ 
with\hfill\break $\mathcal{E}^{(j)}({\bf k})=t_{\rm s}^{(j)}\sin k^{(j)}$ for $j=1,2,\ldots,2n+1$, 
and 
$$
\mathcal{E}^{(2n+2)}({\bf k})=m_0+\sum_{j=1}^{2n}t_{\rm c}^{(j)}\cos k^{(j)}.
$$
{From} the properties of the gamma matrices, one has 
$$
H({\bf k})^2=|\mathcal{E}({\bf k})|^2 \quad \mbox{with} \quad 
|\mathcal{E}({\bf k})|=\sqrt{\sum_{i=1}^{2n+2} [\mathcal{E}^{(i)}({\bf k})]^2}.
$$
Therefore, if $|\mathcal{E}({\bf k})|$ is nonvanishing for all ${\bf k}$, 
then there appears a nonvanishing spectral gap between the upper and lower bands. 
In the following, we consider such a situation. 

The projection onto the lower band can be written as 
$$
P_-({\bf k})=\frac{1}{2}\Bigl[1-\frac{1}{|\mathcal{E}({\bf k})|}\gamma\cdot\mathcal{E}({\bf k})\Bigr].
$$
We define the unit vector $n({\bf k})$ with the $j$-th component, 
$$
n^{(j)}({\bf k}):=\frac{1}{|\mathcal{E}({\bf k})|}\mathcal{E}^{(j)}({\bf k}), 
$$
for $j=1,2,\ldots,2n+2$. Substituting these into the expression of the chiral index (\ref{IndChiralPDel}) 
and using the properties of the gamma matrices, we obtain 
\begin{multline*}
{\rm Ind}^{(2n+1)}(\vartheta_a,S,U)=\frac{(-1)^{n-1}n!}{2\cdot \pi^{n+1}}\\
\times 
\sum_\tau(-1)^\tau \oint dk^{(1)}\cdots dk^{(2n+1)}\;
n^{(\tau_1)}({\bf k}))\frac{\partial n^{(\tau_2)}({\bf k})}{\partial k^{(1)}}\cdots
\frac{\partial n^{(\tau_{2n+2})}({\bf k})}{\partial k^{(2n+1)}}. 
\end{multline*}
Since the quantity in the right-hand side is written 
$$
\sum_\tau(-1)^\tau \oint dk^{(1)}\cdots dk^{(2n+1)}\;
n^{(\tau_1)}({\bf k}))\frac{\partial n^{(\tau_2)}({\bf k})}{\partial k^{(1)}}\cdots
\frac{\partial n^{(\tau_{2n+2})}({\bf k})}{\partial k^{(2n+1)}}=\nu \left|\S^{2n+1}\right|
$$
in terms of the winding number $\nu$ for the $(2n+1)$-dimensional sphere $\S^{2n+1}$, 
we obtain 
$$
{\rm Ind}^{(2n+1)}(\vartheta_a,S,U)=(-1)^{n-1}\nu\in\ze.
$$
Here, we have used that the surface area of $\S^{2n+1}$ is given by 
$$
\left|\S^{2n+1}\right|={2 \pi^{n+1}}/{n!}.
$$

\section{Periodic Table}
\label{PeriodicTable} 

In this section, we show that all of the indices in Table~\ref{Ptable} can be explained by our noncommutative 
index theorems in Sec.~\ref{Sec:MainResults}.

\subsection{Signatures of Operators}

To begin with, we define the signatures of operators which we will use below. 

We denote a complex conjugation transformation by $K$ defined by  
$K\varphi=\overline{\varphi}$ for the wavefunctions $\varphi\in\ell^2(\ze^d,\co^M)\otimes\co^{2^n}$. 
{From} the definition (\ref{C+}) of the time-reversal operator $C_+$, one has  
\begin{equation}
\label{C+KC+K}
C_+KC_+K\varphi=C_+KC_+\overline{\varphi}
=\begin{cases}
(-1)^\ell\varphi, & n=2\ell;\\
(-1)^{\ell-1}\varphi, & n=2\ell-1
\end{cases}
\end{equation}
for a wavefunction $\varphi$, where $\ell\in\{0,1,2,\ldots\}$. 
We define the signature of the operator $(C_+K)^2$ by 
\begin{equation}
\label{sgnC+K2}
{\rm sgn}(C_+K)^2=\begin{cases}
(-1)^\ell, & n=2\ell;\\
(-1)^{\ell-1}, & n=2\ell-1.
\end{cases}
\end{equation}
Similarly, for the other time-reversal operator $C_-$, 
\begin{equation}
\label{C-KC-K}
C_-KC_-K\varphi=C_-KC_-\overline{\varphi}=(-1)^\ell\varphi
\end{equation}
for $n=2\ell-1$ or $n=2\ell$. Therefore, the signature of $(C_-K)^2$ is defined by 
\begin{equation}
\label{sgnC-k2}
{\rm sgn}(C_-K)^2=(-1)^\ell
\end{equation} 
for $n=2\ell-1$ or $n=2\ell$.    

Further, we define the signature of the gamma matrix $\gamma^{(2n+1)}$. 
For this purpose, consider 
\begin{align}
\label{commCpmKgamma}
C_\pm K \gamma^{(2n+1)}\varphi&=C_\pm\overline{\gamma^{(2n+1)}}\overline{\varphi}\\ \nonumber
&=C_\pm \gamma^{(2n+1)}C_\pm C_\pm \overline{\varphi}\\ \nonumber
&=(-1)^n\gamma^{(2n+1)}C_\pm\overline{\varphi}=(-1)^n\gamma^{(2n+1)}C_\pm K\varphi,
\end{align}
where we have used $\overline{\gamma^{(2n+1)}}=\gamma^{(2n+1)}$ and 
(\ref{Cpmgamma(2n+1)Cpm}) in Appendix~\ref{AppGamma}. 
By relying this relation, we define the signature of $\gamma^{(2n+1)}$ by 
$$
{\rm sgn}(\gamma^{(2n+1)})=(-1)^n. 
$$
Thus, if ${\rm sgn}(\gamma^{(2n+1)})=+1$, then $C_\pm K$ and $\gamma^{(2n+1)}$ commute 
with each other, and for ${\rm sgn}(\gamma^{(2n+1)})=-1$, they anticommute with each other. 

Further, in order to define the signatures of the Dirac operator $D_a$, we note  
\begin{align}
\label{commCpmKDa}
C_+KD_a\varphi&=C_+\overline{D_a}\overline{\varphi}\\ \nonumber
&=C_+\overline{D_a}C_+C_+K\varphi=(-1)^{n+1}D_aC_+K\varphi,
\end{align}
where we have used (\ref{C+gammaC+}). From this, we define 
$$
{\rm sgn}_+(D_a)=(-1)^{n+1}.
$$
Similarly, 
\begin{align}
\label{commCpmKDa-}
C_-KD_a\varphi&=C_-\overline{D_a}\overline{\varphi}\\ \nonumber
&=C_-\overline{D_a}C_-K\varphi=(-1)^nD_aC_-K\varphi,
\end{align}
where we have used (\ref{C-gammaC-}). Therefore, we define 
$$
{\rm sgn}_-(D_a)=(-1)^n.
$$

All these results are summarized in Table~\ref{sgnTRO} in even dimensions, $d\le 8$. 
Clearly, one can find the periodicity with the period $8$, and hence it is enough to deal with the case of $d\le 8$.

\begin{table}
\begin{center}
\begin{tabular}{|c||c|c|c|c|}
\hline
   $\!d=2n\!$  & $2$ &  $4$ & $6$ & $8$ \\
\hline
  $n$ & $1$  & $2$  & $3$ & $4$\\
\hline\hline
  ${\rm sgn}(C_+K)^2$ & $+1$ & $-1$ & $-1$ & $+1$\\
\hline
  ${\rm sgn}(C_-K)^2$  & $-1$ & $-1$ & $+1$ & $+1$ \\
\hline
  ${\rm sgn}(\gamma^{(2n+1)})$ & $-1$ & $+1$ & $-1$ & $+1$ \\
\hline
  ${\rm sgn}_+(D_a)$ & $+1$ & $-1$ & $+1$ & $-1$ \\
\hline
  ${\rm sgn}_-(D_a)$ & $-1$ & $+1$ & $-1$ & $+1$ \\
\hline
\end{tabular}
\bigskip
\caption{\sl The signatures of the operators related to the time-reversal transformations in even dimensions.}
\label{sgnTRO}
\end{center}
\end{table}

\subsection{Even Dimensions}
\label{PTableEvenD}

Consider first the case of even dimensions. 
To begin with, we recall the index of (\ref{Ind(2n)DPA}) as 
\begin{equation}
\label{IndexpA}
{\rm Ind}^{(2n)}(D_a,P_{\rm F})=\frac{1}{2}\;{\rm Tr}\; A^{2n+1}
\end{equation}
which is written in terms of the operator $A=\gamma^{(2n+1)}(P_{\rm F}-D_aP_{\rm F}D_a)$.

Consider first the classes without chiral symmetry. The rest will be treated in Sec.~\ref{sec:SSE} below. 

\subsubsection{AI and AII Classes}

The Hamiltonian $H$ of the models in these classes exhibits time-reversal symmetry only. 
Namely, it satisfies 
$$
\Theta(H\varphi)=H\varphi^\Theta,
$$
where $\Theta$ is the time-reversal transformation given by (\ref{TRS}).  
First, we extend the time-reversal transformation $\Theta$ to that for the operator $A$ as 
$$
\tilde{\Theta}_\pm:=C_\pm \Theta=C_\pm U^\Theta K.
$$
Note that 
\begin{align*}
(\tilde{\Theta}_\pm)^2\varphi&=C_\pm U^\Theta K C_\pm U^\Theta K\varphi\\
&=U^\Theta C_\pm K C_\pm K K U^\Theta K\varphi\\
&={\rm sgn}(C_\pm K)^2 U^\Theta K U^\Theta K \varphi={\rm sgn}(C_\pm K)^2\Theta^2\varphi,
\end{align*}
where we have used (\ref{C+KC+K}), (\ref{sgnC+K2}), (\ref{C-KC-K}) and (\ref{sgnC-k2}). 
Therefore, we have 
\begin{equation}
\label{tildeTheta2}
(\tilde{\Theta}_\pm)^2\varphi={\rm sgn}(C_\pm K)^2\Theta^2\varphi
\end{equation}
for any function $\varphi\in\ell^2(\ze^d,\co^M\otimes\co^{d_\gamma})$, where $d_\gamma$ is the dimension of 
the Hilbert space for the gamma matrices. 

\bigskip\medskip

\noindent
$\bullet$ {\bf AII Class in Two Dimensions.}
In this case, the Hamiltonian $H$ has only odd time-reversal symmetry. 
We write $\Theta^2=-1$ for the oddness for short. Clearly, from (\ref{tildeTheta2}), we have 
\begin{equation}
\label{tildeTheta2AII2D}
(\tilde{\Theta}_\pm)^2=\mp 1
\end{equation} 
in the present case. 

First, we show that the integer-valued index of (\ref{IndexpA}) is vanishing.  
Let $\lambda\ne 0$ be an eigenvalue of the operator $A$, i.e., 
$$
A\varphi =\lambda\varphi
$$ 
for an eigenvector $\varphi\in\ell^2(\ze^d,\co^M\otimes\co^{d_\gamma})$. Then, 
\begin{align*}
\tilde{\Theta}_\pm A\varphi &=\tilde{\Theta}_\pm \gamma^{(2n+1)}(P_{\rm F}-D_aP_{\rm F}D_a)\varphi\\
&=-\gamma^{(2n+1)}\tilde{\Theta}_\pm(P_{\rm F}-D_aP_{\rm F}D_a)\varphi
=-A\tilde{\Theta}_\pm \varphi,
\end{align*}
where we have used (\ref{commCpmKgamma}), (\ref{commCpmKDa}) and 
\begin{equation}
\label{tildeThetaPFcommu}
\tilde{\Theta}_\pm P_{\rm F}\varphi=P_{\rm F}\tilde{\Theta}_\pm\varphi
\end{equation}
derived from the time-reversal symmetry of the Hamiltonian $H$. 
This implies that $\tilde{\Theta}_\pm\varphi$ is an eigenvector with 
the eigenvalue $-\lambda\ne 0$. This map is one to one from (\ref{tildeTheta2AII2D}). 
Combining this observation with the expression (\ref{IndexpA}) of the index, we obtain 
that the index is vanishing, i.e., 
$$
{\rm Ind}^{(2n)}(D_a,P_{\rm F})=0
$$
for AII class in two dimensions. Therefore, one has 
\begin{align}
{\rm dim}\;{\rm ker}\; (P_{\rm F}-\mathfrak{D}_a^\ast P_{\rm F}\mathfrak{D}_a-1)
&={\rm dim}\;{\rm ker}\; (P_{\rm F}-\mathfrak{D}_a^\ast P_{\rm F}\mathfrak{D}_a+1)\nonumber\\ \nonumber
&={\rm dim}\;{\rm ker}\; (P_{\rm F}-\mathfrak{D}_aP_{\rm F}\mathfrak{D}_a^\ast +1)\\ 
&={\rm dim}\;{\rm ker}\; (P_{\rm F}-\mathfrak{D}_aP_{\rm F}\mathfrak{D}_a^\ast -1) 
\label{dimkerRelations}
\end{align}
{from} the argument for the index in Sec.~\ref{NoSymEvenDInd}. 

Although there is no integer-valued index, the class has a non-trivial $\ze_2$-valued index 
which is defined by  
$$
{\rm Ind}_2^{(2n)}(D_a,P_{\rm F}):=\frac{1}{2}{\rm dim}\;{\rm ker}\; (A-1) \ \mbox{modulo \ } 2. 
$$
The right-hand side can be written as 
\begin{align*}
&\frac{1}{2}{\rm dim}\;{\rm ker}\; (A-1)\\
&=\frac{1}{2}\left[{\rm dim}\;{\rm ker}\; (P_{\rm F}-\mathfrak{D}_a^\ast P_{\rm F}\mathfrak{D}_a-1)
+{\rm dim}\;{\rm ker}\; (P_{\rm F}-\mathfrak{D}_aP_{\rm F}\mathfrak{D}_a^\ast +1)\right]\\
&={\rm dim}\;{\rm ker}\; (P_{\rm F}-\mathfrak{D}_a^\ast P_{\rm F}\mathfrak{D}_a-1),
\end{align*} 
where we have used the above relation (\ref{dimkerRelations}). Therefore, the $\ze_2$ index is written as 
$$
{\rm Ind}_2^{(2n)}(D_a,P_{\rm F})=
{\rm dim}\;{\rm ker}\; (P_{\rm F}-\mathfrak{D}_a^\ast P_{\rm F}\mathfrak{D}_a-1)\ \mbox{modulo \ } 2.
$$
In the following, we will show that the $\ze_2$ index is well defined, i.e., 
the even-oddness of the right-hand side is robust against generic perturbations. 

For this purpose, we recall the operator $B=\gamma^{(2n+1)}(1-P_{\rm F}-D_aP_{\rm F}D_a)$, 
and the relations (\ref{A2B21}) and (\ref{anticommuABBA1}).  
As shown in Sec.~\ref{NoSymEvenDInd}, there exists a one-to-one correspondence between the two eigenvectors, 
$\varphi$ and $B\varphi$, of the operator $A$, 
and hence their eigenvalues come in pairs $\pm\lambda$, provided $0<|\lambda|<1$. 

As shown above, the time-reversal transformation $\tilde{\Theta}_\pm$ induces a similar map to the operator $B$.  
However, from (\ref{commCpmKgamma}), the transformation $\tilde{\Theta}_\pm$ anticommutes with $\gamma^{(2n+1)}$ as 
\begin{equation}
\label{anticommutildeThetagamma}
\{\tilde{\Theta}_\pm,\gamma^{(2n+1)}\}\varphi=0
\end{equation}
for any wavefunction $\varphi$. This implies that the transformation $\tilde{\Theta}_\pm$ maps 
an eigenvector of $\gamma^{(2n+1)}$ onto that with the different eigenvalue. 

On the other hand, the operator $B$ commutes with $\gamma^{(2n+1)}$ as 
\begin{equation}
\label{Bgammacommu}
[B,\gamma^{(2n+1)}]=0. 
\end{equation}
Clearly, this implies that $B$ maps an eigenvector of $\gamma^{(2n+1)}$ onto that without 
changing the eigenvalue of $\gamma^{(2n+1)}$. Thus, we want to find a transformation such that 
the transformation shows a similar pairing property to $B$ as shown above 
and commutes with $\gamma^{(2n+1)}$. 

We show that the transformation $\tilde{\Theta}_\pm D_a$ has the desired properties. 
Since the Dirac operator $D_a$ anticommutes with $\gamma^{(2n+1)}$ by definition, 
we have 
\begin{equation}
\label{tildeThetaDagammacommu}
[\tilde{\Theta}_\pm D_a,\gamma^{(2n+1)}]\varphi=0
\end{equation}
for any wavefunction $\varphi$, where we have used the anticommutation relation (\ref{anticommutildeThetagamma}) 
for $\tilde{\Theta}_\pm$ and $\gamma^{(2n+1)}$. 

Next, we show that the operation $\tilde{\Theta}_\pm D_a$ maps an eigenvector $\varphi$ of the operator $A$ 
with eigenvalue $\lambda$ onto that with $-\lambda$. 
By using the same argument as in the above, one has 
\begin{align*}
\tilde{\Theta}_\pm D_a A\varphi&=\tilde{\Theta}_\pm D_a\gamma^{(2n+1)}(P_{\rm F}-D_aP_{\rm F}D_a)\varphi\\
&=-\tilde{\Theta}_\pm\gamma^{(2n+1)}D_a(P_{\rm F}-D_aP_{\rm F}D_a)\varphi\\
&=\tilde{\Theta}_\pm\gamma^{(2n+1)}(P_{\rm F}-D_aP_{\rm F}D_a)D_a\varphi=-A\tilde{\Theta}_\pm D_a\varphi. 
\end{align*}
This implies that, if $\varphi$ be an eigenvector of $A$ with eigenvalue $\lambda$, then 
$\tilde{\Theta}_\pm D_a\varphi$ is an eigenvector of $A$ with eigenvalue $-\lambda$.  
Thus, the transformation $\tilde{\Theta}_\pm D_a$ has the desired properties. 

Further, the two transformations, $\tilde{\Theta}_\pm D_a$ and $B$, commute with each other 
on the corresponding subspace as: 

\begin{lem}
Let $\varphi$ be an eigenvector of $A$ with eigenvalue $\lambda$ satisfying $0<|\lambda|<1$. Then, 
$$ 
\tilde{\Theta}_\pm D_aB\varphi=B\tilde{\Theta}_\pm D_a\varphi.
$$
\end{lem}

\begin{proof}
To begin with, we show that the Dirac operator $D_a$ anticommutes with $B$. This follows from 
explicit calculations as  
\begin{align*}
D_aB&=D_a\gamma^{(2n+1)}(1-P_{\rm F}-D_aP_{\rm F}D_a)\\
&=-\gamma^{(2n+1)}D_a(1-P_{\rm F}-D_aP_{\rm F}D_a)\\
&=-\gamma^{(2n+1)}(1-P_{\rm F}-D_aP_{\rm F}D_a)D_a=-BD_a,
\end{align*}
where we have used $D_a^2=1$ and $\{D_a,\gamma^{(2n+1)}\}=0$ which is derived from their definitions. 
Combining this, (\ref{anticommutildeThetagamma}), (\ref{tildeThetaPFcommu}), (\ref{commCpmKDa}) 
and (\ref{commCpmKDa-}), one has 
\begin{align*}
\tilde{\Theta}_\pm D_aB\varphi&=-\tilde{\Theta}_\pm BD_a\varphi\\
&=-\tilde{\Theta}_\pm\gamma^{(2n+1)}(1-P_{\rm F}-D_aP_{\rm F}D_a)D_a\varphi=B\tilde{\Theta}_\pm D_a\varphi. 
\end{align*}
\end{proof}

Clearly, the two vectors, $\varphi$ and $B\tilde{\Theta}_\pm D_a\varphi$, are eigenvectors of $A$ with 
the same eigenvalue $\lambda$. Therefore, if these two vectors are linearly independent of each other, 
then the corresponding sector which is spanned by the two vectors is a two-dimensional subspace.   
In fact, we can prove the following lemma: 

\begin{lem}
\label{lem:orthoAII}
Let $\varphi$ be an eigenvector of $A$ with eigenvalue $\lambda$ satisfying $0<|\lambda|<1$. Then, 
$$
\langle B\varphi,\tilde{\Theta}_\pm D_a\varphi\rangle=0.
$$
\end{lem}

\begin{proof}
In the same way as in the proof of the preceding lemma, one has  
\begin{align}
\langle \tilde{\Theta}_\pm \tilde{\Theta}_\pm D_a\varphi,\tilde{\Theta}_\pm B\varphi\rangle
&=\mp \langle D_a\varphi,\tilde{\Theta}_\pm B\varphi\rangle \nonumber\\ \nonumber 
&=\pm \langle D_a\varphi,B\tilde{\Theta}_\pm \varphi\rangle \\ 
&=\mp \langle B\varphi,D_a\tilde{\Theta}_\pm \varphi\rangle =-\langle B\varphi,\tilde{\Theta}_\pm D_a\varphi\rangle,
\label{tildeTheta2inner}
\end{align}
where we have used (\ref{tildeTheta2AII2D}) to obtain the first equality, and 
\begin{equation}
\label{tildeThetaDacommusgn}
\tilde{\Theta}_\pm D_a\varphi={\rm sgn}_\pm(D_a)D_a\tilde{\Theta}_\pm \varphi
\end{equation}
with ${\rm sgn}_\pm(D_a)=\pm 1$ for the fourth equality. 
The last relation (\ref{tildeThetaDacommusgn}) can be derived from (\ref{commCpmKDa}) and (\ref{commCpmKDa-}). 

On the other hand, one has 
$$
\langle \tilde{\Theta}_\pm\psi,\tilde{\Theta}_\pm\varphi\rangle=\langle \varphi,\psi\rangle
$$
for wavefunctions, $\psi$ and $\varphi$, {from} the definitions of $\tilde{\Theta}_\pm$. 
By using this relation, the above quantity in the left-hand side can be written as 
$$
\langle \tilde{\Theta}_\pm \tilde{\Theta}_\pm D_a\varphi,\tilde{\Theta}_\pm B\varphi\rangle
=\langle B\varphi, \tilde{\Theta}_\pm D_a\varphi\rangle. 
$$ 
Combining this with the above result, we obtain 
$$
\langle B\varphi, \tilde{\Theta}_\pm D_a\varphi\rangle=0.
$$ 
\end{proof}

In order to prove the existence of the $\ze_2$ index, we recall the expression of $A$ as 
$$
A=\left(\begin{matrix}P_{\rm F}-\mathfrak{D}_a^\ast P_{\rm F}\mathfrak{D}_a & 0\cr 
0 & -(P_{\rm F}-\mathfrak{D}_a P_{\rm F}\mathfrak{D}_a^\ast)\cr
\end{matrix}\right)
$$
in the basis which diagonalizes $\gamma^{(2n+1)}$. We write 
$\mathfrak{A}:=P_{\rm F}-\mathfrak{D}_a^\ast P_{\rm F}\mathfrak{D}_a$ for short. 
Combining the above two lemma with commutativity, (\ref{Bgammacommu}) and (\ref{tildeThetaDagammacommu}),  
we can conclude that the multiplicity of the eigenvalue $\lambda$ of the operator $\mathfrak{A}$ 
must be even when $\lambda$ satisfies $0<|\lambda|<1$. 
Therefore, if the eigenvalues of $A$ change continuously 
under a continuous variation of the parameters of the Hamiltonian $H$, 
then the parity of the multiplicity of the eigenvalue $\lambda=1$ of $\mathfrak{A}$ must be invariant 
under the deformation of the Hamiltonian. 
Thus, we have two possibilities: (i) An even number of eigenvectors of $\mathfrak{A}$ are lifted from the sector spanned 
by the eigenvectors of $\mathfrak{A}$ with the eigenvalue $\lambda=1$; (ii) an even number of eigenvectors of $\mathfrak{A}$ 
with eigenvalue $\lambda \ne 1$ become degenerate with the eigenvectors of $\mathfrak{A}$ with the eigenvalue $\lambda=1$. 
In both cases, it is enough to prove the continuity of the eigenvalues of the operator $A$ 
under deformation of the Hamiltonian \cite{RSB,Koma2}, in order to establish that the $\ze_2$ index is 
robust against generic perturbations. This is proved in Sec.~\ref{HomoArg}. 
Consequently, AII class in two dimensions has a nontrivial $\ze_2$ index.  

\bigskip\medskip

\noindent
$\bullet$ {\bf AI Class in Two Dimensions.}
In this case, the Hamiltonian $H$ has only even time-reversal symmetry. 
We write $\Theta^2=+1$ for the evenness for short. From (\ref{tildeTheta2}), we have 
\begin{equation}
\label{tildeTheta2AI2D}
(\tilde{\Theta}_\pm)^2=\pm 1
\end{equation} 
in the present case. Clearly, the difference between AI and AII classes is the even-oddness of 
the time-reversal transformation $\tilde{\Theta}_\pm$. 
Therefore, in the proof of Lemma~\ref{lem:orthoAII}, the right-hand side in the last line of (\ref{tildeTheta2inner}) 
becomes opposite sign in the case of AI class. This implies 
that the orthogonality between the corresponding two vectors does not generally holds.  
In consequence, AI class in two dimensions has no index. 

\bigskip\medskip

\noindent
$\bullet$ {\bf AII Class in Four Dimensions.}
In this case, we have $\Theta^2=-1$. Combining this with (\ref{tildeTheta2}), we have 
$$
(\tilde{\Theta}_\pm)^2=+1.
$$
{From} (\ref{commCpmKgamma}), one has 
$$
[\tilde{\Theta}_\pm,\gamma^{(2n+1)}]\varphi=0
$$
for any wavefunction $\varphi$. This yields  
\begin{align*}
\tilde{\Theta}_\pm A\varphi&=\tilde{\Theta}_\pm \gamma^{(2n+1)}(P_{\rm F}-D_aP_{\rm F}D_a)\varphi\\
&=\gamma^{(2n+1)}\tilde{\Theta}_\pm (P_{\rm F}-D_aP_{\rm F}D_a)\varphi=A\tilde{\Theta}_\pm\varphi
\end{align*}
in the same way as in the case of AII class in two dimensions. 
Clearly, this implies that, if $\varphi$ is an eigenvector of $A$, then $\tilde{\Theta}_\pm\varphi$ is 
an eigenvector of $A$ with the same eigenvalue. But this result does not give any information about 
the spectrum of $A$. Thus, AII class in four dimensions preserves an integer-valued index.  

\bigskip\medskip

\noindent
$\bullet$ {\bf AI Class in Four Dimensions.}
Since we have $\Theta^2=+1$ in this case, we obtain 
$$
(\tilde{\Theta}_\pm)^2=-1
$$
from (\ref{tildeTheta2}). In the same way as in the case of AII class in four dimensions, one has 
$$
[\tilde{\Theta}_\pm,\gamma^{(2n+1)}]\varphi=0
$$
and 
$$
\tilde{\Theta}_\pm A\varphi=A\tilde{\Theta}_\pm\varphi.
$$
Therefore, if $\varphi$ is an eigenvector of $A$ with eigenvalue $\lambda$, then 
$\tilde{\Theta}_\pm\varphi$ is also an eigenvector of $A$ with the same eigenvalue in the eigenspace 
of $\gamma^{(2n+1)}$.   
Further, we have 
$$
\langle (\tilde{\Theta}_\pm)^2\varphi,\tilde{\Theta}_\pm\varphi\rangle
=-\langle \varphi,\tilde{\Theta}_\pm\varphi\rangle
$$
{from} $(\tilde{\Theta}_\pm)^2=-1$. On the other hand, one has 
$$  
\langle (\tilde{\Theta}_\pm)^2\varphi,\tilde{\Theta}_\pm\varphi\rangle
=\langle \varphi,\tilde{\Theta}_\pm\varphi\rangle. 
$$ 
These imply $\langle \varphi,\tilde{\Theta}_\pm\varphi\rangle=0$. 
This is nothing but an analogue of the Kramers doublet. In consequence, the integer-valued index 
is always even for AI class in four dimensions

\bigskip\medskip

\noindent
$\bullet$ {\bf AII Class in Six Dimensions.}
In the same way, we have 
$$
(\tilde{\Theta}_\pm)^2=\pm 1
$$
from $\Theta^2=-1$, and 
$$
\{\tilde{\Theta}_\pm,\gamma^{(2n+1)}\}\varphi=0
$$
for any wavefunction $\varphi$. In addition, ${\rm sgn}_\pm(D_a)=\pm 1$. 
These conditions are the same as in the case of AI class in two dimensions. 
Therefore, AII class in six dimensions has no index. 
 
\bigskip\medskip

\noindent
$\bullet$ {\bf AI Class in Six Dimensions.}
In this case, we have 
$$
(\tilde{\Theta}_\pm)^2=\mp 1
$$
from $\Theta^2=+1$, and 
$$
\{\tilde{\Theta}_\pm,\gamma^{(2n+1)}\}\varphi=0
$$
for any wavefunction $\varphi$. In addition, ${\rm sgn}_\pm(D_a)=\pm 1$. 
These conditions are the same as in the case of AII class in two dimensions. 
As a result, AI class in six dimensions has a $\ze_2$ index. 

\bigskip\medskip

\noindent
$\bullet$ {\bf AII Class in Eight Dimensions.}
In the same way, we have 
$$
(\tilde{\Theta}_\pm)^2=-1
$$
from $\Theta^2=-1$, and 
$$
[\tilde{\Theta}_\pm,\gamma^{(2n+1)}]\varphi=0
$$
for any wavefunction $\varphi$. These are the same as in the case of AI class in four dimensions. 
Therefore, the integer-valued index is always even, i.e., $2\ze$. 

\bigskip\medskip

\noindent
$\bullet$ {\bf AI Class in Eight Dimensions.}
Similarly, we have 
$$
(\tilde{\Theta}_\pm)^2=+1
$$
from $\Theta^2=+1$, and 
$$
[\tilde{\Theta}_\pm,\gamma^{(2n+1)}]\varphi=0
$$
for any wavefunction $\varphi$. These conditions are the same as in the case of AII in four dimensions. 
Thus, AI class in eight dimensions preserves an integer-valued index. 

\subsubsection{C and D Classes}
The Hamiltonian $H$ of these classes has particle-hole symmetry only. By the particle-hole transformation $\Xi$, 
it satisfies  
$$
\Xi(H\varphi)=-H\Xi\varphi=-H\varphi^\Xi
$$
for any wavefunction $\varphi$. This yields 
\begin{equation}
\label{XiPF}
\Xi P_{\rm F}\varphi=(1-P_{\rm F})\varphi^\Xi,
\end{equation}
where we have chosen the Fermi level $E_{\rm F}=0$.  

To begin with, we extend $\Xi$ to that for the operator $A$ as 
\begin{equation}
\tilde{\Xi}_\pm:=C_\pm \Xi. 
\end{equation}
Then, we have 
\begin{equation}
\label{tildeXi2}
(\tilde{\Xi}_\pm)^2\varphi={\rm sgn}(C_\pm K)^2\Xi^2\varphi
\end{equation}
for any wavefunction $\varphi$, in the same way as in the case of the time-reversal transformation. 
Further, 
$$
\{\tilde{\Xi}_\pm,\gamma^{(2n+1)}\}\varphi=0\quad \mbox{if \ } {\rm sgn}(\gamma^{(2n+1)})=-1 
$$
and 
$$
[\tilde{\Xi}_\pm,\gamma^{(2n+1)}]\varphi=0\quad \mbox{if \ } {\rm sgn}(\gamma^{(2n+1)})=+1 
$$
for any wavefunction $\varphi$. Clearly, similar commutation relations hold for the Dirac operator $D_a$, 
depending on ${\rm sgn}_\pm(D_a)$. 

\bigskip\medskip

\noindent
$\bullet$ {\bf D Class in Two Dimensions.}
The class D is characterized by the even particle-hole transformation $\Xi$ satisfying $\Xi^2=+1$. 
Therefore, one has 
$$
(\tilde{\Xi}_\pm)^2=\pm 1 
$$
in two dimensions, from the above formula (\ref{tildeXi2}). From ${\rm sgn}(\gamma^{(2n+1)})=-1$, 
\begin{equation}
\{\tilde{\Xi}_\pm,\gamma^{(2n+1)}\}\varphi=0
\end{equation}
for any wavefunction $\varphi$. By using this, (\ref{XiPF}) and $D_a^2=1$, one has 
\begin{align*}
\tilde{\Xi}_\pm A\varphi&=\tilde{\Xi}_\pm\gamma^{(2n+1)}(P_{\rm F}-D_aP_{\rm F}D_a)\varphi\\
&=-\gamma^{(2n+1)}\tilde{\Xi}_\pm (P_{\rm F}-D_aP_{\rm F}D_a)\varphi\\
&=\gamma^{(2n+1)}(P_{\rm F}-D_aP_{\rm F}D_a)\tilde{\Xi}_\pm \varphi=A\tilde{\Xi}_\pm \varphi 
\end{align*}
for any wavefunction $\varphi$. This yields 
\begin{align*}
\tilde{\Xi}_\pm D_aA\varphi&=\tilde{\Xi}_\pm D_a\gamma^{(2n+1)}(P_{\rm F}-D_aP_{\rm F}D_a)\varphi\\
&=-\tilde{\Xi}_\pm\gamma^{(2n+1)}D_a(P_{\rm F}-D_aP_{\rm F}D_a)\varphi\\
&=\tilde{\Xi}_\pm\gamma^{(2n+1)}(P_{\rm F}-D_aP_{\rm F}D_a)D_a\varphi=A\tilde{\Xi}_\pm D_a\varphi,
\end{align*}
where we have used $\{D_a,\gamma^{(2n+1)}\}=0$ and $D_a^2=1$. This implies that, 
if $\varphi$ is an eigenvector of $A$, then $\tilde{\Xi}_\pm D_aA\varphi$ is also an eigenvector of $A$ 
with the same eigenvalue. Further, one has 
$$
[\tilde{\Xi}_\pm D_a,\gamma^{(2n+1)}]\varphi=0
$$
for any wavefunction $\varphi$. 
 
In order to check whether the two vectors, $\varphi$ and $\tilde{\Xi}_\pm D_aA\varphi$, are 
independent of each other, consider the inner product, $\langle \varphi,\tilde{\Xi}_\pm D_aA\varphi\rangle$.  
Note that 
\begin{align*}
\langle \tilde{\Xi}_\pm\tilde{\Xi}_\pm D_a\varphi, \tilde{\Xi}_\pm\varphi\rangle
&=\pm \langle D_a\varphi, \tilde{\Xi}_\pm\varphi\rangle\\
&=\pm \langle \varphi, D_a\tilde{\Xi}_\pm\varphi\rangle
=\langle \varphi, \tilde{\Xi}_\pm D_a\varphi\rangle, 
\end{align*}
where we have used $(\tilde{\Xi}_\pm)^2=\pm 1$ and ${\rm sgn}_\pm(D_a)=\pm 1$. 

On the other hand, one has 
$$
\langle \tilde{\Xi}_\pm\tilde{\Xi}_\pm D_a\varphi, \tilde{\Xi}_\pm\varphi\rangle
=\langle \varphi, \tilde{\Xi}_\pm D_a\varphi\rangle. 
$$
These observations imply that one cannot obtain any information about the spectrum of $A$ from 
the symmetry. In consequence, D class in two dimensions preserves an integer-valued index.   

\bigskip\medskip

\noindent
$\bullet$ {\bf C Class in Two Dimensions.}
Since the class C is characterized by the odd particle-hole transformation $\Xi$ satisfying $\Xi^2=-1$, 
one has 
$$
(\tilde{\Xi}_\pm)^2=\mp 1 
$$
in two dimensions, in the same way as in the preceding D class. 
This difference between C and D classes leads to the orthogonality, 
$$
\langle \varphi, \tilde{\Xi}_\pm D_a\varphi\rangle=0,
$$
in the case of C class. Thus, the integer-valued index for C class in two dimensions is always even, i.e., $2\ze$. 

\bigskip\medskip

\noindent
$\bullet$ {\bf D Class in Four Dimensions.}
In this case, we have 
$$
(\tilde{\Xi}_\pm)^2=-1
$$
and 
\begin{equation}
\label{tildeXigammacommuD4}
[\tilde{\Xi}_\pm,\gamma^{(2n+1)}]\varphi=0
\end{equation}
for any wavefunction $\varphi$. By using the second relation and (\ref{XiPF}), one has 
\begin{align*}
\tilde{\Xi}_\pm A\varphi&=\tilde{\Xi}_\pm \gamma^{(2n+1)}(P_{\rm F}-D_aP_{\rm F}D_a)\varphi\\
&=\gamma^{(2n+1)}\tilde{\Xi}_\pm (P_{\rm F}-D_aP_{\rm F}D_a)\varphi=-A\tilde{\Xi}_\pm \varphi. 
\end{align*}
This implies that the integer-valued index is vanishing. 

In order to check whether the present class has a $\ze_2$ index, we recall the expression of the operator,  
$B=\gamma^{(2n+1)}(1-P_{\rm F}-D_aP_{\rm F}D_a)$. 
Then, by using (\ref{tildeXigammacommuD4}) and (\ref{XiPF}), one has 
\begin{equation}
\label{tildeXiBanicommuD4}
\{\tilde{\Xi}_\pm ,B\}\varphi=0
\end{equation}
for any wavefunction $\varphi$. 

Let us check the orthogonality between $\tilde{\Xi}_\pm\varphi$ and $B\varphi$. 
Note that 
$$
\langle \tilde{\Xi}_\pm\tilde{\Xi}_\pm\varphi, \tilde{\Xi}_\pm B\varphi\rangle 
=-\langle\varphi,\tilde{\Xi}_\pm B\varphi\rangle=\langle B\varphi, \tilde{\Xi}_\pm\varphi\rangle, 
$$
where we have used $(\tilde{\Xi}_\pm)^2=-1$ and (\ref{tildeXiBanicommuD4}). 

On the other hand, 
$$
\langle \tilde{\Xi}_\pm\tilde{\Xi}_\pm\varphi, \tilde{\Xi}_\pm B\varphi\rangle 
=\langle B\varphi, \tilde{\Xi}_\pm\varphi\rangle.
$$
Therefore, one cannot obtain any information about the spectrum of $A$ from the symmetry in the present case. 
In consequence, D class in four dimensions has no index. 

\bigskip\medskip

\noindent
$\bullet$ {\bf C Class in Four Dimensions.}
The difference between D and C classes is even-oddness of the extended particle-hole 
transformation $\tilde{\Xi}_\pm$. In the present C class, we have 
$$
(\tilde{\Xi}_\pm)^2=+1 
$$
which is opposite sign to that for D class. Therefore, the argument of the orthogonality 
between the two vectors, $\tilde{\Xi}_\pm\varphi$ and $B\varphi$, yields 
$$
\langle B\varphi, \tilde{\Xi}_\pm\varphi\rangle=0
$$
for the present C class. Thus, C class in four dimensions has a $\ze_2$ index. 

\bigskip\medskip

\noindent
$\bullet$ {\bf D Class in Six Dimensions.}
In this case, we have 
$$
(\tilde{\Xi}_\pm)^2=\mp 1
$$
and 
\begin{equation}
\{\tilde{\Xi}_\pm,\gamma^{(2n+1)}\}\varphi=0
\end{equation}
for any wavefunction $\varphi$. In addition, ${\rm sgn}_\pm(D_a)=\pm 1$. 
These conditions are all the same as in the case of C class in two dimensions. 
Therefore, the integer-valued index for D class in six dimensions is always even, i.e., $2\ze$. 

\bigskip\medskip

\noindent
$\bullet$ {\bf C Class in Six Dimensions.}
Similarly, in this case, we have 
$$
(\tilde{\Xi}_\pm)^2=\pm 1
$$
and 
\begin{equation}
\{\tilde{\Xi}_\pm,\gamma^{(2n+1)}\}\varphi=0
\end{equation}
for any wavefunction $\varphi$. In addition, ${\rm sgn}_\pm(D_a)=\pm 1$.
These conditions are all the same as in the case of D class in two dimensions. 
In consequence, C class in six dimensions preserves an integer-valued index.  

\bigskip\medskip

\noindent
$\bullet$ {\bf D Class in Eight Dimensions.}
In this case, we have 
$$
(\tilde{\Xi}_\pm)^2=+1
$$
and 
\begin{equation}
[\tilde{\Xi}_\pm,\gamma^{(2n+1)}]\varphi=0
\end{equation}
for any wavefunction $\varphi$. These conditions are the same as in the case of C class in four dimensions. 
Thus, D class in eight dimensions has a $\ze_2$ index. 

\bigskip\medskip

\noindent
$\bullet$ {\bf C Class in Eight Dimensions.}
Similarly, in this case, we have 
$$
(\tilde{\Xi}_\pm)^2=-1
$$
and 
\begin{equation}
[\tilde{\Xi}_\pm,\gamma^{(2n+1)}]\varphi=0
\end{equation}
for any wavefunction $\varphi$. These two conditions are the same as in the case of D class in four dimensions. 
Thus, C class in eight dimensions has no index. 

\subsubsection{Chiral Symmetry in Even Dimensions}
\label{sec:SSE}

In even dimensions, the rest of the CAZ classes are AIII, BDI, CI, DIII and CII, which have chiral symmetry. 
To begin with, we recall the properties of the chiral operator $S$. 
The operator $S$ satisfies $S^\ast=S$ and $S^2=1$. The Hamiltonian $H$ is transformed as $SHS=-H$. 
This yields 
\begin{equation}
\label{SPFS}
SP_{\rm F}S=1-P_{\rm F},
\end{equation}
where we have chosen the Fermi level $E_{\rm F}=0$. Combining this, $[S,\gamma^{(2n+1)}]=0$ and $[S,D_a]=0$, 
we have 
\begin{align}
SA\varphi&=S\gamma^{(2n+1)}(P_{\rm F}-D_aP_{\rm F}D_a)\varphi\nonumber
\\ &=\gamma^{(2n+1)}S(P_{\rm F}-D_aP_{\rm F}D_a)\varphi=-AS\varphi.
\label{SAcommuchiralEven}
\end{align}
This implies that the integer-valued index is always vanishing. Thus, it is enough to determine 
whether or not the classes have a $\ze_2$ index.

Consider first AIII class, which has only chiral symmetry. 
{From} the expression of the operator $B$, (\ref{SPFS}) and $D_a^2=1$, one has 
\begin{align*}
SBS&=\gamma^{(2n+1)}[P_{\rm F}-D_a(1-P_{\rm F})D_a]\\
&=\gamma^{(2n+1)}(P_{\rm F}-1+D_aP_{\rm F}D_a)=-B.
\end{align*}
However, we cannot obtain any information about $\langle S\varphi,B\varphi\rangle$, 
where $\varphi$ is an eigenvector of $A$ with eigenvalue $\lambda$ satisfying $0<|\lambda|<1$. 
This implies that AIII class in all even dimensions has no $\ze_2$ index. 
\bigskip\medskip

\noindent
$\bullet$ {\bf BDI Class in Two Dimensions.}
In this case, from $\Theta^2=+1$, one has 
$$
(\tilde{\Theta}_\pm)^2=\pm 1
$$
and, from ${\rm sgn}(\gamma^{(2n+1)})=-1$,   
$$
\{\tilde{\Theta}_\pm,\gamma^{(2n+1)}\}\varphi=0
$$
for any wavefunction $\varphi$. Further, from the even-oddness of the particle-hole symmetry, 
$$
(S\Theta)^2=+1.
$$
This implies 
$$
S\Theta S\Theta\varphi=\varphi
$$
for any wavefunction $\varphi$. Using $S^2=1$, one has 
$$
\Theta S\Theta S\varphi=S\varphi.
$$
Further, from $\Theta^2=+1$, 
$$
S\Theta\varphi=\Theta S\varphi.
$$
Thus, $S$ and $\Theta$ commutes with each other as 
$$
[S,\Theta]\varphi=0.
$$
Therefore, by the definition of $\tilde{\Theta}_\pm$, we have 
$$
[S,\tilde{\Theta}_\pm]\varphi=0.
$$

Note that 
\begin{align*}
\tilde{\Theta}_\pm A\varphi&=\tilde{\Theta}_\pm \gamma^{(2n+1)}(P_{\rm F}-D_aP_{\rm F}D_a)\varphi\\
&=-\gamma^{(2n+1)}\tilde{\Theta}_\pm (P_{\rm F}-D_aP_{\rm F}D_a)\varphi=-A\tilde{\Theta}_\pm \varphi. 
\end{align*}
{From} this and $\{D_a,\gamma^{(2n+1)}\}=0$, one has 
\begin{align}
\tilde{\Theta}_\pm D_aA\varphi&=\tilde{\Theta}_\pm D_a\gamma^{(2n+1)}(P_{\rm F}-D_aP_{\rm F}D_a)\varphi\nonumber\\ \nonumber
&=-\tilde{\Theta}_\pm\gamma^{(2n+1)}D_a(P_{\rm F}-D_aP_{\rm F}D_a)\varphi\\ 
&=\tilde{\Theta}_\pm\gamma^{(2n+1)}(P_{\rm F}-D_aP_{\rm F}D_a)D_a\varphi=-A\tilde{\Theta}_\pm D_a\varphi.
\label{tildeThetaDaCommuBDI2}
\end{align}
Further, 
\begin{equation}
\label{tildeThetaDagammaBDI2commu}
[\tilde{\Theta}_\pm D_a,\gamma^{(2n+1)}]\varphi=0.
\end{equation}
Thus, the two transformations, $S$ and $\tilde{\Theta}_\pm D_a$, commute with $\gamma^{(2n+1)}$, and 
map an eigenvector of $A$ with eigenvalue $\lambda$ onto that with eigenvalue $-\lambda$. 
 
We recall the fact that 
the transformation, $B=\gamma^{(2n+1)}(1-P_{\rm F}-D_aP_{\rm F}D_a)$, commutes with $\gamma^{(2n+1)}$, and 
maps an eigenvector of $A$ with eigenvalue $\lambda$ satisfying $0<|\lambda|<1$ onto that with eigenvalue $-\lambda$.
In order to determine the index of BDI class in two dimensions, we must deal with 
the three vectors, $S\varphi$, $\tilde{\Theta}_\pm D_a\varphi$ and $B\varphi$. 

Consider first the inner product $\langle S\varphi,\tilde{\Theta}_\pm D_a\varphi\rangle$. Note that 
\begin{align*}
\langle \tilde{\Theta}_\pm S\varphi, \tilde{\Theta}_\pm\tilde{\Theta}_\pm D_a\varphi\rangle 
&=\pm \langle \tilde{\Theta}_\pm S\varphi,D_a\varphi\rangle\\
&=\pm \langle S\tilde{\Theta}_\pm \varphi,D_a\varphi\rangle\\
&=\pm \langle D_a\tilde{\Theta}_\pm \varphi,S\varphi\rangle=\langle \tilde{\Theta}_\pm D_a\varphi,S\varphi\rangle,
\end{align*}
where we have used $(\tilde{\Theta}_\pm)^2=\pm 1$, $[S,\tilde{\Theta}_\pm]\varphi=0$ and ${\rm sgn}_\pm(D_a)=\pm 1$.  
On the other hand, the left-hand side can be written as 
$$
\langle \tilde{\Theta}_\pm S\varphi, \tilde{\Theta}_\pm\tilde{\Theta}_\pm D_a\varphi\rangle
=\langle \tilde{\Theta}_\pm D_a\varphi,S\varphi\rangle.
$$
These do not give any information about the spectrum of $A$. 

Next, consider the inner product, $\langle B\varphi, \tilde{\Theta}_\pm D_a\varphi\rangle$. 
Similarly, 
\begin{align*}
\langle \tilde{\Theta}_\pm \tilde{\Theta}_\pm D_a\varphi,\tilde{\Theta}_\pm B\varphi\rangle
&=\pm \langle D_a\varphi,\tilde{\Theta}_\pm B\varphi\rangle\\
&=\mp \langle D_a\varphi,B\tilde{\Theta}_\pm \varphi\rangle\\
&=\pm \langle B\varphi, D_a\tilde{\Theta}_\pm \varphi\rangle=\langle B\varphi, \tilde{\Theta}_\pm D_a\varphi\rangle
\end{align*}
where we have used $(\tilde{\Theta}_\pm)^2=\pm 1$, $\{B,\tilde{\Theta}_\pm\}\varphi=0$, 
$\{D_a,B\}=0$ and \hfill\break ${\rm sgn}_\pm(D_a)=\pm 1$.
On the other hand, the left-hand side can be written as 
$$
\langle \tilde{\Theta}_\pm \tilde{\Theta}_\pm D_a\varphi,\tilde{\Theta}_\pm B\varphi\rangle
=\langle B\varphi, \tilde{\Theta}_\pm D_a\varphi\rangle. 
$$
These do not give any information, again. 

Finally, consider the inner product $\langle S\varphi, B\varphi\rangle$. 
In order to treat this quantity, we want to find an anti-linear transformation $\tilde{\Sigma}:=U^\Sigma K$ 
which satisfies the following three conditions with a unitary operator $U^\Sigma$: 
$$
(\tilde{\Sigma})^2=+1, \quad 
[\tilde{\Sigma},\gamma^{(2n+1)}]\varphi=0,
\quad \mbox{and} \quad 
\tilde{\Sigma}A\varphi=A\tilde{\Sigma}\varphi
$$
for any wavefunction $\varphi$. Here, $K$ is complex conjugation. 
Then, $A$ and $\tilde{\Sigma}$ are simultaneously diagonalized as 
$$
A\varphi=\lambda\varphi\quad \mbox{and}\quad \tilde{\Sigma}\varphi=\varphi.
$$
The proof is given in Appendix~\ref{EvenALO}. 

We introduce a transformation, 
$$
\tilde{\Sigma}_\pm:=SD_a\tilde{\Theta}_\pm, 
$$
which satisfies the above three conditions. Actually, 
\begin{align*}
(\tilde{\Sigma}_\pm)^2&=SD_a\tilde{\Theta}_\pm SD_a\tilde{\Theta}_\pm\\
&=D_a\tilde{\Theta}_\pm D_a\tilde{\Theta}_\pm=\pm D_a^2(\tilde{\Theta}_\pm)^2=1
\end{align*}
where we have used $[S,\tilde{\Theta}_\pm]\varphi=0$, $S^2=1$, ${\rm sgn}_\pm(D_a)=\pm 1$, $D_a^2=1$ and 
$(\tilde{\Theta}_\pm)^2=\pm 1$. 
The second commutation relation follows from (\ref{tildeThetaDagammaBDI2commu}) and $[S,\gamma^{(2n+1)}]=0$. 
The third condition follows from (\ref{SAcommuchiralEven}) and (\ref{tildeThetaDaCommuBDI2}). 

Note that 
$$
\langle \tilde{\Sigma}_\pm S\varphi, \tilde{\Sigma}_\pm B\varphi\rangle=\langle B\varphi, S\varphi\rangle.
$$
The two vectors in the left-hand side can be calculated as follows:  
\begin{align*}
\tilde{\Sigma}_\pm S\varphi&=SD_a\tilde{\Theta}_\pm S\varphi\\
&=SD_aS\tilde{\Theta}_\pm\varphi=S\tilde{\Sigma}_\pm\varphi=S\varphi,
\end{align*}
where we have used $[S,\tilde{\Theta}_\pm]\varphi=0$ and $\tilde{\Sigma}_\pm\varphi=\varphi$, and  
\begin{align*}
\tilde{\Sigma}_\pm B\varphi&=\tilde{\Sigma}_\pm\gamma^{(2n+1)}(1-P_{\rm F}-D_aP_{\rm F}D_a)\varphi\\
&=\gamma^{(2n+1)}SD_a\tilde{\Theta}_\pm(1-P_{\rm F}-D_aP_{\rm F}D_a)\varphi\\
&=\gamma^{(2n+1)}S(1-P_{\rm F}-D_aP_{\rm F}D_a)D_a\tilde{\Theta}_\pm\varphi\\
&=-\gamma^{(2n+1)}(1-P_{\rm F}-D_aP_{\rm F}D_a)\tilde{\Sigma}_\pm\varphi=-B\varphi, 
\end{align*}
where we have used $[\tilde{\Sigma}_\pm,\gamma^{(2n+1)}]\varphi=0$, (\ref{SPFS}) and $\tilde{\Sigma}_\pm\varphi=\varphi$. 
{From} these observations, we have 
$$
\langle \tilde{\Sigma}_\pm S\varphi, \tilde{\Sigma}_\pm B\varphi\rangle
=-\langle S\varphi,B\varphi\rangle=\langle B\varphi, S\varphi\rangle,
$$
where we have used $\{S,B\}=0$. In consequence, we do not have any information about the spectrum of $A$, again. 
We conclude that BDI class in two dimensions has no $\ze_2$ index. 

\bigskip\medskip

\noindent
$\bullet$ {\bf CI Class in Two Dimensions.}
Similarly to the preceding BDI class, one has 
$$
(\tilde{\Theta}_\pm)^2=\pm 1
$$
from $\Theta^2=+1$, and   
$$
\{\tilde{\Theta}_\pm,\gamma^{(2n+1)}\}\varphi=0
$$
for any wavefunction $\varphi$, from ${\rm sgn}(\gamma^{(2n+1)})=-1$. 

But, the even-oddness of the particle-hole symmetry implies  
$$
(S\Theta)^2=-1
$$
which is different from that of the preceding BDI class.  
Combining this, $S^2=1$ and $\Theta^2=+1$, one has 
$$
\{S,\Theta \}\varphi=0
$$
for any wavefunction $\varphi$. 
Therefore, we have 
$$
\{S,\tilde{\Theta}_\pm\}\varphi=0
$$
by the definition of $\tilde{\Theta}_\pm$. This affects the calculation of 
the orthogonality between $S\varphi$ and $\tilde{\Theta}_\pm D_a\varphi$.  
As a result, a similar calculation to that in the preceding BDI class yields 
$$
\langle S\varphi,\tilde{\Theta}_\pm D_a\varphi\rangle=0.
$$
In addition, we have 
$$
S\tilde{\Theta}_\pm D_a\varphi=-\tilde{\Theta}_\pm D_aS\varphi
$$
{from} the above anticommutativity between $S$ and $\tilde{\Theta}_\pm$.  
These results imply that ${\rm dim}\;{\rm ker}\;(P_{\rm F}-\mathfrak{D}_a^\ast P_{\rm F}\mathfrak{D}_a-1)$ 
is even. Therefore, the $\ze_2$ index is vanishing in CI class in two dimensions. 

\bigskip\medskip

\noindent
$\bullet$ {\bf DIII Class in Two Dimensions.}
In this case, one has 
$$
(\tilde{\Theta}_\pm)^2=\mp 1
$$
{from} $\Theta^2=-1$, and 
$$
\{\tilde{\Theta}_\pm,\gamma^{(2n+1)}\}\varphi=0
$$
for any wavefunction $\varphi$, from ${\rm sgn}(\gamma^{(2n+1)})=-1$. 
The even-oddness of the particle-hole symmetry is 
$$
(S\Theta)^2=+1. 
$$  
Combining this, $\Theta^2=-1$ and $S^2=1$, one has 
$$
\{S,\tilde{\Theta}_\pm\}\varphi=0
$$
for any wavefunction $\varphi$. 

Since the present class satisfies $(\tilde{\Theta}_\pm)^2=\mp 1$ and $\{S,\tilde{\Theta}_\pm\}\varphi=0$, 
a similar calculation to that in the BDI class does not yield any information about 
the inner product $\langle \tilde{\Theta}_\pm D_a\varphi,S\varphi\rangle$.   

Similarly, we obtain 
$$
\langle B\varphi,\tilde{\Theta}_\pm D_a\varphi\rangle=0 
\quad \mbox{and} \quad 
\langle B\varphi, S\varphi\rangle=0. 
$$
Combining these results with $\{S,B\}=0$, we conclude that DIII class in two dimensions has an $\ze_2$ index. 

\bigskip\medskip

\noindent
$\bullet$ {\bf CII Class in Two Dimensions.}
In this case, one has 
$$
(\tilde{\Theta}_\pm)^2=\mp 1
$$
{from} $\Theta^2=-1$, and 
$$
\{\tilde{\Theta}_\pm,\gamma^{(2n+1)}\}\varphi=0
$$
for any wavefunction $\varphi$, from ${\rm sgn}(\gamma^{(2n+1)})=-1$. 
The even-oddness of the particle-hole symmetry is 
$$
(S\Theta)^2=-1. 
$$  
{From} this, $\Theta^2=-1$ and $S^2=1$, one has 
$$
[S,\tilde{\Theta}_\pm]\varphi=0
$$
for any wavefunction $\varphi$. 
Therefore, in the same way as in CI Class in two dimensions, we obtain 
$$
\langle S\varphi,\tilde{\Theta}_\pm D_a\varphi\rangle=0.
$$
In addition, we have 
$$
S\tilde{\Theta}_\pm D_a\varphi=\tilde{\Theta}_\pm D_aS\varphi
$$
{from} the above commutativity between $S$ and $\tilde{\Theta}_\pm$.  
These results imply that ${\rm dim}\;{\rm ker}\;(P_{\rm F}-\mathfrak{D}_a^\ast P_{\rm F}\mathfrak{D}_a-1)$ 
is even. Thus, the $\ze_2$ index is vanishing in CII class in two dimensions.

\bigskip\medskip

\noindent
$\bullet$ {\bf BDI and CI Classes in Four Dimensions.}
These two classes are derived from AI class with an additional chiral symmetry. 
As shown in the case of AI class in four dimensions, 
${\rm dim}\;{\rm ker}\;(P_{\rm F}-\mathfrak{D}_a^\ast P_{\rm F}\mathfrak{D}_a-1)$ 
is already even. Therefore, these two classes have no non-trivial $\ze_2$ index. 

\bigskip\medskip

\noindent
$\bullet$ {\bf DIII Class in Four Dimensions.}
In this case, one has 
$$
(\tilde{\Theta}_\pm)^2=+1
$$
{from} $\Theta^2=-1$, and 
$$
[\tilde{\Theta}_\pm,\gamma^{(2n+1)}]\varphi=0
$$
for any wavefunction $\varphi$, from ${\rm sgn}(\gamma^{(2n+1)})=+1$. 
Further, the evenness, $(S\Theta)^2=+1$, of the particle-hole symmetry for DIII class yields  
$$
\{S,\Theta\}\varphi=0,
$$
where we have used $S^2=1$ and $\Theta^2=-1$. Immediately, 
$$
\{S,\tilde{\Theta}_\pm\}\varphi=0.
$$

{From} the above observations, we have 
\begin{align*}
\tilde{\Theta}_\pm A\varphi&=\tilde{\Theta}_\pm \gamma^{(2n+1)}(P_{\rm F}-D_aP_{\rm F}D_a)\varphi\\
&=\gamma^{(2n+1)}\tilde{\Theta}_\pm (P_{\rm F}-D_aP_{\rm F}D_a)\varphi=A\tilde{\Theta}_\pm \varphi.
\end{align*}
Therefore, the transformation $\tilde{\Theta}_\pm$ satisfies the three conditions for 
the transformation $\tilde{\Sigma}$ in the case of BDI class in two dimensions.    
As a result, $A$ and $\tilde{\Theta}_\pm$ are simultaneously diagonalized as 
$$
A\varphi=\lambda\varphi\quad \mbox{and}\quad \tilde{\Theta}_\pm\varphi=\varphi.
$$

Let us consider the inner product $\langle S\varphi,B\varphi\rangle$. 
Note that 
$$
\langle \tilde{\Theta}_\pm S\varphi,\tilde{\Theta}_\pm B\varphi\rangle=
-\langle S\varphi, B\varphi\rangle,
$$
where we have used $\{S,\tilde{\Theta}_\pm\}\varphi=0$, $\tilde{\Theta}_\pm\varphi=\varphi$ and 
$\tilde{\Theta}_\pm B\varphi=B\tilde{\Theta}_\pm\varphi$ which is 
obtained from $[\tilde{\Theta}_\pm,\gamma^{(2n+1)}]\varphi=0$. 
On the other hand, one has 
$$
\langle \tilde{\Theta}_\pm S\varphi,\tilde{\Theta}_\pm B\varphi\rangle=
\langle B\varphi,S\varphi\rangle=-\langle S\varphi,B\varphi\rangle,
$$
where we have used $\{S,B\}=0$. These results do not yield any information. 
Thus, DIII in four dimensions does not show an $\ze_2$ index. 

\bigskip\medskip

\noindent
$\bullet$ {\bf CII Class in Four Dimensions.}
In this case, one has 
$$
(\tilde{\Theta}_\pm)^2=+1
$$
{from} $\Theta^2=-1$, and 
$$
[\tilde{\Theta}_\pm,\gamma^{(2n+1)}]\varphi=0
$$
for any wavefunction $\varphi$, from ${\rm sgn}(\gamma^{(2n+1)})=+1$. 
But, in contrast to DIII class, CII class shows the oddness for the particle-hole symmetry as 
$$
(S\Theta)^2=-1. 
$$
Combining this, $S^2=1$ and $\Theta^2=-1$, one has  
$$
[S,\Theta]\varphi=0.
$$
This yields 
$$
[S,\tilde{\Theta}_\pm]\varphi=0.
$$
This difference yields 
$$
\langle S\varphi,B\varphi\rangle=0
$$
in the same way as in the case of DIII class in four dimensions.  
In consequence, CII class in four dimensions has a $\ze_2$ index. 

\bigskip\medskip

\noindent
$\bullet$ {\bf CII Class in Six Dimensions.}
For convenience' sake, we consider first CII class for six dimensionality. 
In this case, we have 
$$
(\tilde{\Theta}_\pm)^2=\pm 1
$$
{from} $\Theta^2=-1$, and  
$$
\{\tilde{\Theta}_\pm,\gamma^{(2n+1)}\}\varphi=0
$$
for any wavefunction $\varphi$, from ${\rm sgn}(\gamma^{(2n+1)})=-1$. 
Combining $S^2=1$, $\Theta^2=-1$ and $(S\Theta)^2=-1$, one has 
$$
[S,\Theta]\varphi=0.
$$
Therefore, 
$$
[S,\tilde{\Theta}_\pm]\varphi=0
$$
for any wavefunction $\varphi$. 

Note that 
\begin{align*}
\tilde{\Theta}_\pm A\varphi&=\tilde{\Theta}_\pm \gamma^{(2n+1)}(P_{\rm F}-D_aP_{\rm F}D_a)\varphi\\
&=-\gamma^{(2n+1)}\tilde{\Theta}_\pm (P_{\rm F}-D_aP_{\rm F}D_a)\varphi=-A\tilde{\Theta}_\pm \varphi
\end{align*}
{from} $\{\tilde{\Theta}_\pm,\gamma^{(2n+1)}\}\varphi=0$. Further, this yields 
\begin{align*}
\tilde{\Theta}_\pm D_aA\varphi&=\tilde{\Theta}_\pm D_a\gamma^{(2n+1)}(P_{\rm F}-D_aP_{\rm F}D_a)\varphi\\
&=-\tilde{\Theta}_\pm \gamma^{(2n+1)}D_a(P_{\rm F}-D_aP_{\rm F}D_a)\varphi\\
&=\tilde{\Theta}_\pm AD_a\varphi=-A\tilde{\Theta}_\pm D_a\varphi.
\end{align*}
In addition, 
$$
[\tilde{\Theta}_\pm D_a,\gamma^{(2n+1)}]\varphi=0
$$
{from} $\{\tilde{\Theta}_\pm,\gamma^{(2n+1)}\}\varphi=0$.
Thus, we have to deal with the three vectors, $\tilde{\Theta}_\pm\varphi$, $S\varphi$ and $B\varphi$, again, 
as in the case of BDI class in two dimensions. 

Consider first the inner product $\langle S\varphi,\tilde{\Theta}_\pm D_a\varphi\rangle$. 
Note that 
\begin{align*}
\langle \tilde{\Theta}_\pm \tilde{\Theta}_\pm D_a\varphi, \tilde{\Theta}_\pm S\varphi\rangle
&=\pm \langle D_a\varphi, \tilde{\Theta}_\pm S\varphi\rangle\\
&=\pm \langle D_a\varphi, S\tilde{\Theta}_\pm \varphi\rangle\\
&=\pm \langle S\varphi, D_a\tilde{\Theta}_\pm \varphi\rangle
=\langle S\varphi, \tilde{\Theta}_\pm D_a\varphi\rangle,
\end{align*}
where we have used $(\tilde{\Theta}_\pm)^2=\pm 1$, $[S,\tilde{\Theta}_\pm]\varphi=0$, 
and ${\rm sgn}_\pm(D_a)=\pm 1$. On the other hand, the left-hand side can be written as 
$$
\langle \tilde{\Theta}_\pm \tilde{\Theta}_\pm D_a\varphi, \tilde{\Theta}_\pm S\varphi\rangle
=\langle S\varphi,\tilde{\Theta}_\pm D_a\varphi\rangle. 
$$
Therefore, these results do not give any information about the spectrum of $A$. 

Next, consider $\langle B\varphi,\tilde{\Theta}_\pm D_a\varphi\rangle$. Similarly, 
\begin{align*}
\langle \tilde{\Theta}_\pm \tilde{\Theta}_\pm D_a\varphi,\tilde{\Theta}_\pm B\varphi\rangle
&=\pm\langle D_a\varphi,\tilde{\Theta}_\pm B\varphi\rangle\\
&=\mp\langle D_a\varphi,B\tilde{\Theta}_\pm \varphi\rangle\\
&=\pm \langle B\varphi,D_a\tilde{\Theta}_\pm \varphi\rangle 
=\langle B\varphi,\tilde{\Theta}_\pm D_a\varphi\rangle,
\end{align*}
where we have used $(\tilde{\Theta}_\pm)^2=\pm 1$, $\{\tilde{\Theta}_\pm,\gamma^{(2n+1)}\}\varphi=0$, 
$\{B,D_a\}=0$ and ${\rm sgn}_\pm(D_a)=\pm 1$. The left-hand side can be written as 
$$
\langle \tilde{\Theta}_\pm \tilde{\Theta}_\pm D_a\varphi,\tilde{\Theta}_\pm B\varphi\rangle
=\langle B\varphi, \tilde{\Theta}_\pm D_a\varphi\rangle.
$$
These do not yield any information, again. 

Finally, consider $\langle B\varphi, S\varphi\rangle$. In order to treat this quantity, we introduce 
$$
\tilde{\Sigma}_\pm:=SD_a\tilde{\Theta}_\pm. 
$$
{From} $[\tilde{\Theta}_\pm D_a,\gamma^{(2n+1)}]\varphi=0$, one has 
$$
[\tilde{\Sigma}_\pm,\gamma^{(2n+1)}]\varphi=0.
$$
Since 
$$
SA\varphi=-AS\varphi\quad \mbox{and}\quad \tilde{\Theta}_\pm D_a A\varphi=-A\tilde{\Theta}_\pm D_a \varphi,
$$
we have 
$$
\tilde{\Sigma}_\pm A\varphi=A\tilde{\Sigma}_\pm \varphi. 
$$
Further, one has 
$$
(\tilde{\Sigma}_\pm)^2=SD_a\tilde{\Theta}_\pm SD_a\tilde{\Theta}_\pm=D_a\tilde{\Theta}_\pm D_a\tilde{\Theta}_\pm
=\pm (\tilde{\Theta}_\pm)^2=+1
$$
where we have used $[S,\tilde{\Theta}_\pm]\varphi=0$, $S^2=1$, ${\rm sgn}_\pm(D_a)=\pm 1$ 
and $(\tilde{\Theta}_\pm)^2=\pm 1$. 
Thus, the transformation $\tilde{\Sigma}_\pm$ satisfies the three conditions in the case of 
BDI class in two dimensions. As a result, we can choose eigenvectors $\varphi$ as 
$$
A\varphi=\lambda\varphi \quad \mbox{and}\quad \tilde{\Sigma}_\pm\varphi=\varphi.
$$
For such a vector $\varphi$, we have 
\begin{align*}
\tilde{\Sigma}_\pm B\varphi&=SD_a\tilde{\Theta}_\pm \gamma^{(2n+1)}(1-P_{\rm F}-D_aP_{\rm F}D_a)\varphi\\
&=-SD_a\gamma^{(2n+1)}\tilde{\Theta}_\pm (1-P_{\rm F}-D_aP_{\rm F}D_a)\varphi\\
&=-SD_a\gamma^{(2n+1)}(1-P_{\rm F}-D_aP_{\rm F}D_a)\tilde{\Theta}_\pm \varphi\\
&=S\gamma^{(2n+1)}D_a(1-P_{\rm F}-D_aP_{\rm F}D_a)\tilde{\Theta}_\pm \varphi\\
&=S\gamma^{(2n+1)}(1-P_{\rm F}-D_aP_{\rm F}D_a)D_a\tilde{\Theta}_\pm \varphi\\
&=-\gamma^{(2n+1)}(1-P_{\rm F}-D_aP_{\rm F}D_a)SD_a\tilde{\Theta}_\pm \varphi=-B\tilde{\Sigma}_\pm\varphi=-B\varphi, 
\end{align*}
where we have used $\{\tilde{\Theta}_\pm,\gamma^{(2n+1)}\}\varphi=0$, $\{D_a,\gamma^{(2n+1)}\}=0$ 
and $\{S,B\}=0$. 
Further, 
$$
\tilde{\Sigma}_\pm S\varphi=SD_a\tilde{\Theta}_\pm S\varphi=S^2D_a\tilde{\Theta}_\pm \varphi=S\tilde{\Sigma}_\pm =S\varphi,
$$
where we have used $[S,\tilde{\Theta}_\pm]\varphi=0$. From these observations, we obtain 
$$
\langle \tilde{\Sigma}_\pm B\varphi, \tilde{\Sigma}_\pm S\varphi\rangle 
=-\langle B\varphi, S\varphi\rangle.
$$ 
On the other hand, 
$$
\langle \tilde{\Sigma}_\pm B\varphi, \tilde{\Sigma}_\pm S\varphi\rangle
=\langle S\varphi,B\varphi\rangle=-\langle B\varphi,S\varphi\rangle, 
$$
where we have used $\{S,B\}=0$. These also do not give any information about the spectrum of $A$, again. 
In conclusion, CII class in six dimensions has no $\ze_2$ index. 

\bigskip\medskip

\noindent
$\bullet$ {\bf BDI Class in Six Dimensions.}
In this case, we have 
$$
(\tilde{\Theta}_\pm)^2=\mp 1
$$
{from} $\Theta^2=+1$, and 
$$
\{\tilde{\Theta}_\pm,\gamma^{(2n+1)}\}\varphi=0
$$
for any wavefunction $\varphi$, from ${\rm sgn}(\gamma^{(2n+1)})=-1$. 
The even-oddness of the particle-hole symmetry is given by 
$$
(S\Theta)^2=+1. 
$$
Combining this, $S^2=1$ and $\Theta^2=+1$, one has 
$$
[S,\Theta]\varphi=0
$$
for any wavefunction $\varphi$. Therefore, 
$$
[S,\tilde{\Theta}_\pm]\varphi=0.
$$
Compared with the preceding CII class, BDI class has $(\tilde{\Theta}_\pm)^2=\mp 1$ which 
is opposite to that of CII. Therefore, we obtain 
$$
\langle S\varphi,\tilde{\Theta}_\pm D_a\varphi\rangle=0 
$$
in the same way as in the case of the preceding CII class. Combining this with 
$$
S\tilde{\Theta}_\pm D_a\varphi=\tilde{\Theta}_\pm D_aS\varphi
$$
which is derived from the above $[S,\tilde{\Theta}_\pm]\varphi=0$, we can conclude that\hfill\break  
${\rm dim}\;{\rm ker}\;(P_{\rm F}-\mathfrak{D}_a^\ast P_{\rm F}\mathfrak{D}_a-1)$ 
is even. Namely, the $\ze_2$ index is vanishing for BDI class in six dimensions.  

\bigskip\medskip

\noindent
$\bullet$ {\bf CI Class in Six Dimensions.}
Similarly, in this case, we have 
$$
(\tilde{\Theta}_\pm)^2=\mp 1
$$
{from} $\Theta^2=+1$, and 
$$
\{\tilde{\Theta}_\pm,\gamma^{(2n+1)}\}\varphi=0
$$
for any wavefunction $\varphi$, from ${\rm sgn}(\gamma^{(2n+1)})=-1$. 
The even-oddness of the particle-hole symmetry is given by 
$$
(S\Theta)^2=-1. 
$$
Combining this, $S^2=1$ and $\Theta^2=+1$, one has 
$$
\{S,\Theta\}\varphi=0
$$
for any wavefunction $\varphi$. This yields 
$$
\{S,\tilde{\Theta}_\pm\}\varphi=0.
$$
Therefore, in the same way as in the case of CII class in six dimensions, we obtain the following:
There appears no information about the inner product $\langle S\varphi, \tilde{\Theta}_\pm D_a\varphi\rangle$. 
In addition, 
$$
\langle B\varphi,\tilde{\Theta}_\pm D_a\varphi\rangle=0, 
$$ 
for any wavefunction $\varphi$, and 
$$
\langle B\psi,S\psi\rangle=0
$$
for an eigenvector $\psi$ such that $A\psi=\lambda\psi$ and $SD_a\tilde{\Theta}_\pm\psi=\psi$. 
{From} these observations, we can conclude that CI class in six dimensions has a $\ze_2$ index. 

\bigskip\medskip

\noindent
$\bullet$ {\bf DIII Class in Six Dimensions.}
Similarly, in this case, one has 
$$
(\tilde{\Theta}_\pm)^2=\pm 1
$$
{from} $\Theta^2=-1$, and 
$$
\{\tilde{\Theta}_\pm,\gamma^{(2n+1)}\}\varphi=0
$$
for any wavefunction $\varphi$, from ${\rm sgn}(\gamma^{(2n+1)})=-1$. 
The even-oddness of the particle-hole symmetry is given by 
$$
(S\Theta)^2=+1. 
$$
Combining this, $S^2=1$ and $\Theta^2=-1$, one has 
$$
\{S,\Theta\}\varphi=0
$$
for any wavefunction $\varphi$. This yields 
$$
\{S,\tilde{\Theta}_\pm\}\varphi=0.
$$
Therefore, in the same way as in the case of CII class in six dimensions, 
we obtain 
$$
\langle S\varphi, \tilde{\Theta}_\pm D_a\varphi\rangle=0
$$
for any wavefunction $\varphi$. This implies 
${\rm dim}\;{\rm ker}\;(P_{\rm F}-\mathfrak{D}_a^\ast P_{\rm F}\mathfrak{D}_a-1)$ is even, 
and hence the $\ze_2$ index is vanishing for DIII class in six dimensions. 

\bigskip\medskip

\noindent
$\bullet$ {\bf BDI Class in Eight Dimensions.}
In this case, we have 
$$
(\tilde{\Theta}_\pm)^2=+1
$$
{from} $\Theta^2=+1$, and 
$$
[\tilde{\Theta}_\pm,\gamma^{(2n+1)}]\varphi=0
$$
for any wavefunction $\varphi$, from ${\rm sgn}(\gamma^{(2n+1)})=+1$. The parity of the particle-hole 
symmetry is given by 
$$
(S\Theta)^2=+1.
$$
Combining this, $S^2=1$ and $\Theta^2=+1$, one has 
$$
[S,\Theta]\varphi=0.
$$
Therefore, 
$$
[S,\tilde{\Theta}_\pm]\varphi=0
$$
for any wavefunction $\varphi$. These conditions are the same as in CII class in four dimensions.  
Therefore, we obtain 
$$
\langle S\varphi, B\varphi\rangle=0.
$$
In consequence, BDI class in eight dimensions has a $\ze_2$ index. 
\bigskip\medskip

\noindent
$\bullet$ {\bf CI Class in Eight Dimensions.}
Similarly, we have 
$$
(\tilde{\Theta}_\pm)^2=+1
$$
{from} $\Theta^2=+1$, and 
$$
[\tilde{\Theta}_\pm,\gamma^{(2n+1)}]\varphi=0
$$
for any wavefunction $\varphi$, from ${\rm sgn}(\gamma^{(2n+1)})=+1$. But the parity of the particle-hole 
symmetry is given by 
$$
(S\Theta)^2=-1.
$$
in the present CI class. 
Combining this, $S^2=1$ and $\Theta^2=+1$, one has 
$$
\{S,\Theta\}\varphi=0.
$$
This yields  
$$
\{S,\tilde{\Theta}_\pm\}\varphi=0
$$
for any wavefunction $\varphi$. This is different from the commutation relation in the case of the preceding BDI class. 
This implies that we cannot obtain any information about the inner product $\langle S\varphi, B\varphi\rangle$.  
Thus, CI class in eight dimensions does not show a non-trivial $\ze_2$ index. 
\bigskip\medskip

\noindent
$\bullet$ {\bf CII and DIII Classes in Eight Dimensions.}
These two classes are derived from AII class with an additional chiral symmetry. 
As shown in the case of AII class in eight dimensions, 
${\rm dim}\;{\rm ker}\;(P_{\rm F}-\mathfrak{D}_a^\ast P_{\rm F}\mathfrak{D}_a-1)$ 
is already even. Therefore, these two classes have no non-trivial $\ze_2$ index. 

\subsection{Odd Dimensions}
Let us treat the case of odd dimensions, $d=2n+1$, $n=1,2,\ldots$.
As seen in the statement of Theorem~\ref{thm:Z2oddNoS}, the case without chiral symmetry in odd dimensions 
is slightly special compared to the others. 
Therefore, we consider first this case. The rest will be treated in Sec.~\ref{sec:SSO} below. 

\subsubsection{No Chiral Symmetry}
\label{sec:NoSS}

In the case without chiral symmetry, we introduce two operators as 
$$
A=P_{\rm F}-D_aP_{\rm F}D_a
$$
and 
$$
B=1-P_{\rm F}-D_aP_{\rm F}D_a,
$$
where the Dirac operator $D_a$ is given by (\ref{defDiracOdd}).   
In the same way as in Appendix~\ref{traceclassA}, one can show that $A^{2n+2}$ is trace class. 
Therefore, the operator $A$ is compact. 
Since one can easily show that $A^2+B^2=1$ and $AB+BA=0$,
the operator $B$ maps an eigenvector of $A$ with eigenvalue $\lambda$ onto 
that with eigenvalue $-\lambda$, provided $0<|\lambda|<1$.  

In addition, one has 
$$
D_aA=D_a(P_{\rm F}-D_aP_{\rm F}D_a)=-(P_{\rm F}-D_aP_{\rm F}D_a)D_a=-AD_a.
$$
This implies that, if we define an integer-valued index by 
$$
{\rm Ind}^{(2n+1)}(D_a,P_{\rm F})={\rm dim}\;{\rm ker}\;(A-1)-{\rm dim}\;{\rm ker}\;(A+1),
$$
then the index ${\rm Ind}^{(2n+1)}(D_a,P_{\rm F})$ is always vanishing. 
But we can define a $\ze_2$ index by (\ref{Z2oddNoS}) in Theorem~\ref{thm:Z2oddNoS}, 
if the spectrum of $A$ satisfies the condition as in the case of AII class in two dimensions.

Instead of the above operator $A$, we can use 
$$
A':=\mathcal{P}_{\rm D}-U\mathcal{P}_{\rm D}U,
$$
in order to define a $\ze_2$ index, where $\mathcal{P}_{\rm D}=(1+D_a)/2$ and $U=(1-P_{\rm F})-P_{\rm F}$. 
These two indices coincide with each other. 
In fact, the following relation holds: \cite{GSB} 
$$
{\rm dim}\;{\rm ker}\;(A-1)={\rm dim}\;{\rm ker}\;(A'-1). 
$$ 
The proof is given in Appendix~\ref{RelationIndices}. 

Consider first A class, which has no symmetry. Clearly, we have
$$
D_aB=(D_a-D_aP_{\rm F}-P_{\rm F}D_a)=(1-D_aP_{\rm F}D_a-P_{\rm F})D_a=BD_a
$$
by using $D_a^2=1$. However, this does not yield any information about $\langle D_a\varphi,B\varphi\rangle$, 
where $\varphi$ is an eigenvector of $A$ with eigenvalue $\lambda$ satisfying $0<|\lambda|<1$.
Thus, A class has no $\ze_2$ index in any odd dimensions. 
\bigskip

\subsubsection{AI an AII Classes}
\label{sec:AIAII}
Next let us consider AI and AII classes, which have time-reversal symmetry only.

\medskip

\noindent
$\bullet$ {\bf AI Class in One Dimension.}
In one dimension, there is no gamma matrix. Therefore, the Dirac operator is given by 
\begin{equation}
\label{Dirac1D}
D_a(x)=
\begin{cases} \ \ 1, & \text{$x> a$};\\
 -1, & \text{$x<a$},
\end{cases}
\end{equation}
and the extended time-reversal transformation is given by $\tilde{\Theta}=\Theta$. 
Clearly, one has 
$$
\tilde{\Theta}^2=+1
$$
{from} $\Theta^2=+1$.

Note that 
$$
\tilde{\Theta}A\varphi=\tilde{\Theta}(P_{\rm F}-D_aP_{\rm F}D_a)\varphi=A\tilde{\Theta}\varphi.
$$
Therefore, the transformation $\tilde{\Theta}$ satisfies the conditions without 
that for the gamma matrix $\gamma^{(2n+1)}$ in the case of BDI class in two dimensions. 
We can choose the eigenvectors of $A$ so that 
$$
A\varphi=\lambda \varphi \quad \mbox{and}\quad \tilde{\Theta}\varphi=\varphi.
$$
Combining this, $\tilde{\Theta}D_a\varphi=D_a\tilde{\Theta}\varphi$ 
and $\tilde{\Theta}B\varphi=B\tilde{\Theta}\varphi$, we obtain 
$$
\langle \tilde{\Theta}D_a\varphi,\tilde{\Theta}B\varphi\rangle=\langle D_a\varphi,B\varphi\rangle.
$$
On the other hand, the left-hand side can be written as  
$$
\langle \tilde{\Theta}D_a\varphi,\tilde{\Theta}B\varphi\rangle
=\langle B\varphi,D_a\varphi\rangle=\langle D_a\varphi,B\varphi\rangle.
$$
These imply that AI class in one dimension has no $\ze_2$ index. 

\bigskip\medskip

\noindent
$\bullet$ {\bf AII Class in One Dimension.}
In this case, we have $\tilde{\Theta}^2=-1$ from $\Theta^2=-1$. 
Further, in the same way as in the preceding AI class, one has 
$$
\tilde{\Theta}A\varphi=A\tilde{\Theta}\varphi.
$$
Therefore, if $\varphi$ is an eigenvector of $A$, then $\tilde{\Theta}\varphi$ is also an eigenvector 
of $A$ with the same eigenvalue. Combining this with $\langle\varphi, \tilde{\Theta}\varphi\rangle=0$ 
which is derived from $\tilde{\Theta}^2=-1$ as in the case of the Kramers doublet, we obtain 
that ${\rm dim}\;{\rm ker}\;(A-1)$ is even. This implies that the $\ze_2$ index is vanishing 
for AII class in one dimension.  

\bigskip\medskip

\noindent
$\bullet$ {\bf AI Class in Three Dimensions.}
In three or higher dimensions, the $(2n+1)$-th component of $\gamma$ is given by $\gamma^{(2n+1)}$ 
{from} the definition (\ref{gammaOddDim}).  
Therefore, by relying on (\ref{C-gammaC-}) in Appendix~\ref{AppGamma}, we use only 
$$
\tilde{\Theta}_-:=C_-\Theta
$$
as the time-reversal transformation in odd dimensions. Namely, we can define the signature 
of the Dirac operator $D_a$ only for $C_-K$ in the present case. 
 
{From} $\Theta^2=+1$ for AI class and ${\rm sgn}(C_-K)^2=-1$ for $n=1$, we have 
$$
(\tilde{\Theta}_-)^2=-1.
$$
Clearly, one has 
$$
\tilde{\Theta}_-A\varphi=\tilde{\Theta}_-(P_{\rm F}-D_aP_{\rm F}D_a)\varphi=A\tilde{\Theta}_-\varphi.
$$

On the other, hand, we have 
$$
\langle \varphi, \tilde{\Theta}_-\varphi\rangle=0
$$
{from} $(\tilde{\Theta}_-)^2=-1$. 
As in the preceding case of AII class in one dimensions, these observations imply that 
${\rm dim}\;{\rm ker}\;(A-1)$ is even. Thus, the $\ze_2$ index is vanishing 
for AI class in three dimensions.   

\bigskip\medskip

\noindent
$\bullet$ {\bf AII Class in Three Dimensions.}
In this case, we have 
$$
(\tilde{\Theta}_-)^2=+1
$$ 
from $\Theta^2=-1$ for AII class and ${\rm sgn}(C_K)^2=-1$. 
Further, 
\begin{equation}
\label{tildeTheta-DacommuAII3D}
\tilde{\Theta}_-D_a\varphi=-D_a\tilde{\Theta}_-\varphi
\end{equation}
for any wavefunction $\varphi$, from ${\rm sgn}_-(D_a)=-1$. 
Since we have 
$$
A\tilde{\Theta}_-\varphi=\tilde{\Theta}_-A\varphi
$$
for any wavefunction $\varphi$, 
the transformation $\tilde{\Theta}$ satisfies the conditions without 
that for the gamma matrix $\gamma^{(2n+1)}$ in the case of BDI class in two dimensions. 
Further, one has 
$$
D_aB=D_a(1-P_{\rm F}-D_aP_{\rm F}D_a)=BD_a
$$
and 
$$
\tilde{\Theta}_-B\varphi=B\tilde{\Theta}_-\varphi.
$$
{From} these observations, the above condition (\ref{tildeTheta-DacommuAII3D}) 
which is different from that in the case of AI class in one dimensions yields 
$$
\langle D_a\varphi, B\varphi\rangle=0
$$
in the same way as in the case of AI class in one dimensions. 
In consequence, AII class in three dimensions has a $\ze_2$ index. 

\bigskip\medskip

\noindent
$\bullet$ {\bf AI Class in Five Dimensions.}
Similarly, one has 
$$
(\tilde{\Theta}_-)^2=-1
$$
{from} $\Theta^2=+1$ and ${\rm sgn}(C_-K)^2=-1$. 
This yields  
$$
\langle \varphi,\tilde{\Theta}_-\varphi\rangle=0
$$
in the same way as in the case of AI class in three dimensions.
Therefore, ${\rm dim}\;{\rm ker}\;(A-1)$ is even, and the $\ze_2$ index is vanishing 
for AI class in five dimensions. 

\bigskip\medskip

\noindent
$\bullet$ {\bf AII Class in Five Dimensions.}
{From} $\Theta^2=-1$ and ${\rm sgn}(C_-K)^2=-1$, one has 
$$
(\tilde{\Theta}_-)^2=+1. 
$$
Further, we have 
$$
\tilde{\Theta}_-D_a\varphi=D_a\tilde{\Theta}_-\varphi
$$
by the property ${\rm sgn}_-(D_a)=+1$. Therefore, AII class in five dimensions has no $\ze_2$ index 
in the same way as in the case of AI class in one dimension. 

\bigskip\medskip

\noindent
$\bullet$ {\bf AI Class in Seven Dimensions.}
Similarly, we obtain 
$$
(\tilde{\Theta}_-)^2=+1
$$
by using $\Theta^2=+1$ for AI class and ${\rm sgn}(C_-K)^2=+1$. In addition, 
$$
\tilde{\Theta}_-D_a\varphi=-D_a\tilde{\Theta}_-\varphi
$$
for any wavefunction $\varphi$, from ${\rm sgn}_-(D_a)=-1$.
Thus, we obtain that AI class in seven dimensions has a $\ze_2$ index in the same way as in 
the case of AII class in three dimensions. 

\bigskip\medskip

\noindent
$\bullet$ {\bf AII Class in Seven Dimensions.}
In this case, we have 
$$
(\tilde{\Theta}_-)^2=-1
$$
{from} $\Theta^2=-1$ for AII class and ${\rm sgn}(C_-K)^2=+1$. 
This yields  
$$
\langle \varphi,\tilde{\Theta}_-\varphi\rangle=0
$$
in the same way as in the case of AI class in three dimensions.
Therefore, ${\rm dim}\;{\rm ker}\;(A-1)$ is even, and the $\ze_2$ index is vanishing 
for AII class in seven dimensions. 

\bigskip 

\subsubsection{C and D Classes}
\label{sec:CD}
These classes have particle-hole symmetry only. 
We define the extended particle-hole symmetry by 
$$
\tilde{\Xi}_-:=C_-\Xi. 
$$

\medskip

\noindent
$\bullet$ {\bf D Class in One Dimension.}
In one dimension, there appears no gamma matrix, and hence the Dirac operator is given by (\ref{Dirac1D}). 
Furthermore, one has $\tilde{\Xi}=\Xi$. This yields 
$$
\tilde{\Xi}^2=+1
$$
{from} $\Xi^2=+1$ for D class. 

Since we have $\Xi P_{\rm F}\varphi=(1-P_{\rm F})\Xi\varphi$ with $E_{\rm F}=0$ 
for any wavefunction $\varphi$, we obtain 
\begin{equation}
\label{tildeXiAcommuD1D}
\tilde{\Xi}A\varphi=\tilde{\Xi}(P_{\rm F}-D_aP_{\rm F}D_a)\varphi=-A\tilde{\Xi}\varphi.
\end{equation}
Thus, we have to deal with three vectors, $\tilde{\Xi}\varphi$, $B\varphi$ and $D_a\varphi$, 
{from} the commutation relation $\{D_a,A\}=0$. 

Consider first the inner product $\langle D_a\varphi,\tilde{\Xi}\varphi\rangle$.  
Note that 
$$
\langle \tilde{\Xi} \tilde{\Xi}\varphi, \tilde{\Xi}D_a\varphi\rangle =\langle \varphi,\tilde{\Xi}D_a\varphi\rangle
=\langle D_a\varphi,\tilde{\Xi}\varphi\rangle,
$$
where we have used $\tilde{\Xi}^2=+1$ and the fact that the Dirac operator is given by (\ref{Dirac1D}). 
On the other hand, the left-hand side can be written as 
$$
\langle \tilde{\Xi} \tilde{\Xi}\varphi, \tilde{\Xi}D_a\varphi\rangle =\langle D_a\varphi,\tilde{\Xi}\varphi\rangle.
$$
These do not give any information about the inner product $\langle D_a\varphi,\tilde{\Xi}\varphi\rangle$. 

Next, consider $\langle B\varphi,\tilde{\Xi}\varphi\rangle$. 
To begin with, we note that 
$$
\tilde{\Xi}B\varphi=\tilde{\Xi}(1-P_{\rm F}-D_aP_{\rm F}D_a)\varphi
=-B\tilde{\Xi}\varphi 
$$
for any wavefunction $\varphi$, where we have used $\Xi P_{\rm F}\varphi=(1-P_{\rm F})\Xi\varphi$. 
By using this and $\tilde{\Xi}^2=+1$, we obtain 
$$
\langle \tilde{\Xi}\tilde{\Xi}\varphi,\tilde{\Xi}B\varphi\rangle
=\langle \varphi,\tilde{\Xi}B\varphi\rangle=-\langle B\varphi,\tilde{\Xi}\varphi\rangle.
$$
On the other hand, the left-hand side can be written as 
$$
\langle \tilde{\Xi}\tilde{\Xi}\varphi,\tilde{\Xi}B\varphi\rangle
=\langle B\varphi,\tilde{\Xi}\varphi\rangle.
$$
These two results imply 
$$
\langle B\varphi,\tilde{\Xi}\varphi\rangle=0.
$$

Finally, let us treat $\langle D_a\varphi,B\varphi\rangle$.  
For this purpose, we introduce an operator, 
$$
\tilde{\Sigma}:=D_a\tilde{\Xi}.
$$
Then, from (\ref{tildeXiAcommuD1D}) and $\{D_a,A\}=0$, one has 
$$
\tilde{\Sigma}A\varphi=A\tilde{\Sigma}\varphi.
$$
Further, 
$$
(\tilde{\Sigma})^2=D_a\tilde{\Xi}D_a\tilde{\Xi}=D_a^2(\tilde{\Xi})^2=1,
$$
where we have used $D_a^2=1$ and $(\tilde{\Xi})^2=1$. 
{From} these observations, we can choose the eigenvectors of $A$ so that 
$$
A\varphi=\lambda\varphi\quad\mbox{and}\quad \tilde{\Sigma}\varphi=\varphi,
$$
as in the same way in the case of BDI class in two dimensions. 
By relying on this, we have 
$$
\langle \tilde{\Sigma}D_a\varphi,\tilde{\Sigma}B\varphi\rangle
=-\langle D_a\varphi,B\varphi\rangle,
$$
where we have used $\{\tilde{\Xi},B\}\varphi=0$. On the other hand, the left-hand side can be written as 
$$
\langle \tilde{\Sigma}D_a\varphi,\tilde{\Sigma}B\varphi\rangle
=\langle B\varphi,D_a\varphi\rangle=\langle D_a\varphi,B\varphi\rangle.
$$
These imply $\langle D_a\varphi,B\varphi\rangle=0$. 

To summarize, the above three results about the inner products between the three vectors yield 
the conclusion that the present D class in one dimension has a $\ze_2$ index. 

\bigskip\medskip

\noindent
$\bullet$ {\bf C Class in One Dimension.}
In this case, one has 
$$
(\tilde{\Xi})^2=-1
$$
{from} $\Xi^2=-1$. 

As in the preceding case of D class in one dimension, let us consider 
the inner product $\langle\tilde{\Xi}\varphi, D_a\varphi\rangle$. 
Then, we obtain 
$$
\langle\tilde{\Xi}\varphi, D_a\varphi\rangle=0
$$
because $(\tilde{\Xi})^2=-1$ which is opposite sign to the case of the preceding D class in one dimension. 
This implies that ${\rm dim}\;{\rm ker}\;(A-1)$ is even, and hence the $\ze_2$ index is vanishing 
for C class in one dimension. 
 
\bigskip\medskip

\noindent
$\bullet$ {\bf D Class in Three Dimensions.}
In this case, one has 
$$
(\tilde{\Xi}_-)^2=-1
$$
{from} $\Xi^2=+1$ for D class and ${\rm sgn}(C_K)^2=-1$. Further, 
$$
\tilde{\Xi}_-D_a\varphi=-D_a\tilde{\Xi}_-\varphi
$$
for any wavefunction $\varphi$, from ${\rm sgn}_-(D_a)=-1$. 
Taking account of these two facts, it is sufficient to deal with 
the three inner products about the three vectors, $D_a\varphi$, $B\varphi$ and $\tilde{\Xi}_-\varphi$, 
in the same way as in the case of D class in one dimension. 
As a result, one cannot obtain any information about the spectrum of $A$. 
Therefore, the present D class in three dimensions has no $\ze_2$ index.
  
\bigskip\medskip

\noindent
$\bullet$ {\bf C Class in Three Dimensions.}
Similarly, one has 
$$
(\tilde{\Xi}_-)^2=+1
$$
{from} $\Xi^2=-1$ for C class and ${\rm sgn}(C_K)^2=-1$. Further, 
$$
\tilde{\Xi}_-D_a\varphi=-D_a\tilde{\Xi}_-\varphi
$$
for any wavefunction $\varphi$, from ${\rm sgn}_-(D_a)=-1$. 
In this case, we obtain 
$$
\langle \tilde{\Xi}_-\varphi,D_a\varphi\rangle=0
$$
in the same way as in the case of D class in one dimension because $(\tilde{\Xi}_-)^2=+1$ 
which is opposite sign to that of the preceding D class in three dimensions.
Thus, ${\rm dim}\;{\rm ker}\;(A-1)$ is even, and hence the $\ze_2$ index is vanishing 
for C class in three dimensions.

\bigskip\medskip

\noindent
$\bullet$ {\bf D Class in Five Dimensions.}
In this case, one has 
$$
(\tilde{\Xi}_-)^2=-1
$$
{from} $\Xi^2=+1$ for D class and ${\rm sgn}(C_-K)^2=-1$. Further, 
$$
\tilde{\Xi}_-D_a\varphi=D_a\tilde{\Xi}_-\varphi
$$
for any wavefunction $\varphi$, from ${\rm sgn}_-(D_a)=+1$. 
These two signs are opposite to those in the preceding case of C class in three dimensions. 
Therefore, we obtain the same conclusion as 
$$
\langle \tilde{\Xi}_-\varphi,D_a\varphi\rangle=0
$$
in the same way as in the case of D class in one dimension.
In consequence, ${\rm dim}\;{\rm ker}\;(A-1)$ is even, and the $\ze_2$ index is vanishing 
for D class in five dimensions. 

\bigskip\medskip

\noindent
$\bullet$ {\bf C Class in Five Dimensions.}
Similarly, one has 
$$
(\tilde{\Xi}_-)^2=+1
$$
{from} $\Xi^2=-1$ for C class and ${\rm sgn}(C_-K)^2=-1$. Further, 
$$
\tilde{\Xi}_-D_a\varphi=D_a\tilde{\Xi}_-\varphi
$$
for any wavefunction $\varphi$, from ${\rm sgn}_-(D_a)=+1$. 
In this case, one cannot obtain any information about the inner 
product $\langle \tilde{\Xi}_-\varphi,D_a\varphi\rangle$. 
Further, on the rest of the two inner products, 
we can obtain the same results as in the case of D class in one dimension. 
Therefore, we conclude that C class in five dimensions has a $\ze_2$ index. 

\bigskip\medskip

\noindent
$\bullet$ {\bf D Class in Seven Dimensions.}
In this case, one has 
$$
(\tilde{\Xi}_-)^2=+1
$$
{from} $\Xi^2=+1$ for D class and ${\rm sgn}(C_-K)^2=+1$. In addition, 
$$
\tilde{\Xi}_-D_a\varphi=-D_a\tilde{\Xi}_-\varphi
$$
for any wavefunction $\varphi$, from ${\rm sgn}_-(D_a)=-1$. 
These conditions are the same as in the case of C class in three dimensions. 
Therefore, the $\ze_2$ index for D class in seven dimensions is vanishing. 

\bigskip\medskip

\noindent
$\bullet$ {\bf C Class in Seven Dimensions.}
Similarly, one has 
$$
(\tilde{\Xi}_-)^2=-1
$$
{from} $\Xi^2=-1$ for C class and ${\rm sgn}(C_-K)^2=+1$. Further, 
$$
\tilde{\Xi}_-D_a\varphi=-D_a\tilde{\Xi}_-\varphi
$$
for any wavefunction $\varphi$, from ${\rm sgn}_-(D_a)=-1$. 
These conditions are the same as in the case of D class in three dimensions. 
In conclusion, C class in seven dimensions has no $\ze_2$ index. 

\subsubsection{Chiral Symmetry in Odd Dimensions}
\label{sec:SSO}
In odd dimensions, the rest of the CAZ classes are BDI, CI, DIII and CII, which have chiral 
and time-reversal symmetries. To begin with, we recall two operators, 
$$
A=S(\mathcal{P}_{\rm D}-U\mathcal{P}_{\rm D}U)
\quad \mbox{and} \quad  
B=S(1-\mathcal{P}_{\rm D}-U\mathcal{P}_{\rm D}U),
$$
where the projection $\mathcal{P}_{\rm D}$ is given by (\ref{projectionchiral}), 
and the unitary operator $U$ is given by (\ref{U}); $S$ is the chiral operator. 
As shown in Sec.~\ref{subsec:ChiralIndex}, theses two operators satisfy $A^2+B^2=1$ {and} $AB+BA=0$,  
and the integer-valued index is defined by 
$$
{\rm Ind}^{(2n+1)}(D_a,S,U)=\frac{1}{2}\left[{\rm dim}\;{\rm ker}\;(A-1)-{\rm dim}\;{\rm ker}\;(A+1)\right].
$$
When the integer-valued index is vanishing, the $\ze_2$-valued index is defined by 
$$
{\rm Ind}_2^{(2n+1)}(D_a,S,U)=\frac{1}{2}{\rm dim}\;{\rm ker}\;(A-1)\ \mbox{modulo}\ 2. 
$$

\medskip

\noindent
$\bullet$ {\bf BDI Class in One Dimension.}
In one dimension, there is no gamma matrix. The Dirac operator $D_a$ is given by (\ref{Dirac1D}). 
Therefore, one has 
$$
\tilde{\Theta}=\Theta\quad \mbox{and}\quad (\tilde{\Theta})^2=\Theta^2=+1
$$
for the present BDI class in one dimension. Since $(S\Theta)^2=+1$ for the particle-hole symmetry, 
we have 
$$
[S,\tilde{\Theta}]\varphi=0
$$
for any wavefunction $\varphi$. Using this relation, we have  
$$
\tilde{\Theta}A\varphi=\tilde{\Theta}S(\mathcal{P}_{\rm D}-U\mathcal{P}_{\rm D}U)\varphi
=S\tilde{\Theta}(\mathcal{P}_{\rm D}-U\mathcal{P}_{\rm D}U)\varphi=A\tilde{\Theta}\varphi.
$$
Therefore, the time-reversal symmetry dose not give any information about the spectrum of $A$. 
In conclusion, BDI class in one dimension has an integer-valued index as in the case of generic 
chiral symmetric class in odd dimensions. 

\bigskip\medskip

\noindent
$\bullet$ {\bf CI Class in One Dimension.}
Similarly, one has $\tilde{\Theta}=\Theta$, and hence 
$$
(\tilde{\Theta})^2=\Theta^2=+1. 
$$
{From} the particle-hole symmetry $(S\Theta)^2=-1$, 
$$
\{S,\tilde{\Theta}\}\varphi=0
$$
for any wavefunction $\varphi$. This yields 
\begin{align*}
\tilde{\Theta}A\varphi&=\tilde{\Theta}S(\mathcal{P}_{\rm D}-U\mathcal{P}_{\rm D}U)\varphi\\
&=-S\tilde{\Theta}(\mathcal{P}_{\rm D}-U\mathcal{P}_{\rm D}U)\varphi=-A\tilde{\Theta}\varphi. 
\end{align*}
This implies that the integer-valued index is vanishing. 

Note that 
\begin{align*}
\tilde{\Theta}UA\varphi&=\tilde{\Theta}US(\mathcal{P}_{\rm D}-U\mathcal{P}_{\rm D}U)\varphi\\
&=-\tilde{\Theta}SU(\mathcal{P}_{\rm D}-U\mathcal{P}_{\rm D}U)\varphi\\
&=\tilde{\Theta}S(\mathcal{P}_{\rm D}-U\mathcal{P}_{\rm D}U)U\varphi=-A\tilde{\Theta}U\varphi,
\end{align*}
where we have used $SU=-US$ and the above $\tilde{\Theta}A\varphi=-A\tilde{\Theta}\varphi$. 
Further, one has 
$$
[\tilde{\Theta}U,S]\varphi=0
$$
for any wavefunction $\varphi$, where we have used $\{S,\tilde{\Theta}\}\varphi=0$ and $\{S,U\}=0$. 
Therefore, in order to check whether the present class has a $\ze_2$ index, it is enough to consider 
the inner product $\langle \tilde{\Theta}U\varphi,B\varphi\rangle$. 
Note that 
$$
\tilde{\Theta}B\varphi=\tilde{\Theta}S(1-\mathcal{P}_{\rm D}-U\mathcal{P}_{\rm D}U)\varphi
=-B\tilde{\Theta}\varphi,
$$
where we have used $\{S,\tilde{\Theta}\}\varphi=0$. Further, from $\{S,U\}=0$, we have 
$$
UB=US(1-\mathcal{P}_{\rm D}-U\mathcal{P}_{\rm D}U)
=-SU(1-\mathcal{P}_{\rm D}-U\mathcal{P}_{\rm D}U)=-BU.
$$
Combining these with $\tilde{\Theta}^2=+1$, we obtain 
\begin{align*}
\langle \tilde{\Theta}\tilde{\Theta}U\varphi,\tilde{\Theta}B\varphi\rangle
&=\langle U\varphi,\tilde{\Theta}B\varphi\rangle\\
&=-\langle U\varphi,B\tilde{\Theta}\varphi\rangle\\
&=\langle B\varphi,U\tilde{\Theta}\varphi\rangle=\langle B\varphi,\tilde{\Theta}U\varphi\rangle.
\end{align*}
On the other hand, the left-hand side can be written as 
$$
\langle \tilde{\Theta}\tilde{\Theta}U\varphi,\tilde{\Theta}B\varphi\rangle
=\langle B\varphi,\tilde{\Theta}U\varphi\rangle. 
$$
These two results do not give any information about the spectrum of $A$, and 
hence CI class in one dimension has no index. 

\bigskip\medskip

\noindent
$\bullet$ {\bf DIII Class in One Dimension.}
In this case, from $\tilde{\Theta}=\Theta$, one has  
$$
(\tilde{\Theta})^2=\Theta^2=-1.
$$
Combining this with $(S\Theta)^2=+1$ for the particle-hole symmetry, we have  
$$
\{S,\tilde{\Theta}\}\varphi=0
$$
for any wavefunction $\varphi$. Therefore, we obtain that the integer-valued index is vanishing, and that   
$$
\langle B\varphi,\tilde{\Theta}U\varphi\rangle=0
$$
in the same way as in the preceding case of CI class in one dimension. 
In conclusion, DIII class in one dimension has a $\ze_2$ index. 

\bigskip\medskip

\noindent
$\bullet$ {\bf CII Class in One Dimension.}
Similarly, we have 
$$
(\tilde{\Theta})^2=-1
$$
{from} $\Theta^2=-1$, and 
$$
[S,\tilde{\Theta}]\varphi=0
$$
for any wavefunction $\varphi$, from $(S\Theta)^2=-1$. The second condition yields 
$$
\tilde{\Theta}A\varphi=A\tilde{\Theta}\varphi.
$$
{From} the first condition $(\tilde{\Theta})^2=-1$, we have 
$$
\langle \varphi,\tilde{\Theta}\varphi\rangle=0.
$$
These imply that the integer-valued index for CII class in one dimension is always even, i.e., $2\ze$. 

\bigskip\medskip

\noindent
$\bullet$ {\bf BDI Class in Three Dimensions.}
In three or higher dimensions, we define the extended time-reversal transformation by 
$$
\tilde{\Theta}_-:=C_-\Theta. 
$$
For the present case, we have 
$$
(\tilde{\Theta}_-)^2=-1
$$
{from} $\Theta^2=+1$ for BDI class and ${\rm sgn}(C_-K)^2=-1$. On the other hand, from 
$(S\Theta)^2=+1$ for the particle-hole symmetry, we obtain 
$$
[S,\Theta]\varphi=0
$$
for any wavefunction $\varphi$, where we have used $\Theta^2=+1$, again. Clearly, this yields 
$$
[S,\tilde{\Theta}_-]\varphi=0.
$$
Further, 
$$
\tilde{\Theta}_-D_a\varphi=-D_a\tilde{\Theta}_-\varphi
$$
{from} ${\rm sgn}_-(D_a)=-1$. By using this, we have 
\begin{align*}
\tilde{\Theta}_-(\mathcal{P}_{\rm D}-U\mathcal{P}_{\rm D}U)\varphi
&=\frac{1}{2}\tilde{\Theta}_-(D_a-UD_aU)\varphi\\
&=-\frac{1}{2}(D_a-UD_aU)\tilde{\Theta}_-\varphi\\
&=-(\mathcal{P}_{\rm D}-U\mathcal{P}_{\rm D}U)\tilde{\Theta}_-\varphi
=-(\mathcal{P}_{\rm D}-U\mathcal{P}_{\rm D}U)\tilde{\Theta}_-\varphi.
\end{align*}
Combining this with the above $[S,\tilde{\Theta}_-]\varphi=0$, we obtain 
$$ 
\tilde{\Theta}_-A\varphi=-A\tilde{\Theta}_-\varphi.
$$
This implies that the integer-valued index is vanishing. 

In order to check whether the present case has a $\ze_2$ index, it is sufficient to consider 
the inner product $\langle \tilde{\Theta}_-\varphi,B\varphi\rangle$. 
Note that 
\begin{align*}
\tilde{\Theta}_-B\varphi&=\tilde{\Theta}_-S(1-\mathcal{P}_{\rm D}-U\mathcal{P}_{\rm D}U)\varphi\\
&=S\tilde{\Theta}_-(1-\mathcal{P}_{\rm D}-U\mathcal{P}_{\rm D}U)\varphi\\
&=\frac{1}{2}S\tilde{\Theta}_-(-D_a-UD_aU)\varphi\\
&=-\frac{1}{2}S(-D_a-UD_aU)\tilde{\Theta}_-\varphi=-B\tilde{\Theta}_-\varphi,
\end{align*}
where we have used $[S,\tilde{\Theta}_-]\varphi=0$ and $\tilde{\Theta}_-D_a\varphi=-D_a\tilde{\Theta}_-\varphi$. 
Combining this with $(\tilde{\Theta}_-)^2=-1$, we obtain 
$$
\langle \tilde{\Theta}_-\tilde{\Theta}_-\varphi,\tilde{\Theta}_-B\varphi\rangle
=-\langle\varphi,\tilde{\Theta}_-B\varphi\rangle
=\langle B\varphi,\tilde{\Theta}_-\varphi\rangle.
$$
On the other hand, the left-hand side can be written as 
$$
\langle \tilde{\Theta}_-\tilde{\Theta}_-\varphi,\tilde{\Theta}_-B\varphi\rangle
=\langle B\varphi, \tilde{\Theta}_-\varphi\rangle.
$$
These results do not give any information about the spectrum of $A$. 
Thus, BDI class in three dimensions has no index. 

\bigskip\medskip

\noindent
$\bullet$ {\bf CI Class in Three Dimensions.}
In this case, we have 
$$
(\tilde{\Theta}_-)^2=-1
$$
{from} $\Theta^2=+1$ and ${\rm sgn}(C_-K)^2=-1$, and 
$$
\{S,\Theta\}\varphi=0
$$
for any wavefunction $\varphi$, from $(S\Theta)^2=-1$ and $\Theta^2=+1$. The second relation yields 
$$
\{S,\tilde{\Theta}_-\}\varphi=0. 
$$
Further, we have 
$$
\tilde{\Theta}_-D_a\varphi=-D_a\tilde{\Theta}_-\varphi
$$
{from} ${\rm sgn}_-(D_a)=-1$. 

Note that 
\begin{align*}
\tilde{\Theta}_-A\varphi&=\tilde{\Theta}_-S(\mathcal{P}_{\rm D}-U\mathcal{P}_{\rm D}U)\varphi\\
&=-S\tilde{\Theta}_-(\mathcal{P}_{\rm D}-U\mathcal{P}_{\rm D}U)\varphi=A\tilde{\Theta}_-\varphi, 
\end{align*}
where we have used $\{S,\tilde{\Theta}_-\}\varphi=0$ and $\tilde{\Theta}_-D_a\varphi=-D_a\tilde{\Theta}_-\varphi$. 
By using this and $\{S,U\}=0$, we obtain 
\begin{align*}
\tilde{\Theta}_-UA\varphi&=\tilde{\Theta}_-US(\mathcal{P}_{\rm D}-U\mathcal{P}_{\rm D}U)\varphi\\
&=-\tilde{\Theta}_-SU(\mathcal{P}_{\rm D}-U\mathcal{P}_{\rm D}U)\varphi\\
&=\tilde{\Theta}_-S(\mathcal{P}_{\rm D}-U\mathcal{P}_{\rm D}U)U\varphi\\
&=\tilde{\Theta}_-AU\varphi=A\tilde{\Theta}_-U\varphi
\end{align*}
and
$$
[\tilde{\Theta}_-U,S]\varphi=0 
$$
for any wavefunctions $\varphi$. Therefore, if $\varphi$ is an eigenvector of $A$, 
then $\tilde{\Theta}_-UA\varphi$ is also an eigenvector of $A$ with the same eigenvalue 
in the same sector of the eigenspace of $S$. 

By relying on the above observations, let us consider 
the inner product $\langle \varphi,\tilde{\Theta}_-U\varphi\rangle$. Note that 
$$
\langle \tilde{\Theta}_-\varphi,\tilde{\Theta}_-\tilde{\Theta}_-U\varphi\rangle
=-\langle \tilde{\Theta}_-\varphi,U\varphi\rangle=-\langle \tilde{\Theta}_-U\varphi,\varphi\rangle,
$$
where we have used $(\tilde{\Theta}_-)^2=-1$. The left-hand side can be computed as 
$$
\langle \tilde{\Theta}_-\varphi,\tilde{\Theta}_-\tilde{\Theta}_-U\varphi\rangle
=\langle \tilde{\Theta}_-U\varphi,\varphi\rangle.
$$
Therefore, we obtain  
$$
\langle \tilde{\Theta}_-U\varphi,\varphi\rangle=0.
$$
This implies that the integer-valued index for CI class in three dimensions is always even, i.e., $2\ze$. 

\bigskip\medskip

\noindent
$\bullet$ {\bf DIII Class in Three Dimensions.}
Similarly, we have 
$$
(\tilde{\Theta}_-)^2=+1
$$
{from} $\Theta^2=-1$ and ${\rm sgn}(C_-K)^2=-1$, and 
$$
\{S,\Theta\}\varphi=0
$$
for any wavefunction $\varphi$, from $(S\Theta)^2=+1$ and $\Theta^2=-1$. The latter yields 
$$
\{S,\tilde{\Theta}_-\}\varphi=0. 
$$
Further, we have 
$$
\tilde{\Theta}_-D_a\varphi=-D_a\tilde{\Theta}_-\varphi
$$
{from} ${\rm sgn}_-(D_a)=-1$. 
Therefore, it is sufficient to treat the inner product $\langle \varphi, \tilde{\Theta}_-U\varphi\rangle$ 
in the same way as in the preceding CI class in three dimensions. 
As a result, we cannot obtain any information about $\langle \varphi, \tilde{\Theta}_-U\varphi\rangle$ 
because $(\tilde{\Theta}_-)^2=+1$ which is opposite sign to that for CI class in three dimensions. 
Thus, DIII class in three dimensions has an integer-valued index.  

In passing, we can obtain a relation between the indices of the DIII and AII classes in three dimensions.  
Under a weak perturbation which breaks the chiral symmetry of the DIII class, the DIII class of the model 
changes to AII class. In the weak limit of the perturbation, a relation of the indices between the two classes 
holds. For the details, see Appendix~\ref{RelationIndices}.  
\bigskip\medskip

\noindent
$\bullet$ {\bf CII Class in Three Dimensions.}
In this case, one has 
$$
(\tilde{\Theta}_-)^2=+1
$$
{from} $\Theta^2=-1$ and ${\rm sgn}(C_-K)^2=-1$, and 
$$
[S,\Theta]\varphi=0
$$
for any wavefunction $\varphi$, from $(S\Theta)^2=-1$ and $\Theta^2=-1$. 
Clearly, the commutation relation yields 
$$
[S,\tilde{\Theta}_-]\varphi=0. 
$$
In addition, we have 
$$
\tilde{\Theta}_-D_a\varphi=-D_a\tilde{\Theta}_-\varphi
$$
{from} ${\rm sgn}_-(D_a)=-1$. The last two relations yield 
\begin{align*}
\tilde{\Theta}_-A\varphi&=\tilde{\Theta}_-S(\mathcal{P}_{\rm D}-U\mathcal{P}_{\rm D}U)\varphi\\
&=S\tilde{\Theta}_-(\mathcal{P}_{\rm D}-U\mathcal{P}_{\rm D}U)\varphi\\
&=-S(\mathcal{P}_{\rm D}-U\mathcal{P}_{\rm D}U)\tilde{\Theta}_-\varphi=-A\tilde{\Theta}_-\varphi.
\end{align*}
This implies that the integer-valued index is vanishing. 

In order to check whether the present case shows the $\ze_2$-valued index, it is enough to 
consider the inner product $\langle \tilde{\Theta}_-\varphi,B\varphi\rangle$. 
Note that
$$
\tilde{\Theta}_-B\varphi=\tilde{\Theta}_-S(1-\mathcal{P}_{\rm D}-U\mathcal{P}_{\rm D}U)\varphi
=S\tilde{\Theta}_-(1-\mathcal{P}_{\rm D}-U\mathcal{P}_{\rm D}U)\varphi=-B\tilde{\Theta}_-\varphi,
$$
where we have used $[S,\tilde{\Theta}_-]\varphi=0$ and $\tilde{\Theta}_-D_a\varphi=-D_a\tilde{\Theta}_-\varphi$. 
Combining this with $(\tilde{\Theta}_-)^2=+1$, we obtain 
$$
\langle \tilde{\Theta}_-\tilde{\Theta}_-\varphi,\tilde{\Theta}_-B\varphi\rangle
=\langle \varphi,\tilde{\Theta}_-B\varphi\rangle=-\langle B\varphi,\tilde{\Theta}_-\varphi\rangle.
$$
On the other hand, the left-hand side can be written as 
$$
\langle \tilde{\Theta}_-\tilde{\Theta}_-\varphi,\tilde{\Theta}_-B\varphi\rangle
=\langle B\varphi,\tilde{\Theta}_-\varphi\rangle.
$$
Immediately, 
$$
\langle B\varphi,\tilde{\Theta}_-\varphi\rangle=0.
$$
In consequence, CII class in three dimensions has a $\ze_2$ index. 

\bigskip\medskip

\noindent
$\bullet$ {\bf BDI Class in Five Dimensions.}
In this case, we have 
$$
(\tilde{\Theta}_-)^2=-1
$$
{from} $\Theta^2=+1$ and ${\rm sgn}(C_-K)^2=-1$, and 
$$
[S,\Theta]\varphi=0
$$
for any wavefunction $\varphi$, from $(S\Theta)^2=+1$ and $\Theta^2=+1$. 
Clearly, this yields 
$$
[S,\tilde{\Theta}_-]\varphi=0. 
$$
In addition, we have 
$$
\tilde{\Theta}_-D_a\varphi=D_a\tilde{\Theta}_-\varphi
$$
{from} ${\rm sgn}_-(D_a)=+1$. These conditions are the same as in the case of CII class in one dimension. 
Therefore, the integer-valued index for BDI class in five dimensions is always even, i.e., $2\ze$. 

\bigskip\medskip

\noindent
$\bullet$ {\bf CI Class in Five Dimensions.}
Similarly, one has  
$$
(\tilde{\Theta}_-)^2=-1
$$
{from} $\Theta^2=+1$ and ${\rm sgn}(C_-K)^2=-1$, and 
$$
\{S,\Theta\}\varphi=0
$$
for any wavefunction $\varphi$, from $(S\Theta)^2=-1$ and $\Theta^2=+1$. 
Clearly, this yields 
$$
\{S,\tilde{\Theta}_-\}\varphi=0. 
$$
In addition, we have 
$$
\tilde{\Theta}_-D_a\varphi=D_a\tilde{\Theta}_-\varphi
$$
{from} ${\rm sgn}_-(D_a)=+1$. These conditions are the same as in the case of DIII class in one dimension. 
Thus, CI class in five dimensions has a $\ze_2$ index. 

\bigskip\medskip

\noindent
$\bullet$ {\bf DIII Class in Five Dimensions.}
Similarly, one has  
$$
(\tilde{\Theta}_-)^2=+1
$$
{from} $\Theta^2=-1$ and ${\rm sgn}(C_-K)^2=-1$, and 
$$
\{S,\Theta\}\varphi=0
$$
for any wavefunction $\varphi$, from $(S\Theta)^2=+1$ and $\Theta^2=-1$. 
Clearly, this yields 
$$
\{S,\tilde{\Theta}_-\}\varphi=0. 
$$
In addition, we have 
$$
\tilde{\Theta}_-D_a\varphi=D_a\tilde{\Theta}_-\varphi
$$
{from} ${\rm sgn}_-(D_a)=+1$. These conditions are the same as in the case of CI class in one dimension. 
Therefore, DIII class in five dimensions has no index. 

\bigskip\medskip

\noindent
$\bullet$ {\bf CII Class in Five Dimensions.}
Similarly, one has  
$$
(\tilde{\Theta}_-)^2=+1
$$
{from} $\Theta^2=-1$ and ${\rm sgn}(C_-K)^2=-1$, and 
$$
[S,\Theta]\varphi=0
$$
for any wavefunction $\varphi$, from $(S\Theta)^2=-1$ and $\Theta^2=-1$. 
Clearly, this yields 
$$
[S,\tilde{\Theta}_-]\varphi=0. 
$$
In addition, we have 
$$
\tilde{\Theta}_-D_a\varphi=D_a\tilde{\Theta}_-\varphi
$$
{from} ${\rm sgn}_-(D_a)=+1$. These conditions are the same as in the case of BDI class in one dimension. 
In consequence, CII class in five dimensions has an integer-valued index, i.e., $\ze$. 

\bigskip\medskip

\noindent
$\bullet$ {\bf BDI Class in Seven Dimensions.}
In this case, we have 
$$
(\tilde{\Theta}_-)^2=+1
$$
{from} $\Theta^2=+1$ and ${\rm sgn}(C_-K)^2=+1$, and 
$$
[S,\Theta]\varphi=0
$$
for any wavefunction $\varphi$, from $(S\Theta)^2=+1$ and $\Theta^2=+1$. 
Clearly, this yields 
$$
[S,\tilde{\Theta}_-]\varphi=0. 
$$
In addition, we have 
$$
\tilde{\Theta}_-D_a\varphi=-D_a\tilde{\Theta}_-\varphi
$$
{from} ${\rm sgn}_-(D_a)=-1$. These conditions are the same as in the case of CII class in three dimensions. 
Therefore, BDI class in seven dimensions has a $\ze_2$ index. 

\bigskip\medskip

\noindent
$\bullet$ {\bf CI Class in Seven Dimensions.}
Similarly, we have 
$$
(\tilde{\Theta}_-)^2=+1
$$
{from} $\Theta^2=+1$ and ${\rm sgn}(C_-K)^2=+1$, and 
$$
\{S,\Theta\}\varphi=0
$$
for any wavefunction $\varphi$, from $(S\Theta)^2=-1$ and $\Theta^2=+1$. 
Clearly, this yields 
$$
\{S,\tilde{\Theta}_-\}\varphi=0. 
$$
In addition, we have 
$$
\tilde{\Theta}_-D_a\varphi=-D_a\tilde{\Theta}_-\varphi
$$
{from} ${\rm sgn}_-(D_a)=-1$. These conditions are the same as in the case of DIII class in three dimensions. 
Thus, CI class in seven dimensions has an integer-valued index, i.e., $\ze$.  

\bigskip\medskip

\noindent
$\bullet$ {\bf DIII Class in Seven Dimensions.}
Similarly, we obtain  
$$
(\tilde{\Theta}_-)^2=-1
$$
{from} $\Theta^2=-1$ and ${\rm sgn}(C_-K)^2=+1$, and 
$$
\{S,\Theta\}\varphi=0
$$
for any wavefunction $\varphi$, from $(S\Theta)^2=+1$ and $\Theta^2=-1$. 
Clearly, this yields 
$$
\{S,\tilde{\Theta}_-\}\varphi=0. 
$$
In addition, we have 
$$
\tilde{\Theta}_-D_a\varphi=-D_a\tilde{\Theta}_-\varphi
$$
{from} ${\rm sgn}_-(D_a)=-1$. These conditions are the same as in the case of CI class in three dimensions. 
Therefore, the integer-valued index for DIII class in five dimensions is always even, i.e., $2\ze$. 

\bigskip\medskip

\noindent
$\bullet$ {\bf CII Class in Seven Dimensions.}
Similarly, one has 
$$
(\tilde{\Theta}_-)^2=-1
$$
{from} $\Theta^2=-1$ and ${\rm sgn}(C_-K)^2=+1$, and 
$$
[S,\Theta]\varphi=0
$$
for any wavefunction $\varphi$, from $(S\Theta)^2=-1$ and $\Theta^2=-1$. 
Clearly, this yields 
$$
[S,\tilde{\Theta}_-]\varphi=0. 
$$
In addition, we have 
$$
\tilde{\Theta}_-D_a\varphi=-D_a\tilde{\Theta}_-\varphi
$$
{from} ${\rm sgn}_-(D_a)=-1$. These conditions are the same as in the case of BDI class in three dimensions. 
In consequence, CII class in seven dimensions has no index. 

\section{Homotopy argument}
\label{HomoArg}

In this section, we prove that 
the nonzero eigenvalue $\lambda$ of the operator $A$ of (\ref{AEvenD}) in the case of even dimensions 
is continuous with respect to the norm of perturbations. This leads to the robustness of the indices under perturbations 
because the indices are given by the multiplicities of the eigenvalues $\lambda=\pm 1$ of $A$. 
The case of odd dimensions can be treated in the same way. 

Consider a perturbation $\delta H$ for the Hamiltonian. We assume that $\delta H$ has only finite-range 
hopping terms for simplicity although we can treat long-range hoppings which rapidly decay with large distance. 
The total Hamiltonian $H'$ is given by 
$$
H'=H+g\delta H
$$
with a small real parameter $g$. For this Hamiltonian $H'$, the operator $A'$ which corresponds to $A$ of (\ref{AEvenD}) 
is given by 
$$
A'=\gamma^{(2n+1)}(P_{\rm F}'-D_aP_{\rm F}'D_a).
$$
The difference between $A$ and $A'$ is written as  
$$
A'-A=\gamma^{(2n+1)}\left[(P_{\rm F}'-P_{\rm F})-D_a(P_{\rm F}'-P_{\rm F})D_a\right].
$$
Thus, in order to prove the continuity under the perturbation, it is enough to estimate $P_{\rm F}'-P_{\rm F}$.  

Using the contour-integral expression (\ref{PFcontour}) of the projection $P_{\rm F}'$ 
onto the Fermi sea for the Hamiltonian $H'$, one has  
\begin{equation}
\label{diffPF'PF}
P_{\rm F}'-P_{\rm F}=\frac{1}{2\pi i}\oint dz \left[\frac{1}{z-H'}-\frac{1}{z-H}\right]
=\frac{g}{2\pi i}\oint dz\frac{1}{z-H'}\delta H\frac{1}{z-H}.
\end{equation}
We write 
$$
\psi=\frac{1}{z-H'}\delta H\frac{1}{z-H}\varphi
$$
with a wavefunction $\varphi\in\ell^2(\ze^d,\co^M)$. The norm $\Vert \psi\Vert$ is written as 
$$
\Vert \psi\Vert^2=\left\langle\varphi,\frac{1}{\overline{z}-H'}\delta H\frac{1}{\overline{z}-H}
\frac{1}{z-H'}\delta H\frac{1}{z-H}\varphi\right\rangle.  
$$
By using the basis of (\ref{ONSzeta}) in Sec.~\ref{TopoInvEven}, one has 
\begin{multline*}
\left\langle \zeta_u^\alpha,\frac{1}{z-H'}\delta H\frac{1}{z-H}\zeta_v^\beta\right\rangle\\
=\sum_{w,\mu}\sum_{w',\mu'}
\left\langle \zeta_u^\alpha,\frac{1}{z-H'}\zeta_w^\mu\right\rangle 
\left\langle \zeta_w^\mu,\delta H\zeta_{w'}^{\mu'}\right\rangle 
\left\langle\zeta_{w'}^{\mu'},\frac{1}{z-H}\zeta_v^\beta\right\rangle. 
\end{multline*}
In order to evaluate the right-hand side, we recall Assumption~\ref{Assumption}, which yields  
$$
\left|\left\langle \zeta_u^\alpha,\frac{1}{z-H}\zeta_v^\beta\right\rangle\right|\le {\rm Const.}
e^{-\kappa|u-v|}
$$
with some positive constant $\kappa$. Combining this with the assumption that $\delta H$ has only finite-range 
hopping terms, we have 
\begin{align*}
\left|\left\langle \zeta_u^\alpha,\frac{1}{z-H'}\delta H\frac{1}{z-H}\zeta_v^\beta\right\rangle\right|
&\le {\rm Const.}\sum_{w,w':|w-w'|\le r_0}e^{-\kappa|u-w|}e^{-\kappa|w'-v|}\\
&\le {\rm Const.}e^{-\kappa'|u-v|}
\end{align*}
with some constants, $r_0>0$ and $\kappa'$ satisfying $0<\kappa'<\kappa$. 
Using this estimate for $\Vert\psi\Vert^2$, we obtain 
\begin{align*}
\Vert \psi\Vert^2&\le {\rm Const.}
\sum_{u,\alpha}\sum_{v,\beta}
\left|\langle\varphi,\zeta_u^\alpha\rangle\right| |\langle \zeta_v^\beta,\varphi\rangle|e^{-\kappa''|u-v|}\\
&\le {\rm Const.}
\sum_{u,\alpha}\sum_{v,\beta}
\left[\left|\langle\varphi,\zeta_u^\alpha\rangle\right|^2+|\langle \zeta_v^\beta,\varphi\rangle|^2\right]e^{-\kappa''|u-v|}\\
&={\rm Const.}\left[\sum_{u,\alpha}\left|\langle\varphi,\zeta_u^\alpha\rangle\right|^2\sum_{v,\beta}e^{-\kappa''|u-v|}
+\sum_{v,\beta}|\langle \zeta_v^\beta,\varphi\rangle|^2\sum_{u,\alpha}e^{-\kappa''|u-v|}\right]\\
&\le {\rm Const.}\Vert\varphi\Vert^2,
\end{align*}
where $\kappa''$ is some positive constant. This implies that the operator in the integrand in the right-hand side of 
(\ref{diffPF'PF}) is bounded. Therefore, we obtain 
$$
\Vert P_{\rm F}'-P_{\rm F}\Vert\le {\rm Const.}g.
$$
This yields 
$$
\Vert A'-A\Vert\le{\rm Const.}g.
$$
Namely, the operator $A$ is continuous with respect to the norm of the perturbation $g\delta H$.  
Combining this with the min-max principle \cite{RSIV}, we can obtain the desired result that 
the nonzero eigenvalue $\lambda$ of the operator $A$ is continuous with respect to the norm of the perturbation $g\delta H$.

\appendix
\section{A Dirac-Clifford Garden} 
\label{AppGamma}

In this appendix, we give a concrete expression of the gamma matrices, and 
the corresponding time-reversal operators. 

We use the standard convention for the Pauli matrices as  
$$
\sigma_0=\left(\begin{matrix}
1 & 0\\ 0 & 1 
\end{matrix}\right),\quad
\sigma_1=\left(\begin{matrix}
0 & 1\\ 1 & 0 
\end{matrix}\right),\quad 
\sigma_2=\left(\begin{matrix}
0 & -i\\ i & 0 
\end{matrix}\right),\quad
\sigma_3=\left(\begin{matrix}
1 & 0\\ 0 & -1 
\end{matrix}\right).
$$
By using these Pauli matrices, one can define the gamma matrices $\gamma=(\gamma^{(1)},\ldots,\gamma^{(2n+1)})$ as 
$$
\gamma^{(2m-1)}=\underbrace{\sigma_0\otimes\cdots\otimes\sigma_0}_{m-1}\otimes\sigma_1\otimes
\underbrace{\sigma_3\otimes\cdots\otimes\sigma_3}_{n-m}
$$
and
$$
\gamma^{(2m)}=\underbrace{\sigma_0\otimes\cdots\otimes\sigma_0}_{m-1}\otimes\sigma_2\otimes
\underbrace{\sigma_3\otimes\cdots\otimes\sigma_3}_{n-m}
$$
for $m=1,2,\ldots,n$, and 
$$
\gamma^{(2n+1)}=\underbrace{\sigma_3\otimes\cdots\otimes\sigma_3}_n.
$$

We introduce two time-reversal operators, $C_+$ and $C_-$ for the gamma matrices as 
\begin{equation}
\label{C+}
C_+:=\begin{cases} 
\overbrace{\sigma_1\otimes\sigma_2\otimes\sigma_1\otimes\sigma_2\otimes\cdots\otimes\sigma_1\otimes\sigma_2}^n, 
& \text{$n=$ even};\\
\underbrace{\sigma_1\otimes\sigma_2\otimes\sigma_1\otimes\sigma_2\otimes\cdots\otimes\sigma_1\otimes\sigma_2
\otimes\sigma_1}_n, 
& \text{$n=$ odd},
\end{cases}
\end{equation}
and
\begin{equation}
C_-:=\begin{cases} 
\overbrace{\sigma_2\otimes\sigma_1\otimes\sigma_2\otimes\sigma_1\otimes\cdots\otimes\sigma_2\otimes\sigma_1}^n, 
& \text{$n=$ even};\\
\underbrace{\sigma_2\otimes\sigma_1\otimes\sigma_2\otimes\sigma_1\otimes\cdots
\otimes\sigma_2\otimes\sigma_1\otimes\sigma_2}_n, 
& \text{$n=$ odd}.
\end{cases}
\end{equation}
Then, one can easily show that the time-reversal transformations for the gamma matrices are given by  
\begin{equation}
\label{C+gammaC+}
C_+\gamma^{(j)}C_+=(-1)^{n+1}\overline{\gamma^{(j)}},\quad \mbox{for \ }j=1,2,\ldots,2n,
\end{equation}
$$
C_+\gamma^{(2n+1)}C_+=(-1)^n\overline{\gamma^{(2n+1)}}
$$
and 
\begin{equation}
\label{C-gammaC-}
C_-\gamma^{(j)}C_-=(-1)^n\overline{\gamma^{(j)}},\quad \mbox{for \ }j=1,2,\ldots,2n,2n+1,
\end{equation}
where $\overline{\cdots}$ stands for the complex conjugate. In particular, one has 
\begin{equation}
\label{Cpmgamma(2n+1)Cpm}
C_\pm\gamma^{(2n+1)}C_\pm=(-1)^n\gamma^{(2n+1)}
\end{equation}
because $\overline{\gamma^{(2n+1)}}=\gamma^{(2n+1)}$ by the definition of $\gamma^{(2n+1)}$. 
By definition, one has 
$$
C_+C_-=\gamma^{(2n+1)}\times \begin{cases} 1, & \text{$n=$ even};\\
i, & \text{$n=$ odd}.  
\end{cases}
$$
Thus, the two time-reversal transformation, $C_\pm$, are not independent of each other.

\section{Trace Class and Trace $p$-Norm}
\label{traceclassA}

We show that the operator $(P_{\rm F}-D_aP_{\rm F}D_a)^{2n+1}$ is trace class \cite{RSI}. 

To begin with, we introduce a complete orthonormal system of wavefunctions,
\begin{equation}
\label{tildezetabasis}
\tilde{\zeta}_u^{\alpha,\mu}:=\chi_{\{u\}}\otimes\Phi^\alpha\otimes\Psi^\mu, \quad 
u\in\ze^d,\ \alpha=1,2,\ldots,d_{\rm s},\ \mu=1,2,\ldots,d_\gamma, 
\end{equation}
where $\chi_{\{u \}}$ and $\Phi^\alpha$ are the wavefunctions which are given in (\ref{ONSzeta}), 
and the wavefunction $\Psi^\mu$ for the gamma matrices are an orthonormal basis 
whose dimension of the Hilbert space is given by $d_\gamma$. 
We write
$$
T=P_{\rm F}-D_aP_{\rm F}D_a
$$
for short. The matrix elements are given by 
$$
T(u,\alpha,\mu;v,\beta,\nu):=\langle\tilde{\zeta}_u^{\alpha,\mu},T\tilde{\zeta}_v^{\beta,\nu}\rangle.
$$

Following Aizenman and Graf \cite{AG}, we introduce 
$$
T^{(b,\eta,\xi)}(u,\alpha,\mu;v,\beta,\nu):=T(u,\alpha,\mu;v,\beta,\nu)\delta_{u-b,v}
\delta_{\alpha-\eta,\beta}^{(d_{\rm s})}
\delta_{\mu-\xi,\nu}^{(d_\gamma)}
$$
for $b\in\ze^d$, $\eta\in\{1,2,\ldots,d_{\rm s}\}$ and $\xi\in\{1,2,\ldots,d_\gamma\}$, 
where $\delta_{u,v}$ is the usual Kronecker delta for $u,v\in\ze^d$, and the rest are defined as  
$$
\delta_{\alpha,\beta}^{(d_{\rm s})}:=\begin{cases} 1, & \text{$\alpha=\beta$ modulo $d_{\rm s}$};\\
0, & \text{otherwise},
\end{cases}
$$
and 
$$
\delta_{\mu,\nu}^{(d_\gamma)}:=\begin{cases} 1, & \text{$\mu=\nu$ modulo $d_\gamma$};\\
0, & \text{otherwise}.
\end{cases}
$$
Clearly, one has 
\begin{equation}
\label{Tsumexpres}
T(u,\alpha,\mu;v,\beta,\nu)=\sum_{b,\eta,\xi}T^{(b,\eta,\xi)}(u,\alpha,\mu;v,\beta,\nu).
\end{equation}

Note that
\begin{align}
&|T^{(b,\eta,\xi)}|^2(u,\alpha,\mu;v,\beta,\nu)\nonumber\\ \nonumber
&=\sum_{w,\rho,\theta}
T^{(b,\eta,\xi)}(w,\rho,\theta;u,\alpha,\mu)^\ast T^{(b,\eta,\xi)}(w,\rho,\theta;v,\beta,\nu)\\ \nonumber
&=\sum_{w,\rho,\theta}T(w,\rho,\theta;u,\alpha,\mu)^\ast\delta_{w-b,u}\delta_{\rho-\eta,\alpha}^{(d_{\rm s})}
\delta_{\theta-\xi,\mu}^{(d_\gamma)}\\ \nonumber
&\times T(w,\rho,\theta;v,\beta,\nu)\delta_{w-b,v}\delta_{\rho-\eta,\beta}^{(d_{\rm s})}
\delta_{\theta-\xi,\nu}^{(d_\gamma)}\\  
&=|T(u+b,\eta+\alpha,\xi+\mu;u,\alpha,\mu)|^2\delta_{u,v}\delta_{\alpha,\beta}^{(d_{\rm s})}\delta_{\mu,\nu}^{(d_\gamma)}. 
\label{Tsqurediagonal}
\end{align}
Clearly, the right-hand side is diagonal in the present basis. 

For $p$ satisfying $1\le p<\infty$, the trace $p$-norm of an operator $\mathcal{A}$ is defined by 
$$
\Vert \mathcal{A}\Vert_p:=\left({\rm Tr}\; |\mathcal{A}|^p\right)^{1/p}.
$$
For two operators, $\mathcal{A}$ and $\mathcal{B}$, the Minkowski inequality \cite{RSII} holds as 
$$
\Vert \mathcal{A}+\mathcal{B}\Vert_p\le \Vert \mathcal{A}\Vert_p+\Vert \mathcal{B}\Vert_p.
$$
Set $p=2n+1$. 
Since the operator $T=P_{\rm F}-D_aP_{\rm F}D_a$ is self-adjoint, one has 
\[
|T^p|=\sqrt{(T^p)^\ast T^p}=\sqrt{T^{2p}}=\sqrt{|T|^{2p}}=|T|^p.
\]
Therefore, in order to prove that the operator $T^P$ is trace class, it is sufficient to show 
that the trace $p$-norm $\Vert T\Vert_p$ is bounded. 

{From} the expression (\ref{Tsumexpres}) and the Minkowski inequality, we have 
\begin{equation}
\label{Tpnormbound}
\Vert T\Vert_p\le \sum_{b,\eta,\xi}\Vert T^{(b,\eta,\xi)}\Vert_p.
\end{equation}
As shown in (\ref{Tsqurediagonal}), the operator $|T^{(b,\eta,\xi)}|$ is diagonal. Therefore, one has
\begin{equation}
\label{Tbetaxibound}
\Vert T^{(b,\eta,\xi)}\Vert_p^p= \sum_{u,\alpha,\mu}|T(u+b,\eta+\alpha,\xi+\mu;u,\alpha,\mu)|^p. 
\end{equation}
Thus, it is enough to estimate the matrix element $T(u+b,\eta+\alpha,\xi+\mu;u,\alpha,\mu)$. 

Since $D_a^2=1$, one has 
\begin{equation}
\label{TcommuDaPF}
T=P_{\rm F}-D_aP_{\rm F}D_a=D_a(D_aP_{\rm F}-P_{\rm F}D_a)=D_a[D_a,P_{\rm F}].
\end{equation}
Note that 
\begin{align}
\langle\tilde{\zeta}_{u}^{\alpha,\mu},[P_{\rm F},D_a]\tilde{\zeta}_{v}^{\beta,\nu}\rangle&=
\langle \tilde{\zeta}_{u}^{\alpha,\mu},P_{\rm F}\tilde{\zeta}_{v}^{\beta,\mu}\rangle 
\langle \tilde{\zeta}_{v}^{\beta,\mu},D_a\tilde{\zeta}_{v}^{\beta,\nu}\rangle\nonumber\\
&-\langle \tilde{\zeta}_{u}^{\alpha,\mu},D_a\tilde{\zeta}_{u}^{\alpha,\nu}\rangle
\langle \tilde{\zeta}_{u}^{\alpha,\nu},P_{\rm F}\tilde{\zeta}_{v}^{\beta,\nu}\rangle\nonumber\\
&=\langle \tilde{\zeta}_{u}^{\alpha,\mu},P_{\rm F},\tilde{\zeta}_{v}^{\beta,\mu}\rangle
\left[\langle\tilde{\zeta}_{v}^{\alpha,\mu},D_a\tilde{\zeta}_{v}^{\alpha,\nu} \rangle 
-\langle\tilde{\zeta}_{u}^{\alpha,\mu},D_a\tilde{\zeta}_{u}^{\alpha,\nu} \rangle\right],
\label{commuPD}
\end{align}
where we have the properties of the basis (\ref{tildezetabasis}):    
$$
\langle \tilde{\zeta}_u^{\alpha,\mu},P_{\rm F}\tilde{\zeta}_{u'}^{\beta,\mu}\rangle
=\langle \tilde{\zeta}_u^{\alpha,\nu},P_{\rm F}\tilde{\zeta}_{u'}^{\beta,\nu}\rangle
\quad \mbox{for all\ } \mu,\nu, 
$$
and 
$$
\langle \tilde{\zeta}_u^{\alpha,\mu},D_a\tilde{\zeta}_u^{\alpha,\nu}\rangle
=\langle \tilde{\zeta}_u^{\beta,\mu},D_a\tilde{\zeta}_u^{\beta,\nu}\rangle
\quad \mbox{for all \ } \alpha, \beta.
$$
The difference between two matrix elements for the Dirac operators $D_a$ in the right-hand side 
is evaluated as: 

\begin{lem}
The following bound is valid: 
\begin{equation}
\label{diffDabound}
\left|\langle \tilde{\zeta}_u^{\alpha,\mu}D_a\tilde{\zeta}_u^{\alpha,\nu}\rangle 
-\langle\tilde{\zeta}_v^{\alpha,\mu},D_a\tilde{\zeta}_v^{\alpha,\nu}\rangle \right|
\le \frac{4|u-v|}{|u-a|}.
\end{equation}
\end{lem}

\begin{proof}
Note that 
\begin{align}
D_a(u)-D_a(v)&=\frac{1}{|u-a|}(u-a)\cdot \gamma-\frac{1}{|v-a|}(v-a)\cdot\gamma \nonumber\\
&=\frac{1}{|u-a|}(u-a)\cdot \gamma-\frac{1}{|u-a|}(v-a)\cdot\gamma \nonumber\\
&+\frac{1}{|u-a|}(v-a)\cdot\gamma-\frac{1}{|v-a|}(v-a)\cdot\gamma \nonumber\\
&=\frac{1}{|u-a|}(u-v)\cdot\gamma+\left(\frac{1}{|u-a|}-\frac{1}{|v-a|}\right)(v-a)\cdot\gamma.
\label{diffDauv}
\end{align}
As to the second term in the last line, one has 
\begin{align*}
\frac{1}{|u-a|}-\frac{1}{|v-a|}&=\frac{|v-a|-|u-a|}{|u-a||v-a|}\\
&=-\frac{(u-v)^2+2(v-a)\cdot(u-v)}{|u-a||v-a|(|u-a|+|v-a|)}.
\end{align*}
Therefore, we have 
$$
\left|\frac{1}{|u-a|}-\frac{1}{|v-a|}\right|
\le \frac{3|u-v|}{|u-a||v-a|},
$$
where we have used $|u-a|+|v-a|\ge|u-v|$.
Combining this with (\ref{diffDauv}), we obtain the desired bound (\ref{diffDabound}). 
\end{proof}

{From} Assumption~\ref{Assumption} for the resolvent $(E_{\rm F}-H)^{-1}$, we have   
\begin{equation}
\label{decaymatPF}
\left|\langle\tilde{\zeta}_u^{\alpha,\mu},P_{\rm F}\tilde{\zeta}_v^{\beta,\mu}\rangle \right|
\le {\rm Const.}e^{-\kappa|u-v|}
\end{equation}
with some positive constant $\kappa$. Combining this, (\ref{commuPD}) and (\ref{diffDabound}), we obtain 
\begin{equation}
\label{decayboundcommuPFDa}
\left|\langle\tilde{\zeta}_{u}^{\alpha,\mu},[P_{\rm F},D_a]\tilde{\zeta}_{v}^{\beta,\nu}\rangle\right|
\le {\rm Const.}e^{-\kappa|u-v|}\min\{2, {4|u-v|}/{|u-a|}\}.
\end{equation}
Since one has 
\begin{multline*}
T(u+b,\eta+\alpha,\xi+\mu;u,\alpha,\mu)
=\langle\tilde{\zeta}_{u+b}^{\eta+\alpha,\xi+\mu},T\tilde{\zeta}_u^{\alpha,\mu}\rangle\\
=\sum_\nu\langle\tilde{\zeta}_{u+b}^{\eta+\alpha,\xi+\mu},D_a\tilde{\zeta}_{u+b}^{\eta+\alpha,\nu}\rangle
\langle\tilde{\zeta}_{u+b}^{\eta+\alpha,\nu},[D_a,P_{\rm F}]\tilde{\zeta}_u^{\alpha,\mu}\rangle
\end{multline*}
from (\ref{TcommuDaPF}), we have 
$$
|T(u+b,\eta+\alpha,\xi+\mu;u,\alpha,\mu)|^p\le 
{\rm Const.}e^{-p\kappa|b|}[\min\{2, {4|b|}/{|u+b-a|}\}]^p
$$
by using the bound (\ref{decayboundcommuPFDa}). 
Combining this, (\ref{Tpnormbound}) and (\ref{Tbetaxibound}), we obtain the desired result: 
\begin{align*}
\Vert T\Vert_p&\le \sum_{b,\eta,\xi}\left(\sum_{u,\alpha,\mu}|T(u+b,\eta+\alpha,\xi+\mu;u,\alpha,\mu)|^p\right)^{1/p}\\
&\le {\rm Const.}\sum_b e^{-\kappa|b|}\times\left(\sum_u\; [\min\{2, {4|b|}/{|u+b-a|}\}]^p\right)^{1/p}\\
&\le \sum_b({\rm Const.}+{\rm Const.}|b|)e^{-\kappa|b|}<\infty,
\end{align*}
where we have used $p=2n+1>d=2n$ for showing that the sum about $u$ is finite.

\section{Proof of Lemma~\ref{lem:DlimtROmegaInd}}
\label{Appendix:DlimtROmegaInd}

In this appendix, we give a proof of Lemma~\ref{lem:DlimtROmegaInd}. Namely, we prove that 
the approximate index converges to the index as 
\begin{equation}
\label{IndIDLemma6}
\lim_{R\nearrow\infty}\lim_{\Omega\nearrow\re^d}{\rm Ind}^{(2n)}(D_a,P_{\rm F};\Omega,R)
={\rm Ind}^{(2n)}(D_a,P_{\rm F}).
\end{equation}
  
To begin with, we note that  
$$
P_{\rm F}-D_aP_{\rm F}D_a=[P_{\rm F},D_a]D_a=-D_a[P_{\rm F},D_a].
$$
because $D_a^2=1$. Immediately,  
\begin{align*}
(P_{\rm F}-D_aP_{\rm F}D_a)^{2n+1}&=-D_a[P_{\rm F},D_a](-1)^n[P_{\rm F},D_a]^{2n}\\
&=(-1)^{n+1}D_a[P_{\rm F},D_a]^{2n+1}.
\end{align*}
Using this identity, we have 
\begin{align}
&{\rm Tr}\> \gamma^{(2n+1)}(P_{\rm F}-D_aP_{\rm F}D_a)^{2n+1}(1-\chi_R^a)\nonumber\\
&=(-1)^{n+1}\>{\rm Tr}\> \gamma^{(2n+1)}D_a[P_{\rm F},D_a]^{2n+1}(1-\chi_R^a)\nonumber\\
&=(-1)^{n+1}\sum_{u_1,\ldots,u_{2n+1}}\sum_{\alpha_1,\ldots,\alpha_{2n+1}}
\sum_{\mu_1,\ldots,\mu_{2n+1}}\langle\tilde{\zeta}_{u_1}^{\alpha_1,\mu_1},\gamma^{(2n+1)}D_a[P_{\rm F},D_a]
\tilde{\zeta}_{u_2}^{\alpha_2,\mu_2}\rangle\nonumber\\
&\times\langle\tilde{\zeta}_{u_2}^{\alpha_2,\mu_2},[P_{\rm F},D_a]\tilde{\zeta}_{u_3}^{\alpha_3,\mu_3}\rangle 
\cdots \langle\tilde{\zeta}_{u_{2n}}^{\alpha_{2n},\mu_{2n}},[P_{\rm F},D_a]
\tilde{\zeta}_{u_{2n+1}}^{\alpha_{2n+1},\mu_{2n+1}}\rangle\nonumber\\
&\times\langle \tilde{\zeta}_{u_{2n+1}}^{\alpha_{2n+1},\mu_{2n+1}},[P_{\rm F},D_a]
\tilde{\zeta}_{u_1}^{\alpha_1,\mu_1}\rangle 
(1-\chi_R^a)(u_1), 
\label{Trsum}
\end{align}
where $\tilde{\zeta}_u^{\alpha,\mu}$ are the complete orthonormal basis of (\ref{tildezetabasis}) 
in the preceding Appendix~\ref{traceclassA}.  

\begin{lem}
\label{lem:TrAnRbound}
The following bound is valid: 
\begin{equation}
\left|{\rm Tr}\> \gamma^{(2n+1)}(P_{\rm F}-D_aP_{\rm F}D_a)^{2n+1}(1-\chi_R^a)\right|\le 
\frac{{\rm Const.}}{R},
\end{equation}
where the constant does not depend on $a$. 
\end{lem}

\begin{proof}
In order to prove the statement, we evaluate the right-hand side in the second equality of (\ref{Trsum}). 
We write $r=|u_1-a|$ for short. 
Consider the contribution such that, for $\ell\in\{2,3,\ldots,2n+1\}$, 
$u_\ell$ satisfies $|u_\ell-a|\le r/2$ 
in the first sum in the right-hand side in the second equality of (\ref{Trsum}).
We write $I_{\rm in}^{(\ell)}$ for the corresponding contribution. 
The contribution is estimated as  
\begin{align*}
\left|I_{\rm in}^{(\ell)}\right|&\le {\rm Const.}
\sum_{\substack{u_1:\\ |u_1-a|\ge R}} \sum_{u_2}\cdots \sum_{\substack{u_\ell:\\ |u_\ell-a|\le r/2}}
\cdots \sum_{u_{2n+1}}e^{-\kappa|u_1-u_2|}e^{-\kappa|u_2-u_3|}\\
&\qquad\qquad\cdots e^{-\kappa|u_{2n}-u_{2n+1}|}e^{-\kappa|u_{2n+1}-u_1|}
\end{align*}
from (\ref{commuPD}) and  (\ref{decaymatPF}). 
Further, by using $r=|u_1-a|$ and 
$$
r/2\le |u_1-a|-|u_\ell-a|\le |u_1-u_\ell|\le |u_1-u_2|+\cdots+|u_{\ell-1}-u_\ell|, 
$$ 
we obtain  
\begin{align*}
\left|I_{\rm in}^{(\ell)}\right|&\le{\rm Const.}
\sum_{\substack{u_1:\\ |u_1-a|\ge R}} \sum_{u_2}\cdots \sum_{\substack{u_\ell:\\ |u_\ell-a|\le r/2}}
\cdots \sum_{u_{2n+1}}e^{-\kappa|u_1-u_\ell|/2}e^{-\kappa|u_2-u_3|/2}\\
&\qquad\cdots e^{-\kappa|u_{2n}-u_{2n+1}|/2}e^{-\kappa|u_{2n+1}-u_1|/2}\\
&\le {\rm Const.}\sum_{\substack{u_1:\\ |u_1-a|\ge R}}e^{-\kappa r/4}\\
&\le {\rm Const.}\sum_{\substack{u_1:\\ |u_1-a|\ge R}}e^{-\kappa|u_1-a|/4}
\le {\rm Const.}\exp[-\kappa' R]
\end{align*}
with some positive constant $\kappa'$. Thus, the contributions such that, at least,  
one of $u_\ell$, $\ell\ne 1$, satisfies $|u_\ell-a|\le r/2$ in the sum exponentially decays in a large $R$. 

For the rest of the contributions, we write $I_{\rm out}$. 
Combining (\ref{commuPD}), (\ref{diffDabound}) and (\ref{decaymatPF}), we have  
\begin{align*}
&|I_{\rm out}|\\
&\le {\rm Const.}\sum_{\substack{u_1:\\ |u_1-a|\ge R}}\sum_{\substack{u_2:\\ |u_2-a|\ge r/2}}\cdots
\sum_{\substack{u_{2n+1}:\\ |u_{2n+1}-a|\ge r/2}}
\frac{1}{|u_1-a||u_2-a|\cdots|u_{2n+1}-a|}\\
&\times e^{-\tilde{\kappa}|u_1-u_2|}e^{-\tilde{\kappa}|u_2-u_3|}\cdots e^{-\tilde{\kappa}|u_{2n}-u_{2n+1}|}
e^{-\tilde{\kappa}|u_{2n+1}-u_1|}\\
&\le {\rm Const.}\sum_{\substack{u_1:\\ |u_1-a|\ge R}}\frac{1}{|u_1-a|^{2n+1}}
\sum_{\substack{u_2:\\ |u_2-a|\ge r/2}}\cdots
\sum_{\substack{u_{2n+1}:\\ |u_{2n+1}-a|\ge r/2}}
e^{-\tilde{\kappa}|u_2-u_3|}\\ 
&\cdots e^{-\tilde{\kappa}|u_{2n}-u_{2n+1}|}e^{-\tilde{\kappa}|u_{2n+1}-u_1|}\\
&\le \frac{{\rm Const.}}{R},
\end{align*}
where $\tilde{\kappa}$ is some positive constant. Combining these estimates, we obtain 
$$
\left|{\rm Tr}\> \gamma^{(2n+1)}(P_{\rm F}-D_aP_{\rm F}D_a)^{2n+1}(1-\chi_R^a)\right|
\le \frac{{\rm Const.}}{R}+{\rm Const.}\exp[-\kappa' R]. 
$$
\end{proof}  

Now we prove Lemma~\ref{lem:DlimtROmegaInd}. 
Note that 
\begin{multline}
\label{IndRIndinfty}
{\rm Ind}^{(2n)}(D_a,P_{\rm F};\Omega,R)
={\rm Ind}^{(2n)}(D_a,P_{\rm F})\\
+\frac{1}{2|\Omega|}\int_\Omega dv(a)\>{\rm Tr}\>\gamma^{(2n+1)}(P_{\rm F}-D_aP_{\rm F}D_a)^{2n+1}(\chi_R^a-1).
\end{multline}
{From} Lemma~\ref{lem:TrAnRbound}, the absolute value of the second term in the right-hand side is 
bounded by ${\rm Const.}/R$. Therefore, we obtain the desired result (\ref{IndIDLemma6}).

\section{Proof of Lemma~\ref{lem:IndtildeInd}}
\label{proof:lem:IndtildeInd}

In this appendix, we prove the relation between the two approximate indices 
in Lemma~\ref{lem:IndtildeInd}. 

The approximate index of (\ref{IndDaPFOmegaR}) is written 
\begin{multline}
\label{AppIndSplt}
{\rm Ind}^{(2n)}(D_a,P_{\rm F};\Omega,R)
=\frac{1}{2|\Omega|}\int_\Omega dv(a) 
\;{\rm Tr}\;\gamma^{(2n+1)}(P_{\rm F}-D_aP_{\rm F}D_a)^{2n+1}\chi_R^a\\
=\frac{1}{2|\Omega|}\int_\Omega dv(a) 
\;{\rm Tr}\;\gamma^{(2n+1)}(P_{\rm F}-D_aP_{\rm F}D_a)^{2n+1}\chi_R^a\chi_\Lambda\\
+\frac{1}{2|\Omega|}\int_\Omega dv(a) 
\>{\rm Tr}\>\gamma^{(2n+1)}(P_{\rm F}-D_aP_{\rm F}D_a)^{2n+1}\chi_R^a(1-\chi_\Lambda). 
\end{multline}
On the other hand, the approximate index of (\ref{IndDaPFLambdaR}) is written 
\begin{align}
\widetilde{{\rm Ind}}^{(2n)}(D_a,P_{\rm F};\Lambda,R)&=
\frac{1}{2|\Lambda|}\int_{\re^d}dv(a)\>{\rm Tr}\> \gamma^{(2n+1)}
(P_{\rm F}-D_aP_{\rm F}D_a)^{2n+1}\chi_R^a\chi_\Lambda \nonumber \\ 
&=\frac{1}{2|\Lambda|}\int_{\Omega}dv(a)\>{\rm Tr}\> \gamma^{(2n+1)}
(P_{\rm F}-D_aP_{\rm F}D_a)^{2n+1}\chi_R^a\chi_\Lambda\nonumber\\ 
+&\frac{1}{2|\Lambda|}\int_{\re^d\backslash\Omega}dv(a)\>{\rm Tr}\>\gamma^{(2n+1)}
(P_{\rm F}-D_aP_{\rm F}D_a)^{2n+1}\chi_R^a\chi_\Lambda.
\label{AppTildIndSplt}
\end{align}
Since the first terms in the right-hand side in the second equality for both the indices 
coincide with each other from $|\Omega|=|\Lambda|$, 
it is enough to estimate the second terms in both indices. 

We set 
$$
\Omega_{\rm out}:=\{a\in\Omega\>|\>{\rm dist}(a,\partial\Omega)\le R\},
$$
where ${\rm dist}(a,\partial\Omega)$ is the distance between $a$ and the boundary $\partial\Omega$ of 
the region $\Omega$. Clearly, one has 
$$
|\Omega_{\rm out}|\le{\rm Const.}|\partial\Omega|R.
$$
The second term in the right-hand side in the second equality of (\ref{AppIndSplt}) is estimated as  
\begin{align*}
&\frac{1}{|\Omega|}\left|\int_\Omega dv(a) 
\>{\rm Tr}\>\gamma^{(2n+1)}(P_{\rm F}-D_aP_{\rm F}D_a)^{2n+1}\chi_R^a(1-\chi_\Lambda)\right|\\
&=\frac{1}{|\Omega|}\Biggl|\int_\Omega dv(a) 
\>\sum_{u\in\ze^d\backslash\Lambda}\sum_{\alpha,\mu}
\langle\tilde{\zeta}_u^{\alpha,\mu},\gamma^{(2n+1)}(P_{\rm F}-D_aP_{\rm F}D_a)^{2n+1}\tilde{\zeta}_u^{\alpha,\mu}\rangle
\chi_R^a(u)\Biggr|\\
&\le \frac{{\rm Const.}}{|\Omega|}\int_{\Omega_{\rm out}} dv(a) 
\>\sum_{\substack{u\in\ze^d\backslash\Lambda:\\ |u-a|\le R}} \ 1\\
&\le \frac{{\rm Const.}}{|\Omega|}\int_{\Omega_{\rm out}} dv(a) R^d\le \frac{{\rm Const.}}{|\Omega|}
|\partial\Omega|R\times R^d.
\end{align*}
Hence the contribution is vanishing in the limit $\Omega\nearrow\re^d$ for a fixed $R$.  
Similarly, the second term in (\ref{AppTildIndSplt}) is vanishing in the limit $\Lambda\nearrow\ze^d$ 
for a fixed $R$. Thus, the statement of Lemma~\ref{lem:IndtildeInd} has been proved. 

\section{Proof of Theorem~\ref{thm:IndDaIndtheta}}
\label{AppenProofthm:IndDaIndtheta}

In this appendix, we prove the relation between the two indices in Theorem~\ref{thm:IndDaIndtheta}. 
For this purpose, we prepare the following two lemmas: 

\begin{lem}
\label{TrvarthetachiLambda}
Fix the position $a$ of the kink and the permutation $\sigma$. Then, we have 
$$
\lim_{\Lambda\nearrow\ze^d}{\rm Tr}\> 
\chi_\Lambda P_{\rm F}[\vartheta_a^{(\sigma_1)},P_{\rm F}]\\
\cdots[\vartheta_a^{(\sigma_{2n})},P_{\rm F}]
={\rm Tr}\> P_{\rm F}[\vartheta_a^{(\sigma_1)},P_{\rm F}]\\
\cdots[\vartheta_a^{(\sigma_{2n})},P_{\rm F}].
$$
\end{lem} 

\begin{proof}
Using the system of the complete orthonormal basis $\zeta_u^\alpha$, one has  
\begin{align*}
&{\rm Tr}\> 
\chi_\Lambda P_{\rm F}[\vartheta_a^{(\sigma_1)},P_{\rm F}]
\cdots[\vartheta_a^{(\sigma_{2n})},P_{\rm F}]\\
&=\sum_{u_1,\ldots,u_{2n}}\sum_{u_{2n+1}\in\Lambda}\sum_{\alpha_1,\ldots,\alpha_{2n+1}}
\langle\zeta_{u_{2n+1}}^{\alpha_{2n+1}},P_{\rm F}\zeta_{u_1}^{\alpha_1}\rangle
\langle\zeta_{u_1}^{\alpha_1},[\vartheta_a^{(\sigma_1)},P_{\rm F}]\zeta_{u_2}^{\alpha_2}\rangle\\
&\cdots \langle\zeta_{u_{2n}}^{\alpha_{2n}},[\vartheta_a^{(\sigma_{2n})},P_{\rm F}]\zeta_{u_{2n+1}}^{\alpha_{2n+1}}
\rangle.
\end{align*}
The matrix elements are computed as 
\begin{align*}
\langle\zeta_u^\alpha,[\vartheta_a^{(j)},P_{\rm F}]\zeta_v^\beta\rangle&=
\langle\zeta_u^\alpha,P_{\rm F}\zeta_v^\beta\rangle[\vartheta(u^{(j)}-a^{(j)})-\vartheta(v^{(j)}-a^{(j)}]\\
&=\langle\zeta_u^\alpha,P_{\rm F}\zeta_v^\beta\rangle\times
\begin{cases} \ \ 1, & \text{$u^{(j)}\ge a^{(j)}>v^{(j)}$};\\
-1, & \text{$v^{(j)}\ge a^{(j)}>u^{(j)}$};\\
\ \ 0, & \text{otherwise}.
\end{cases}
\end{align*}
Since one can easily show  
$$
|u^{(j)}-v^{(j)}|=|u^{(j)}-a^{(j)}+a^{(j)}-v^{(j)}|=|u^{(j)}-a^{(j)}|+|a^{(j)}-v^{(j)}|
$$
for $u^{(j)}, v^{(j)}$ satisfying the above conditions, $u^{(j)}\ge a^{(j)}>v^{(j)}$ or $v^{(j)}\ge a^{(j)}>u^{(j)}$,  
for the nonvanishing matrix elements, the following bound is valid: 
\begin{align*}
|\langle\zeta_u^\alpha,[\vartheta_a^{(j)},P_{\rm F}]\zeta_v^\beta\rangle|
&\le {\rm Const.}e^{-\kappa|u-v|/2}\\
&\times \exp[-\kappa\{|u^{(j)}-a^{(j)}|+|v^{(j)}-a^{(j)}|\}/(4n)]\\
&\times \exp\Bigl[-\kappa\sum_{k\ne j}\bigl|u^{(k)}-v^{(k)}\bigr|/(4n)\Bigr],
\end{align*}
where we have used the decay bound (\ref{decayPF}) for the projection $P_{\rm F}$ onto the Fermi sea. 
By using this bound, the product of the matrix elements is estimated as   
\begin{align*}
&\left|\langle\zeta_{u_{2n+1}}^{\alpha_{2n+1}},P_{\rm F}\zeta_{u_1}^{\alpha_1}\rangle
\langle\zeta_{u_1}^{\alpha_1},[\vartheta_a^{(\sigma_1)},P_{\rm F}]\zeta_{u_2}^{\alpha_2}\rangle
\cdots \langle\zeta_{u_{2n}}^{\alpha_{2n}},[\vartheta_a^{(\sigma_{2n})},P_{\rm F}]\zeta_{u_{2n+1}}^{\alpha_{2n+1}}
\rangle\right|\\
&\le\exp\Bigl[-\kappa|u_{2n+1}-u_1|/2-\kappa\sum_{j=1}^{2n}|u_{2n+1}^{(j)}-u_1^{(j)}|/(4n)\Bigr]\\
&\times \exp\Bigl[-\kappa|u_1-u_2|/2-\kappa|u_1^{(\sigma_1)}-a^{(\sigma_1)}|/(4n)
-\kappa\sum_{j_1\ne \sigma_1}|u_1^{(j_1)}-u_2^{(j_1)}|/(4n)\Bigr]\\
&\times \exp\Bigl[-\kappa|u_2-u_3|/2-\kappa|u_2^{(\sigma_2)}-a^{(\sigma_2)}|/(4n)
-\kappa\sum_{j_2\ne \sigma_2}|u_2^{(j_2)}-u_3^{(j_2)}|/(4n)\Bigr]\\
&\qquad\qquad\qquad\qquad\vdots\\
&\times \exp\Bigl[-\kappa|u_{2n}-u_{2n+1}|/2-\kappa|u_{2n}^{(\sigma_{2n})}-a^{(\sigma_{2n})}|/(4n)\Bigr]\\
&\times \exp\Bigl[-\kappa\sum_{j_{2n}\ne \sigma_{2n}}|u_{2n}^{(j_{2n})}-u_{2n+1}^{(j_{2n})}|/(4n)\Bigr]\\
&\le\exp\Bigl[-\kappa|u_{2n+1}-u_1|/2-\kappa|u_1-u_2|/2-\cdots -\kappa|u_{2n}-u_{2n+1}|/2\Bigr]\\
&\times \exp[-\kappa|u_{2n+1}-a|/(4n)],
\end{align*}
where we have used the inequalities, 
$$
|u_{2n+1}^{(\sigma_1)}-u_1^{(\sigma_1)}|+|u_1^{(\sigma_1)}-a^{(\sigma_1)}|
\ge |u_{2n+1}^{(\sigma_1)}-a^{(\sigma_1)}|, 
$$
$$
|u_{2n+1}^{(\sigma_2)}-u_1^{(\sigma_2)}|+|u_1^{(\sigma_2)}-u_2^{(\sigma_2)}|
+|u_2^{(\sigma_2)}-a^{(\sigma_2)}|\ge |u_{2n+1}^{(\sigma_2)}-a^{(\sigma_2)}|,
$$
$$
\vdots
$$
$$
|u_{2n+1}^{(\sigma_{2n})}-u_1^{(\sigma_{2n})}|+|u_1^{(\sigma_{2n})}-u_2^{(\sigma_{2n})}|+
\cdots+|u_{2n}^{(\sigma_{2n})}-a^{(\sigma_{2n})}|
\ge |u_{2n+1}^{(\sigma_{2n})}-a^{(\sigma_{2n})}|,
$$
and 
$$
\sum_{j=1}^{2n}|u_{2n+1}^{(j)}-a^{(j)}|\ge |u_{2n+1}-a|.
$$
Clearly, the statement of the lemma follows from the above inequality for the product of the matrix elements. 
\end{proof}

\begin{lem}
\label{lem:IndependKinkInd} 
Let $a=(a^{(1)},a^{(2)},\ldots,a^{(2n)})\in\re^d$, and $a'=a+(\Delta a,0,\ldots,0)$ with 
$\Delta a\in\re$. Then, one has 
$$
{\rm Ind}^{(2n)}(\vartheta_{a'},P_{\rm F})={\rm Ind}^{(2n)}(\vartheta_a,P_{\rm F}).
$$
\end{lem}

\begin{proof}
Consider 
\begin{align*}
&{\rm Tr}\thinspace P_{\rm F}[\vartheta_{a'}^{(1)},P_{\rm F}][\vartheta_a^{(2)},P_{\rm F}]\cdots[\vartheta_a^{(2n)},P_{\rm F}]
-{\rm Tr}\thinspace P_{\rm F}[\vartheta_{a}^{(1)},P_{\rm F}][\vartheta_a^{(2)},P_{\rm F}]\cdots[\vartheta_a^{(2n)},P_{\rm F}]\\
&={\rm Tr}\>P_{\rm F}[(\vartheta_{a'}^{(1)}-\vartheta_{a}^{(1)}),P_{\rm F}][\vartheta_a^{(2)},P_{\rm F}]
\cdots[\vartheta_a^{(2n)},P_{\rm F}].
\end{align*}
In the same way as in the proof of Lemma~\ref{TrvarthetachiLambda}, the right-hand side is written as  
\begin{multline*}
{\rm Tr}\;P_{\rm F}[(\vartheta_{a'}^{(1)}-\vartheta_{a}^{(1)}),P_{\rm F}][\vartheta_a^{(2)},P_{\rm F}]
\cdots[\vartheta_a^{(2n)},P_{\rm F}]\\
=\lim_{\Gamma\nearrow\ze^d}
{\rm Tr}\;P_{\rm F}[\Delta\vartheta_a^{(1)},P_{\rm F}][\vartheta_a^{(2)},P_{\rm F}]
\cdots[\vartheta_a^{(2n)},P_{\rm F}],
\end{multline*}
where 
$$
\Delta\vartheta_a^{(1)}:=(\vartheta_{a'}^{(1)}-\vartheta_{a}^{(1)})\chi_\Gamma
$$
with the characteristic function $\chi_\Gamma$ of the finite lattice $\Gamma\subset\ze^d$. For simplicity, 
we choose the lattice $\Gamma$ so that the distance between $\ze^d\backslash\Gamma$ and $\{a,a'\}$ is 
sufficiently large. 

Note that 
\begin{align*}
P_{\rm F}[\Delta\vartheta_a^{(1)},P_{\rm F}]
&=P_{\rm F}(\Delta\vartheta_a^{(1)}P_{\rm F}-P_{\rm F}\Delta\vartheta_a^{(1)})\\
&=P_{\rm F}(\Delta\vartheta_a^{(1)}P_{\rm F}-\Delta\vartheta_a^{(1)})=-P_{\rm F}\Delta\vartheta_a^{(1)}(1-P_{\rm F}). 
\end{align*}
Further, one has
$$
P_{\rm F}[\vartheta_a^{(i)},P_{\rm F}]P_{\rm F}=0
$$
and 
$$
(1-P_{\rm F})[\vartheta_a^{(i)},P_{\rm F}](1-P_{\rm F})=0. 
$$
Using these identities, we obtain 
\begin{align*}
&{\rm Tr}\;
P_{\rm F}[\Delta\vartheta_a^{(1)},P_{\rm F}][\vartheta_a^{(2)},P_{\rm F}] \cdots[\vartheta_a^{(2n)},P_{\rm F}]\\
&=-{\rm Tr}\;P_{\rm F}\Delta\vartheta_a^{(1)}(1-P_{\rm F})
[\vartheta_a^{(2)},P_{\rm F}]P_{\rm F}[\vartheta_a^{(3)},P_{\rm F}](1-P_{\rm F})\\
&\qquad\qquad\cdots(1-P_{\rm F})[\vartheta_a^{(2n)},P_{\rm F}]P_{\rm F}\\
&=-{\rm Tr}\;\Delta\vartheta_a^{(1)}(1-P_{\rm F})
[\vartheta_a^{(2)},P_{\rm F}][\vartheta_a^{(3)},P_{\rm F}]\cdots[\vartheta_a^{(2n)},P_{\rm F}],
\end{align*}
where we have used the property of the trace and the fact that $\chi_\Gamma$ is trace class. 

Next, consider 
\begin{equation}
\label{Trcontper}
-{\rm Tr}\;P_{\rm F}[\vartheta_a^{(2)},P_{\rm F}]\cdots[\vartheta_a^{(2n)},P_{\rm F}][\Delta\vartheta_a^{(1)},P_{\rm F}]
\end{equation}
which is derived from the above quantity with the permutation of the indices of $\vartheta_a^{(i)}$. 
The negative sign comes from the signature of the permutation. 
Note that 
\begin{align*}
[\Delta\vartheta_a^{(1)},P_{\rm F}]P_{\rm F}&=(\Delta\vartheta_a^{(1)}P_{\rm F}-P_{\rm F}\vartheta_a^{(1)})P_{\rm F}\\
&=(\Delta\vartheta_a^{(1)}-P_{\rm F}\vartheta_a^{(1)})P_{\rm F}=(1-P_{\rm F})\vartheta_a^{(1)}P_{\rm F}. 
\end{align*}
Therefore, in the same way as the above, we have 
\begin{align*}
&-{\rm Tr}\;P_{\rm F}[\vartheta_a^{(2)},P_{\rm F}]\cdots[\vartheta_a^{(2n)},P_{\rm F}]
[\Delta\vartheta_a^{(1)},P_{\rm F}]\\
&=-{\rm Tr}\;P_{\rm F}[\vartheta_a^{(2)},P_{\rm F}](1-P_{\rm F})\cdots P_{\rm F}[\vartheta_a^{(2n)},P_{\rm F}]
(1-P_{\rm F})
[\Delta\vartheta_a^{(1)},P_{\rm F}]P_{\rm F}\\
&=-{\rm Tr}\;P_{\rm F}[\vartheta_a^{(2)},P_{\rm F}](1-P_{\rm F})\cdots P_{\rm F}[\vartheta_a^{(2n)},P_{\rm F}]
(1-P_{\rm F})\Delta\vartheta_a^{(1)}P_{\rm F}\\
&=-{\rm Tr}\;P_{\rm F}[\vartheta_a^{(2)},P_{\rm F}]\cdots [\vartheta_a^{(2n)},P_{\rm F}]
\chi_\Gamma\Delta\vartheta_a^{(1)}P_{\rm F}\\
&=-{\rm Tr}\;\Delta\vartheta_a^{(1)}P_{\rm F}[\vartheta_a^{(2)},P_{\rm F}]\cdots [\vartheta_a^{(2n)},P_{\rm F}]
\chi_\Gamma\\
&=-{\rm Tr}\;\Delta\vartheta_a^{(1)}P_{\rm F}[\vartheta_a^{(2)},P_{\rm F}]\cdots [\vartheta_a^{(2n)},P_{\rm F}].
\end{align*}
By adding two quantities, we obtain 
\begin{multline*}
-{\rm Tr}\;\Delta\vartheta_a^{(1)}(1-P_{\rm F})[\vartheta_a^{(2)},P_{\rm F}]\cdots [\vartheta_a^{(2n)},P_{\rm F}]\\
-{\rm Tr}\;\Delta\vartheta_a^{(1)}P_{\rm F}[\vartheta_a^{(2)},P_{\rm F}]\cdots [\vartheta_a^{(2n)},P_{\rm F}]\\
=-{\rm Tr}\;\Delta\vartheta_a^{(1)}[\vartheta_a^{(2)},P_{\rm F}]\cdots [\vartheta_a^{(2n)},P_{\rm F}].
\end{multline*}
Because of the sum of the permutations in the expression of the index, the nonvanishing contributions 
in the right-hand side are given by 
$$
-{\rm Tr}\;\Delta\vartheta_a^{(1)}\vartheta_a^{(2)}P_{\rm F}\vartheta_a^{(3)}P_{\rm F}\cdots 
P_{\rm F}\vartheta_a^{(2n)}P_{\rm F}
$$
and 
$$
{\rm Tr}\;\Delta\vartheta_a^{(1)}P_{\rm F}\vartheta_a^{(2)}P_{\rm F}\vartheta_a^{(3)}\cdots P_{\rm F}\vartheta_a^{(2n)}.
$$
Owing to the sum of the permutations, the second contribution can be written as 
\begin{align*}
{\rm Tr}\;\Delta\vartheta_a^{(1)}P_{\rm F}\vartheta_a^{(3)}\cdots 
P_{\rm F}\vartheta_a^{(2n)}P_{\rm F}\vartheta_a^{(2)}&=
{\rm Tr}\;\vartheta_a^{(2)}\Delta\vartheta_a^{(1)}P_{\rm F}\vartheta_a^{(3)}\cdots 
P_{\rm F}\vartheta_a^{(2n)}P_{\rm F}\\
&={\rm Tr}\;\Delta\vartheta_a^{(1)}\vartheta_a^{(2)}P_{\rm F}\vartheta_a^{(3)}\cdots 
P_{\rm F}\vartheta_a^{(2n)}P_{\rm F}.
\end{align*}
This cancels out the first contribution. 

As to the rest of the cases, it is sufficient to treat  
\begin{align*}
&-{\rm Tr}\thinspace P_{\rm F}[\vartheta_a^{(2n)},P_{\rm F}][\vartheta_{a'}^{(1)},P_{\rm F}][\vartheta_a^{(2)},P_{\rm F}]
\cdots[\vartheta_a^{(2n-1)},P_{\rm F}]\\
&+{\rm Tr}\thinspace P_{\rm F}[\vartheta_a^{(2n)},P_{\rm F}][\vartheta_{a}^{(1)},P_{\rm F}][\vartheta_a^{(2)},P_{\rm F}]
\cdots[\vartheta_a^{(2n-1)},P_{\rm F}]\\
&=-{\rm Tr}\>P_{\rm F}[\vartheta_a^{(2n)},P_{\rm F}][(\vartheta_{a'}^{(1)}-\vartheta_{a}^{(1)}),P_{\rm F}][\vartheta_a^{(2)},P_{\rm F}]
\cdots[\vartheta_a^{(2n-1)},P_{\rm F}]
\end{align*}
without loss of generality. Therefore, we consider 
\begin{align*}
&-{\rm Tr}\>P_{\rm F}[\vartheta_a^{(2n)},P_{\rm F}][\Delta\vartheta_a^{(1)},P_{\rm F}][\vartheta_a^{(2)},P_{\rm F}]
\cdots[\vartheta_a^{(2n-1)},P_{\rm F}]\\
&=-{\rm Tr}\>P_{\rm F}[\vartheta_a^{(2n)},P_{\rm F}](1-P_{\rm F})[\Delta\vartheta_a^{(1)},P_{\rm F}][\vartheta_a^{(2)},P_{\rm F}]
\cdots[\vartheta_a^{(2n-1)},P_{\rm F}]\\
&=-{\rm Tr}\;(1-P_{\rm F})[\Delta\vartheta_a^{(1)},P_{\rm F}][\vartheta_a^{(2)},P_{\rm F}]
\cdots[\vartheta_a^{(2n-1)},P_{\rm F}][\vartheta_a^{(2n)},P_{\rm F}].
\end{align*}
The contribution (\ref{Trcontper}) can be written 
\begin{align*}
&-{\rm Tr}\;P_{\rm F}[\vartheta_a^{(2)},P_{\rm F}]\cdots[\vartheta_a^{(2n)},P_{\rm F}][\Delta\vartheta_a^{(1)},P_{\rm F}]\\
&=-{\rm Tr}\;P_{\rm F}[\vartheta_a^{(2)},P_{\rm F}]\cdots[\vartheta_a^{(2n)},P_{\rm F}]
(1-P_{\rm F})[\Delta\vartheta_a^{(1)},P_{\rm F}]\\
&=-{\rm Tr}\;(1-P_{\rm F})[\Delta\vartheta_a^{(1)},P_{\rm F}][\vartheta_a^{(2)},P_{\rm F}]\cdots[\vartheta_a^{(2n)},P_{\rm F}]
\end{align*}
Thus, the above contribution is already counted, and cancels out the first one. 
\end{proof}

\begin{proof}[Proof of Theorem~\ref{thm:IndDaIndtheta}]
Lemma~\ref{lem:IndependKinkInd} implies that the index ${\rm Ind}^{(2n)}(\vartheta_a,P_{\rm F})$ is independent of 
the position $a$ of the kink. 
Therefore, the statement of Theorem~\ref{thm:IndDaIndtheta} follows from 
the decay estimate in the proof of Lemma~\ref{TrvarthetachiLambda}. 
\end{proof}

\section{Proof of Theorem~\ref{thm:ChiralIndtheta}}
\label{proofTheorem:ChiralIndtheta}

In this appendix, we prove that the chiral index can be written in terms of the step functions as 
in Theorem~\ref{thm:ChiralIndtheta}.

In the same way as in the proof of Lemma~\ref{TrvarthetachiLambda}, one has 

\begin{lem}
Fix the position $a$ of the kink and the permutation $\sigma$. Then, we have
\begin{multline*}
\lim_{\Lambda\nearrow\ze^d}{\rm Tr}\; \chi_\Lambda SU[\vartheta_a^{(\sigma_1)},P_-][\vartheta_a^{(\sigma_2)},P_-]
\cdots[\vartheta_a^{(\sigma_{2n+1})},P_-]\\
={\rm Tr}\; SU[\vartheta_a^{(\sigma_1)},P_-][\vartheta_a^{(\sigma_2)},P_-]\cdots[\vartheta_a^{(\sigma_{2n+1})},P_-].
\end{multline*}
\end{lem}

Therefore, it is enough to prove the following analogue of Lemma~\ref{lem:IndependKinkInd}. 

\begin{lem}
Let $a=(a^{(1)},a^{(2)},\ldots,a^{(2n+1)})\in\re^d$, and $a'=a+(\Delta a,0,\ldots,0)$ with 
$\Delta a\in\re$. Then, we have  
$$
{\rm Ind}^{(2n+1)}(\vartheta_{a'},S,U)={\rm Ind}^{(2n+1)}(\vartheta_a,S,U).
$$
\end{lem}

\begin{proof}
To begin with, we note that 
\begin{align*}
&{\rm Tr}\;SU[\vartheta_{a}^{(1)},P_-][\vartheta_a^{(2)},P_-]\cdots[\vartheta_a^{(2n+1)},P_-]\\
&={\rm Tr}\;SP_+[\vartheta_{a}^{(1)},P_-][\vartheta_a^{(2)},P_-]\cdots[\vartheta_a^{(2n+1)},P_-]\\
&-{\rm Tr}\;SP_-[\vartheta_{a}^{(1)},P_-][\vartheta_a^{(2)},P_-]\cdots[\vartheta_a^{(2n+1)},P_-].
\end{align*}
The first term in the right-hand side can be written as 
\begin{align*}
&{\rm Tr}\;SP_+[\vartheta_{a}^{(1)},P_-][\vartheta_a^{(2)},P_-]\cdots[\vartheta_a^{(2n+1)},P_-]\\
&=-{\rm Tr}\;SP_+[\vartheta_{a}^{(1)},P_+][\vartheta_a^{(2)},P_+]\cdots[\vartheta_a^{(2n+1)},P_+]
\end{align*}
where we have used $P_-=1-P_+$. Thus, it is sufficient to handle one of the two types of terms.  

Consider 
\begin{align*}
&{\rm Tr}\;SP_-[\vartheta_{a'}^{(1)},P_-][\vartheta_a^{(2)},P_-]\cdots[\vartheta_a^{(2n+1)},P_-]\\
&-{\rm Tr}\;SP_-[\vartheta_{a}^{(1)},P_-][\vartheta_a^{(2)},P_-]\cdots[\vartheta_a^{(2n+1)},P_-]\\
&={\rm Tr}\;SP_-[(\vartheta_{a'}^{(1)}-\vartheta_{a}^{(1)}),P_-][\vartheta_a^{(2)},P_-]\cdots[\vartheta_a^{(2n+1)},P_-].
\end{align*}
In the same way as in the proof of Lemma~\ref{TrvarthetachiLambda}, one has 
\begin{align*}
&{\rm Tr}\;SP_-[(\vartheta_{a'}^{(1)}-\vartheta_{a}^{(1)}),P_-][\vartheta_a^{(2)},P_-]\cdots[\vartheta_a^{(2n+1)},P_-]\\
&=\lim_{\Gamma\nearrow\ze^d}{\rm Tr}\;SP_-[\Delta\vartheta_{a}^{(1)},P_-][\vartheta_a^{(2)},P_-]
\cdots[\vartheta_a^{(2n+1)},P_-],
\end{align*}
where 
$$
\Delta\vartheta_{a}^{(1)}:=(\vartheta_{a'}^{(1)}-\vartheta_{a}^{(1)})\chi_\Gamma. 
$$
As in the proof of  Lemma~\ref{lem:IndependKinkInd}, we have 
$$
P_-[\Delta\vartheta_{a}^{(1)},P_-]=-P_-\Delta\vartheta_{a}^{(1)}(1-P_-),
$$
$$
P_-[\vartheta_a^{(i)},P_-]P_-=0
$$
and 
$$
(1-P_-)[\vartheta_a^{(i)},P_-](1-P_-)=0. 
$$
By using these identities, we obtain 
\begin{align*}
&{\rm Tr}\;SP_-[\Delta\vartheta_{a}^{(1)},P_-][\vartheta_a^{(2)},P_-]\cdots[\vartheta_a^{(2n+1)},P_-]\\
&=-{\rm Tr}\;S\Delta\vartheta_{a}^{(1)}(1-P_-)[\vartheta_a^{(2)},P_-]\cdots[\vartheta_a^{(2n+1)},P_-],
\end{align*}
where we have used $SP_-=(1-P_-)S$, the property of the trace and the fact that $\chi_\Gamma$ is trace class. 

Next, consider 
\begin{equation}
\label{2ndTrSP-Commus}
{\rm Tr}\;SP_-[\vartheta_a^{(2)},P_-]\cdots[\vartheta_a^{(2n+1)},P_-][\Delta\vartheta_{a}^{(1)},P_-]
\end{equation}
which is derived from the above quantity with the permutation of the indices of $\vartheta_a^{(i)}$. 
Using the above identities, one has 
\begin{align*}
&{\rm Tr}\;SP_-[\vartheta_a^{(2)},P_-]\cdots[\vartheta_a^{(2n+1)},P_-][\Delta\vartheta_{a}^{(1)},P_-]\\
&=-{\rm Tr}\;S[\vartheta_a^{(2)},P_-]\cdots[\vartheta_a^{(2n+1)},P_-]P_-\Delta\vartheta_{a}^{(1)}(1-P_-)\\
&=-{\rm Tr}\;(1-P_-)S[\vartheta_a^{(2)},P_-]\cdots[\vartheta_a^{(2n+1)},P_-]P_-\chi_\Gamma\Delta\vartheta_{a}^{(1)}\\
&=-{\rm Tr}\;\Delta\vartheta_{a}^{(1)}(1-P_-)S[\vartheta_a^{(2)},P_-]\cdots[\vartheta_a^{(2n+1)},P_-]P_-\chi_\Gamma\\
&=-{\rm Tr}\;\Delta\vartheta_{a}^{(1)}SP_-[\vartheta_a^{(2)},P_-]\cdots[\vartheta_a^{(2n+1)},P_-]\\
&=-{\rm Tr}\;S\Delta\vartheta_{a}^{(1)}P_-[\vartheta_a^{(2)},P_-]\cdots[\vartheta_a^{(2n+1)},P_-]
\end{align*}
By adding two quantities, we have 
$$
-{\rm Tr}\;S\Delta\vartheta_{a}^{(1)}[\vartheta_a^{(2)},P_-]\cdots[\vartheta_a^{(2n+1)},P_-].
$$
Because of the sum of the permutations in the expression of the index, the nonvanishing contributions 
are given by 
$$
-{\rm Tr}\;S\Delta\vartheta_a^{(1)}\vartheta_a^{(2)}P_-\vartheta_a^{(3)}P_-\cdots 
P_-\vartheta_a^{(2n+1)}P_-
$$
and 
$$
-{\rm Tr}\;S\Delta\vartheta_a^{(1)}P_-\vartheta_a^{(2)}P_-\vartheta_a^{(3)}\cdots P_-\vartheta_a^{(2n+1)}.
$$
Owing to the sum of the permutations, the second contribution can be written as 
\begin{align*}
&{\rm Tr}\;S\Delta\vartheta_a^{(1)}P_-\vartheta_a^{(3)}\cdots P_-\vartheta_a^{(2n+1)}P_-\vartheta_a^{(2)}\\
&={\rm Tr}\;\vartheta_a^{(2)}S\Delta\vartheta_a^{(1)}P_-\vartheta_a^{(3)}\cdots P_-\vartheta_a^{(2n+1)}P_-\\
&={\rm Tr}\;S\Delta\vartheta_a^{(1)}\vartheta_a^{(2)}P_-\vartheta_a^{(3)}\cdots P_-\vartheta_a^{(2n+1)}P_-.
\end{align*}
This cancels out the first contribution. 

As to the rest of the cases, it is sufficient to treat  
\begin{align*}
&{\rm Tr}\thinspace SP_-[\vartheta_a^{(2n+1)},P_-][\vartheta_{a'}^{(1)},P_-][\vartheta_a^{(2)},P_-]
\cdots[\vartheta_a^{(2n)},P_-]\\
&-{\rm Tr}\thinspace SP_-[\vartheta_a^{(2n+1)},P_-][\vartheta_{a}^{(1)},P_-][\vartheta_a^{(2)},P_-]
\cdots[\vartheta_a^{(2n)},P_-]\\
&={\rm Tr}\>SP_-[\vartheta_a^{(2n+1)},P_-][(\vartheta_{a'}^{(1)}-\vartheta_{a}^{(1)}),P_-][\vartheta_a^{(2)},P_-]
\cdots[\vartheta_a^{(2n)},P_-]
\end{align*}
without loss of generality. Therefore, we consider 
\begin{align*}
&{\rm Tr}\>SP_-[\vartheta_a^{(2n+1)},P_-][\Delta\vartheta_a^{(1)},P_-][\vartheta_a^{(2)},P_-]
\cdots[\vartheta_a^{(2n)},P_-]\\
&={\rm Tr}\>SP_-[\vartheta_a^{(2n+1)},P_-](1-P_-)[\Delta\vartheta_a^{(1)},P_-][\vartheta_a^{(2)},P_-]
\cdots[\vartheta_a^{(2n)},P_-]\\
&=-{\rm Tr}\>(1-P_-)[\vartheta_a^{(2n+1)},P_-]S(1-P_-)[\Delta\vartheta_a^{(1)},P_-][\vartheta_a^{(2)},P_-]
\cdots[\vartheta_a^{(2n)},P_-]\\
&=-{\rm Tr}\;S(1-P_-)[\Delta\vartheta_a^{(1)},P_-][\vartheta_a^{(2)},P_-]
\cdots[\vartheta_a^{(2n)},P_-][\vartheta_a^{(2n+1)},P_-].
\end{align*}
The contribution (\ref{2ndTrSP-Commus}) can be written 
\begin{align*}
&{\rm Tr}\;SP_-[\vartheta_a^{(2)},P_-]\cdots[\vartheta_a^{(2n+1)},P_-][\Delta\vartheta_a^{(1)},P_-]\\
&={\rm Tr}\;SP_-[\vartheta_a^{(2)},P_-]\cdots[\vartheta_a^{(2n+1)},P_-]
P_-[\Delta\vartheta_a^{(1)},P_-]\\
&={\rm Tr}\;P_-[\Delta\vartheta_a^{(1)},P_-]SP_-[\vartheta_a^{(2)},P_-]\cdots[\vartheta_a^{(2n+1)},P_-]\\
&=-{\rm Tr}\;S(1-P_-)[\Delta\vartheta_a^{(1)},P_-][\vartheta_a^{(2)},P_-]\cdots[\vartheta_a^{(2n+1)},P_-].
\end{align*}
Thus, the above contribution is already counted, and cancels out the first one.
\end{proof}

\section{Chiral Invariant as a Linear Response Coefficient}
\label{sec:LR}

In this appendix, we derive the linear response coefficient in a general setting 
which includes interacting fermion systems. 
As an example, for the present noninteracting fermion system in one dimension,  
we show that the resulting expression of the imaginary part of the coefficient coincides 
with that of the topological invariant in the right-hand side of 
the chiral index formula (\ref{currentChernChiral}). 

Write $H_{\Lambda,N}$ for a $N$-fermion Hamiltonian on a finite lattice $\Lambda$, 
and consider a time-dependent Hamiltonian, 
$$
\tilde{H}_{\Lambda,N}(t):=H_{\Lambda,N}+\lambda_{\rm ext}V_{\Lambda,N}(t), 
$$
with the perturbation of the external potential, 
\begin{equation}
V_{\Lambda,N}(t)=V_{\Lambda,N}\alpha(t)e^{i\omega t}, 
\label{Wt}
\end{equation}
with the adiabatic function, 
$$
\alpha(t):=\begin{cases} 1, & \text{$t\ge 0$};\\
e^{\eta t}, & \text{$t<0$}.
\end{cases}
$$
Here the voltage difference $\lambda_{\rm ext}$ and the ${\rm AC}$ frequency $\omega$ are real parameters, 
and the adiabatic parameter $\eta$ is a small positive number.  
We switch on the potential $V_{\Lambda,N}$ at the initial time $t=-T_0$ 
with a large positive $T_0$, and measure the current at the time $t\ge 0$. 

The time-dependent Schr\"odinger equation is given by  
\[
i\frac{d}{dt}\Psi^{(N)}(t)=\tilde{H}_{\Lambda,N}(t)\Psi^{(N)}(t)
\]
for the wavefunction $\Psi^{(N)}(t)$ for the $N$ fermions. 
We denote the time evolution operator for the unperturbed Hamiltonian $H_{\Lambda,N}$ by 
$$
U_{\Lambda,N}(t,s):=\exp\bigl[-i(t-s)H_{\Lambda,N}\bigr]\quad
\text{for \ } 
t,s\in\re.
$$
We choose the initial vector $\Psi^{(N)}(-T_0)$ at $t=-T_0$ as 
\[
\Psi^{(N)}(-T_0)=U^{(\Lambda)}(-T_0,0)\Phi^{(N)}
\] 
with a vector $\Phi^{(N)}$. Then, the final vector $\Psi^{(N)}(t)$ 
is obtained as 
\begin{multline}
\Psi^{(N)}(t)=U_{\Lambda,N}(t,0)\Phi^{(N)}\\-i\lambda_{\rm ext}\int_{-T_0}^t ds\; U_{\Lambda,N}(t,s)
V_{\Lambda,N}(s)U_{\Lambda,N}(s,0)\Phi^{(N)}+o(\lambda_{\rm ext})
\label{PsiN}
\end{multline}
by using a perturbation theory \cite{Koma1}, where $o(\lambda_{\rm ext})$ denotes a vector $\Psi^{(N)}_R$ 
with the norm $\Vert\Psi_R^{(N)}\Vert$ satisfying $\Vert\Psi_R^{(N)}\Vert/\lambda_{\rm ext}\rightarrow 0$ 
as $\lambda_{\rm ext}\rightarrow 0$.  

We denote the $N$ fermion ground state of the unperturbed Hamiltonian $H_{\Lambda,N}$ 
by $\Phi_0^{(N)}$ with the energy eigenvalue $E_0^{(N)}$. 
We choose the initial vector as the ground-state vector $\Phi^{(N)}=\Phi_0^{(N)}$ with the norm $1$.  
Then, the ground-state expectation value of the local current operator 
$J_\ell$ at the time $t$ is given by 
\begin{equation}
\bigl\langle J_\ell\bigr\rangle_t:=
\bigl\langle\Psi^{(N)}(t),J_\ell \Psi^{(N)}(t)\bigr\rangle.
\label{expApproJ1}
\end{equation}

Using the linear perturbation (\ref{PsiN}), the expectation value is decomposed into 
three parts as  
\[
\bigl\langle J_\ell\bigr\rangle_t=
\bigl\langle J_\ell\bigr\rangle_t^{(0)}+
\bigl\langle J_\ell\bigr\rangle_t^{(1)}+o(\lambda_{\rm ext}),
\]
where 
\[
\bigl\langle J_\ell\bigr\rangle_t^{(0)}
=\big\langle\Phi_0^{(N)}, J_\ell\Phi_0^{(N)}\big\rangle,
\]
and
\begin{multline}
\bigl\langle J_\ell\bigr\rangle_t^{(1)}\\
=-i\lambda_{\rm ext}\int_{-T_0}^t ds\; \big\langle \Phi_0^{(N)},U_{\Lambda,N}(0,t)J_\ell
U_{\Lambda,N}(t,s)V_{\Lambda,N}(s)U_{\Lambda,N}(s,0)\Phi_0^{(N)}\big\rangle\\
+{\rm c}.{\rm c}.
\label{expApproJ11} 
\end{multline}
The first term $\bigl\langle J_\ell\bigr\rangle_t^{(0)}$ 
is the persistent current which is usually vanishing. 
For simplicity, we assume that the persistent current is vanishing. 
We are interested 
in the second term $\bigl\langle J_\ell\bigr\rangle_t^{(1)}$ 
which gives the linear response coefficient, i.e., the conductance.  
Using this and the definition (\ref{Wt}) of $V_{\Lambda,N}(t)$, 
the contribution (\ref{expApproJ11}) is written 
$$
\bigl\langle J_\ell\bigr\rangle_t^{(1)}
=-\sum_{j\ne 0}\frac{\lambda_{\rm ext}}{E_j^{(N)}-E_0^{(N)}+\omega}\;
\big\langle \Phi_0^{(N)},J_\ell\Phi_j^{(N)}\big\rangle
\big\langle \Phi_j^{(N)},V^{(\Lambda)}\Phi_0^{(N)}\big\rangle\; e^{i\omega t}+{\rm c}.{\rm c.}
$$
in the limit $T_0\rightarrow\infty$ and $\eta\rightarrow 0$. Here, $\Phi_j^{(N)}$ are 
the energy eigenvectors of the excited state for the unperturbed Hamiltonian $H_{\Lambda,N}$ 
with the eigenenergy $E_j^{(N)}$, $j\ge 1$. Therefore, the linear response coefficient of the ${\rm AC}$ 
conductance with the frequency $\omega$ is given by 
$$
g(\omega)=-\sum_{j\ne 0}\frac{1}{E_j^{(N)}-E_0^{(N)}+\omega}\;
\big\langle \Phi_0^{(N)},J_\ell\Phi_j^{(N)}\big\rangle
\big\langle \Phi_j^{(N)},V^{(\Lambda)}\Phi_0^{(N)}\big\rangle. 
$$
In particular, the imaginary part of the conductance in the zero frequency limit $\omega\rightarrow 0$ 
is given by 
$$
g_{\rm I}(0)=\sum_{j\ne 0}\frac{1}{E_0^{(N)}-E_j^{(N)}}\;
{\rm Im}\;\big\langle \Phi_0^{(N)},J_\ell\Phi_j^{(N)}\big\rangle
\big\langle \Phi_j^{(N)},V^{(\Lambda)}\Phi_0^{(N)}\big\rangle. 
$$

As an example, consider a noninteracting fermion system in one dimension. 
The $N$-fermion external potential $V_{\Lambda,N}$ is given by the $N$-fermion chiral operator 
which is determined by the chiral operator $S$ for single fermion.  
Then, the above conductance $g_{\rm I}(0)$ coincides with 
the Chern number in the right-hand side of (\ref{currentChernChiral}) except for the prefactor $2$. 

\section{Even Anti-Unitary Transformations}
\label{EvenALO}

Proposition~\ref{prop:EAUT} below is useful for dealing with the cases, BDI class in two dimensions, 
CII class in six dimensions, and D class, in one dimension. 

Consider an anti-linear transformation $\tilde{\Sigma}$ which is given by 
$$
\tilde{\Sigma}:=U^{\Sigma}K,
$$ 
where $U^{\Sigma}$ is a unitary operator and $K$ is complex conjugation. 

\begin{prop}
\label{prop:EAUT}
Let $A$ be a compact operator, and let $\tilde{\Sigma}$ be anti-linear operator satisfying $\tilde{\Sigma}^2=+1$. 
Suppose 
$$
\tilde{\Sigma}A\varphi=A\tilde{\Sigma}\varphi
$$
for any wavefunction $\varphi$. Then, we can choose the eigenvectors of $A$ with eigenvalue $\lambda\ne 0$ so that 
$$
A\varphi=\lambda\varphi\quad\mbox{and}\quad \tilde{\Sigma}\varphi=\varphi.
$$
\end{prop}

\begin{proof}
Clearly, from the assumptions, we have that 
if $\varphi$ is an eigenvector of $A$, then $\tilde{\Sigma}\varphi$ is also an eigenvector of $A$ 
with the same eigenvalue.  

Consider first the case that the two vectors $\varphi$ and ${\tilde\Sigma}\varphi$ are linearly dependent. 
Namely, the following relation holds: 
$$
{\tilde\Sigma}\varphi=\kappa \varphi
$$
with a complex constant $\kappa\neq 0$. By multiplying the relation by ${\tilde \Sigma}$ and 
using ${\tilde \Sigma}^2\varphi=\varphi$, one has 
$$
\varphi=\overline{\kappa}{\tilde\Sigma}\varphi.
$$
Therefore, we have $\kappa\overline{\kappa}=1$. This implies that one can write $\kappa=e^{i\theta}$ 
with $\theta\in[0,2\pi)$. Then, the original relation is written 
$$
{\tilde\Sigma}\varphi=e^{i\theta} \varphi.
$$
Further, this can be written  
$$
{\tilde\Sigma}e^{i\theta/2}\varphi=e^{i\theta/2} \varphi.
$$
Consequently, we can choose the eigenvector $\varphi$ of $A$ so that it satisfies  
${\tilde \Sigma}\varphi=\varphi$. 

When the two vectors $\varphi$ and ${\tilde\Sigma}\varphi$ are linearly independent, 
then we choose 
$$ 
\varphi_+=\varphi+{\tilde\Sigma}\varphi
\quad \mbox{and} \quad 
\varphi_-=i(\varphi-{\tilde\Sigma}\varphi).
$$
These are linearly independent and satisfy 
$$
{\tilde\Sigma}\varphi_+=\varphi_+
\quad \mbox{and} \quad 
{\tilde\Sigma}\varphi_-=\varphi_-,
$$
where we have used ${\tilde\Sigma}^2\varphi=\varphi$. 
Thus, we can choose every eigenvector $\varphi$ of $A$ with nonzero eigenvalue 
to satisfy ${\tilde \Sigma}\varphi=\varphi$.  
\end{proof}

\section{Relations between the Indices}
\label{RelationIndices}

In this appendix, we recall a useful relation between the dimensions of the kernels of 
two Fredholm operators. As an application, we deal with   
the relation between the integer-valued index of DIII class and the $\ze_2$-valued index  
of AII class in the weak limit for the perturbation in three dimensions.

The following proposition is due to Gro{\ss}mann and Schulz-Baldes \cite{GSB}. 

\begin{prop}
\label{TTpProp}
Let $P$ and $E$ be projections. Set 
$$
\mathcal{T}=P(1-2E)P+1-P
$$
and 
$$
\mathcal{T}'=E(1-2P)E+1-E.
$$
Suppose that both of $\mathcal{T}$ and $\mathcal{T}'$ are Fredholm operators. Then, the following relation is valid: 
$$
{\rm dim}\;{\rm ker}\; \mathcal{T}={\rm dim}\;{\rm ker}\; \mathcal{T}'.
$$
\end{prop}

\begin{proof}
To begin with, we note that 
$$
\mathcal{T}=P-2PEP+1-P=1-2PEP.
$$
Similarly, one has 
$$
\mathcal{T}'=1-2EPE.
$$
 
Let $v_0\in{\rm ker}\; \mathcal{T}$. Clearly, one has $v_0=Pv_0$. Further, one has 
$$
0=\mathcal{T}v_0=(1-2PEP)v_0=v_0-2PEPv_0.
$$
Therefore, one obtains 
$$
PEPv_0=\frac{1}{2}v_0.
$$
Set $w_0=Ev_0$. Then, $w_0\ne 0$, and $w_0\in{\rm ker}\; \mathcal{T}'$. Actually, if $w_0=0$, then 
$$
\frac{1}{2}v_0=PEPv_0=PEv_0=Pw_0=0,
$$
where we have used $Pv_0=v_0$. This is a contradiction. Thus, $w_0\ne 0$. Further, 
\begin{align*}
\mathcal{T}'w_0&=(1-2EPE)w_0\\
&=w_0-2EPEw_0\\
&=w_0-2EPEv_0\\
&=w_0-2EPEPv_0=w_0-2E\cdot \frac{1}{2}v_0=0,
\end{align*}
where we have used $v_0=Pv_0$ and $PEPv_0=\frac{1}{2}v_0$. 

Similarly, let $w_0\in{\rm ker}\; \mathcal{T}'$. Then, $w_0=Ew_0$, and $u_0:=Pw_0\in {\rm ker}\; \mathcal{T}$.  
In addition, one has 
$$
u_0=Pw_0=PEv_0=PEPv_0=\frac{1}{2}v_0,
$$
where we have used $v_0=Pv_0$ and $PEPv_0=v_0/2$. The two maps, $E$ and $P$, on the kernels are invertible. 
Therefore, the two dimensionalities of those kernels coincide with each other. 
\end{proof}

In Proposition~\ref{TTpProp}, we choose $P=P_{\rm F}$ and $E=\mathcal{P}_{\rm D}=(1+D_a)/2$. 
Then, one has  
\begin{align}
{\rm dim}\;{\rm ker}\; \mathcal{T}&={\rm dim}\;{\rm ker}\;(P_{\rm F}-D_aP_{\rm F}D_a-1)\nonumber\\ 
&={\rm dim}\;{\rm ker}\; \mathcal{T}' 
={\rm dim}\;{\rm ker}\;(\mathcal{P}_{\rm D}-U\mathcal{P}_{\rm D}U-1),
\label{dimkerrelation}
\end{align}
where $U=(1-P_{\rm F})-P_{\rm F}=1-2P_{\rm F}$. 

As an application, let us consider DIII class in three dimensions, 
and deal with a weak perturbation which breaks the chiral symmetry 
but preserves the time-reversal symmetry. Due to the perturbation, DIII class changes to AII class. 
We want to obtain the relation between the integer-valued index of DIII class and the $\ze_2$-valued index  
of AII class in the weak limit for the perturbation. 

We recall the expression of the chiral index (\ref{IndOddFredT}).  
For DIII class, the integer-valued index is given by 
$$
{\rm Ind}^{(3)}(D_a,S,U)={\rm dim}\;{\rm ker}\; \mathfrak{T}_\chi-{\rm dim}\;{\rm ker}\; \mathfrak{T}_\chi^\ast,
$$
with the Fredholm operator, 
$$
\mathfrak{T}_\chi=\mathcal{P}_{\rm D}\mathcal{U}\mathcal{P}_{\rm D}+1-\mathcal{P}_{\rm D},
$$ 
where $\mathcal{U}$ is given by the off-diagonal matrix element of $U$ of (\ref{Umatu}). 
The parity of the right-hand side of the index can be written as 
\begin{align*}
& {\rm dim}\;{\rm ker}\; \mathfrak{T}_\chi-{\rm dim}\;{\rm ker}\; \mathfrak{T}_\chi^\ast\ \ \mbox{modulo} \ 2 \\
&={\rm dim}\;{\rm ker}\; \mathfrak{T}_\chi+{\rm dim}\;{\rm ker}\; \mathfrak{T}_\chi^\ast\ \ \mbox{modulo} \ 2\\
&={\rm dim}\;{\rm ker}\;(\mathcal{P}_{\rm D}-U\mathcal{P}_{\rm D}U-1)\ \ \mbox{modulo} \ 2,
\end{align*}
where we have used the expression (\ref{PD-UPDUmat}) of the operator $\mathcal{P}_{\rm D}-U\mathcal{P}_{\rm D}U$. 

On the other hand, the $\ze_2$ index of AII class in three dimensions is given by 
$$
{\rm Ind}_2^{(3)}(D_a,P_{\rm F})={\rm dim}\;{\rm ker}\;(P_{\rm F}-D_aP_{\rm F}D_a-1)\ \ \mbox{modulo}\ 2.
$$ 
Combining these observations with the above result (\ref{dimkerrelation}), we obtain the desired result, 
$$
{\rm Ind}_2^{(3)}(D_a,P_{\rm F})={\rm Ind}^{(3)}(D_a,S,U)\ \ \mbox{modulo} \ 2,
$$
in the weak limit of the perturbation.

\bigskip\bigskip

\noindent
\thanks{\textbf{Acknowledgement:} 
HK was supported in part by JSPS Grants-in-Aid for Scientific Research No. JP15K17719 and No. JP16H00985.
\bigskip\bigskip




\begin{thebibliography}{99}
\bibitem{AG} Aizenman, M., Graf, G. M.: Localization Bounds for an 
Electron Gas. J. Phys. A{\bf 31}, 6783--6806 (1998). 

\bibitem{AM} Aizenman, M., Molchanov, S.: Localization at a large disorder and at extreme energies: 
An elementary derivation.
Commun. Math. Phys. {\bf 157}, 245--278 (1993). 

\bibitem{AKK} Akagi, Y., Katsura, H., Koma, T.: A new numerical method for $\ze_2$ topological insulators with strong disorder. 
J. Phys. Soc. Jpn. {\bf 86}, 123710 (2017).  

\bibitem{AZ} Altland, A., Zirnbauer, M. R.: Nonstandard symmetry classes in mesoscopic 
normal-superconducting hybrid structures. 
Phys. Rev. B {\bf 55}, 1142--1161 (1997). 

\bibitem{ASS}  Avron, J. E., Seiler, R., Simon, B.:
Charge Deficiency, Charge Transport and Comparison of Dimensions. 
Commun. Math. Phys. {\bf 159}, 399--422 (1994).

\bibitem{ASS2} Avron, J., Seiler, R., Simon, B.:
The index of a pair of projections. 
J. Func. Anal. {\bf 120}, 220--237 (1994).

\bibitem{BVS} Bellissard, J., Van Elst, A., Schulz-Baldes, H.: 
The Noncommutative Geometry of the Quantum Hall Effect. 
J. Math. Phys. {\bf 35}, 5373--5451 (1994). 

\bibitem{Bourgain} Bourgain, J.: An approach to Wegner's estimate using subharmonicity, 
J. Stat. Phys. {\bf 134}, 969--978 (2009). 

\bibitem{BCR} Bourne, C., Carey, A. L., Rennie, A.: A noncommutative framework for topological insulators. 
Rev. Math. Phys. {\bf 28}, 1650004 (2016). 

\bibitem{BKR} Bourne, C., Kellendonk, J., Rennie, A.: The K-theoretic bulk-edge correspondence for topological 
insulators. Ann. Henri Poincar\'e, {\bf 18}, 1833--1866, (2017).   

\bibitem{ChapStolz} Chapman, J., and Stolz, G.: Localization for Random Block Operators Related to the XY Spin Chain. 
Ann. Henri Poincar\'e {\bf 16}, 405--435 (2015).

\bibitem{CH} Combes, J. -M., Hislop, P. D.: Localization for some continuous, random Hamiltonians in $d$-dimensions. 
J. Func. Anal. {\bf 124}, 149--180 (1994). 

\bibitem{Connes} Connes, A.: {\it Noncommutative Geometry}, Academic Press, 
San Diego, 1994.

\bibitem{CM} Connes, A., Moscovici, H.: The local index formula in noncommutative geometry. 
Geom. Funct. Anal. {\bf 5}, 174--243 (1995). 

\bibitem{DNSB} De Nittis, G., Schulz-Baldes, H.: 
Spectral Flows Associated to Flux Tubes. 
Ann. H. Poincare {\bf 17}, 1--35 (2016). 

\bibitem{EGS} 
Elgart, A., Graf, G. M., Schenker, J. H.: 
Equality of the Bulk and Edge Hall Conductances in a Mobility Gap. 
Commun. Math. Phys. {\bf 259} 185--221 (2005). 

\bibitem{ElSch} Elgart, A., Schlein, B.: 
Adiabatic Charge Transport and the Kubo Formula for Landau Type  
Hamiltonians. Comm. Pure Appl. Math. {\bf 57}, 590--615 (2004).  

\bibitem{EssinMoore} Essin, A. M., Moore, J. E.: 
Topological Insulators beyond the Brillouin Zone via Chern Parity. 
Phys. Rev. B {\bf 76}, 165307 (2007).  

\bibitem{FuKane} Fu, L., Kane, C. L.: 
Time Reversal Polarization and a $\ze_2$ Adiabatic Spin Pump. 
Phys. Rev. B {\bf 74}, 195312 (2006).  

\bibitem{FuKane2} Fu, L., Kane, C. L.: 
Topological Insulators with Inversion Symmetry. 
Phys. Rev. B {\bf 76}, 045302 (2007).  

\bibitem{FukuiFujiwara} Fukui, T., Fujiwara, T.: 
A $\ze_2$ Index of a Dirac Operator with Time Reversal Symmetry.
J. Phys. A: Math. Theor. {\bf 42},  362003--362009 (2009).

\bibitem{FukuiHatsugai} Fukui, T, Hatsugai, Y.: 
Topological Aspect of the Quantum Spin-Hall Effect in Graphene: 
$\ze_2$ Topological Order and Spin Chern Number. 
Phys. Rev. B {\bf 75}, 121403(R) (2007).  

\bibitem{FHA} Fulga, I. C., Hassler, F., Akhmerov, A. R.: 
Scattering theory of topological insulators and superconductors. 
Phys. Rev. B {\bf 85}, 165409 (2012). 

\bibitem{FHAB} Fulga, I. C., Hassler, F., Akhmerov, A. R., Beenakker, C. W. J.: 
Scattering formula for the topological quantum number of a disordered multimode wire. 
Phys. Rev. B {\bf 83}, 155429 (2011). 

\bibitem{GM} Gebert, M., and M\"uller, P.: Localization for random block operators, 
Oper. Theory Adv. Appl. {\bf 232}, 229--246 (2013). 

\bibitem{GSB} Gro{\ss}mann, J., Schulz-Baldes, H.: 
Index pairing in presence of symmetries with applications to topological insulators. 
Commun. Math. Phys. {\bf 343}, 477--513 (2016). 

\bibitem{Guo} Guo, H.-M.: Topological invariant in three-dimensional band insulators with disorder. 
Phys. Rev. B {\bf 82}, 115122 (2010).  

\bibitem{HasanKane} Hasan, M, Kane, C. L.: Colloquium: Topological insulators. 
Rev. Mod. Phys. {\bf 82}, 3045--3067 (2010)

\bibitem{HL1} Hastings, M. B., Loring, T. A.:
Almost Commuting Matrices, Localized Wannier Functions, and the Quantum Hall Effect. 
J. Math. Phys. {\bf 51}, 015214 (2010). 

\bibitem{HL2} Hastings, M. B., Loring, T. A.:
Topological Insulators and $C^\ast$-Algebras: Theory and Numerical Practice. 
Ann. Phys. {\bf 326}, 1699--1759 (2011). 

\bibitem{HastingsMichalakis} Hastings, M. B., Michalakis, S.: 
Quantization of Hall Conductance For Interacting Electrons on a Torus.
Commun. Math. Phys. {\bf 334}, 433--471 (2015).

\bibitem{Higson1} Higson, N.: The local index formula in noncommutative geometry. 
Contemporary developments in algebraic K-theory, 443-536, ICTP Lect. Notes, XV, 
Abdus Salam Int. Cent. Theoret. Phys., Trieste, 2004. 

\bibitem{Higson2} Higson, N.: The residue index theorem of Connes and Moscovici. 
Surveys in noncommutative geometry, 71--126, Clay Math. Proc., 6, 
Amer. Math. Soc., Providence, RI, 2006.

\bibitem{IM} Ishikawa, K., Matsuyama, T.: Magnetic field induced multi-component QED$_3$ and 
quantum Hall effect. 
Z. Phys. C {\bf 33}, 41--45 (1986). 

\bibitem{KaneMele} Kane, C. L., Mele, E. J.: 
$\ze_2$ Topological Order and Quantum Spin Hall Effect: 
Phys. Rev. Lett. {\bf 95}, 146802 (2005).  

\bibitem{KatsuraKoma} Katsura, H., Koma, T.: 
The $\ze_2$ Index of Disordered Topological Insulators with Time Reversal Symmetry. 
J. Math. Phys. {\bf 57}, 021903 (2016). 

\bibitem{KawajiWakabayashi} Kawaji, S., Wakabayashi, J.: 
Temperature Dependence of Transverse and Hall Conductivities of Silicon 
MOS Inversion Layers under Strong Magnetic Fields. In: 
{\it Physics in High Magnetic Fields.} Chikazumi, S., Miura, N. (eds), 
pp.~284--287, Springer, Berlin, Heidelberg, New York, 1981. 

\bibitem{Kellendonk1} Kellendonk, J.: On the C$^\ast$-algebraic approach to topological phases for insulators: 
Ann. Henri Poincar\'e {\bf 18}, 2251--2300, (2017).   

\bibitem{Kellendonk2} Kellendonk, J.: Cyclic cohomology for graded C$^{\ast,r}$-algebras and its pairings 
with van Daele K-theory. 
Preprint, arXiv:1607.08465.

\bibitem{KMM} Kirsch, W., Metzger, B., and M\"uller, P.: Random block operators, 
J. Stat. Phys. {\bf 143}, 1035--1054 (2011). 

\bibitem{Kitaev} Kitaev, A. Y.: 
Unpaired Majorana fermions in quantum wires. 
Physics-Uspekhi {\bf 44}, 131--136 (2001). 

\bibitem{KitaevPT} Kitaev, A. Y.: Periodic table for topological insulators and superconductors, 
(Advances in Theoretical Physics: Landau Memorial Conference) AIP Conference Proceedings,  
{\bf 1134}, 22--30 (2009).

\bibitem{KDP} von Klitzing, K., Dorda, G., Pepper, M.: 
New Method for High Accuracy Determination of the Fine Structure Constant 
Based on Quantized Hall Resistance. 
Phys. Rev. Lett. {\bf 45}, 494--497 (1980) 

\bibitem{Kohmoto} Kohmoto, M.: Topological Invariant and the Quantization 
of the Hall Conductance. Ann. Phys. {\bf 160}, 343--354 (1985) 

\bibitem{Koma1} Koma, T.:
Revisiting the charge transport in quantum Hall systems.
Rev. Math. Phys. {\bf 16}, 1115--1189 (2004).

\bibitem{Koma2} Koma, T.:
Widths of the Hall Conductance Plateaus. 
J. Stat. Phys. {\bf 130}, 843-934 (2008). 

\bibitem{Koma5} Koma, T.: 
Topological Current in Fractional Chern Insulators, 
Preprint, arXiv:1504.01243. 

\bibitem{Kubota} Kubota, Y.: Controlled topological phases and bulk-edge correspondence. 
Commun. Math. Phys. {\bf 349}, 493--525 (2017).  

\bibitem{LeeRyu} Lee, S.-S., Ryu, S.: 
Many-Body Generalization of the $\ze_2$ Topological 
Invariant for the Quantum Spin Hall Effect. 
Phys. Rev. Lett. {\bf 100}, 186807 (2008). 

\bibitem{LeungProdan} Leung, B,. Prodan, E.: Effect of strong disorder in a three-dimensional 
topological insulator: Phase diagram and maps of the $\ze_2$ invariant. 
Phys. Rev. B {\bf 85}, 205136 (2012).  

\bibitem{Loring} Loring, T. A.: K-theory and pseudospectra for topological insulators. 
Ann. Phys. {\bf 356}, 383--416 (2015). 

\bibitem{LH} Loring, T. A., Hastings, M. B.: 
Disordered Topological Insulators via $C^\ast$-Algebras.
EPL (Europhys. Lett.) {\bf 92}, 67004 (2010). 

\bibitem{MP} Marciano, W., Pagels, H.: Quantum Chromodynamics. 
Phys. Reports {\bf 36}, 137--276 (1978). 

\bibitem{MHSP} Mondragon-Shem, I., Hughes, T. L., Song, J., Prodan, E.: 
Topological criticality in the chiral-symmetric AIII class at strong disorder. 
Phys. Rev. Lett. {\bf 113}, 046802 (2014). 

\bibitem{MooreBalents} Moore, J. E., Balents, L.: 
Topological Invariants of Time-Reversal-Invariant Band Structures. 
Phys. Rev. B {\bf 75}, 121306(R) (2007). 

\bibitem{NTW} Niu, Q., Thouless, D. J., Wu, Y. S.: Quantized Hall conductance 
as a topological invariant. Phys. Rev. B {\bf 31}, 3372--3377 (1985).

\bibitem{Penrose} Penrose, R.: The role of aesthetics in pure and applied mathematical research, 
Bull. Inst. Math. Appl. {\bf 10}, 266--271 (1974).

\bibitem{PLB} Prodan, E., Leung, B., Bellissard, J.: The Non-commutative $n_{\rm th}$-Chern 
Number $(n\ge 1)$. 
J. Phys. A: Math. Theor. {\bf 46}, 485202 (2013). 

\bibitem{PSB} Prodan, E., Schulz-Baldes, H.: Non-commutative odd Chern numbers and 
topological phases of disordered chiral systems. 
J. Func. Anal. {\bf 271}, 1150--1176 (2016). 

\bibitem{PSBbook} Prodan, E., Schulz-Baldes, H.: {\it Bulk and boundary invariants for complex topological 
insulators: From K-theory to physics}, Springer, Berlin, 2016.  

\bibitem{QWZ} Qi, X.-L., Wu, Y.-S., Zhang, S.-C.: Topological Quantization of the Spin 
Hall Effect in Two-Dimensional Paramagnetic Semiconductors. 
Phys. Rev. B {\bf 74}, 085308 (2006). 

\bibitem{QiZhang} Qi, X.-L. and Zhang, S.-C.: Topological insulators and superconductors. 
Rev. Mod. Phys. {\bf 83}, 1057--1110 (2011).

\bibitem{RSI} Reed, M., Simon, B.: {\it Methods of Modern Mathematical 
Physics, vol.~I, Functional Analysis}, Academic Press, New York, 1972. 

\bibitem{RSII} Reed, M., Simon, B.: {\it Methods of Modern Mathematical 
Physics, vol.~II, Fourier Analysis, Self-Adjointness}, Academic Press, New York, 1975.

\bibitem{RSIV} Reed, M., Simon, B.: {\it Methods of Modern Mathematical 
Physics, vol.~IV, Analysis of Operators}, Academic Press, New York, 1978. 

\bibitem{RSB} T. Richter and H. Schulz-Baldes, Homotopy Arguments for 
Quantized Hall Conductivity, {\it J. Math. Phys.} {\bf 42},  3439--3444 (2001). 

\bibitem{RSFL} Ryu, S., Schnyder, A. P., Furusaki, A., Ludwig, A. W. W.: 
Topological insulators and superconductors: tenfold way and dimensional hierarchy. 
New J. Phys. {\bf 12}, 065010 (2010). 

\bibitem{Roy} Roy, R.: $\ze_2$ Classification of Quantum Spin Hall Systems: 
An Approach using Time-Reversal Invariance. 
Phys. Rev. B {\bf 79}, 195321 (2009). 

\bibitem{SbBr} Sbierski, B., Brouwer, P. W.: $\ze_2$ phase diagram of three-dimensional disordered topological 
insulators via a scattering matrix approach.  
Phys. Rev. B {\bf 89}, 155311 (2014). 

\bibitem{SRFL} Schnyder, A. P., Ryu, S., Furusaki, A., Ludwig, A. W. W.: 
Classification of topological insulators and superconductors in three spatial dimensions. 
Phys. Rev. B {\bf 78}, 195125 (2008). 

\bibitem{SB} Schulz-Baldes, H.: $\ze_2$-Indices of Odd Symmetric Fredholm Operators: 
Documenta Mathematica {\bf 20}, 1481--1500 (2015). 

\bibitem{ShiozakiSato} Shiozaki, K., Sato, M.: Topology of crystalline insulators and superconductors. 
Phys. Rev. B {\bf 90}, 165114 (2014). 

\bibitem{Sierpinski} Sierpinski, W.: Sur une courbe cantorienne qui contient une image biunivoque 
et continue de toute courbe donnee. 
C. r. hebd. Seanc. Acad. Sci., Paris {\bf 162}, 629--632 (1916).

\bibitem{Simon} Simon, B.: {\it Trace Ideals and their Applications}, 
Cambridge University Press, Cambridge, 1979. 

\bibitem{So} So, H.: Induced topological invariants by lattice fermions in odd dimensions. 
Prog. Theor. Phys. {\bf 74}, 585--593 (1985). 

\bibitem{SFP} Song, J., Fine, C., Prodan, E.: Effect of strong disorder on three-dimensional chiral topological insulators: 
Phase diagrams, maps of the bulk invariant, and existence of topological extended bulk states. 
Phys. Rev. B {\bf 90}, 184201 (2014).   

\bibitem{Stein} Stein, P.: A note on the volume of a simplex. 
The American Mathematical Monthly {\bf 73}, 299--301 (1966). 

\bibitem{TeoKane} Teo, J. C. Y., Kane, C. L.: Topological defects and gapless modes in insulators and 
superconductors. 
Phys. Rev. B {\bf 82}, 115120 (2010). 

\bibitem{Thiang} Thiang, G. C.: On the K-theoretic classification of topological phases of matter. 
Ann. H. Poincar\'e {\bf 17}, 757--794 (2016). 

\bibitem{TKNN} Thouless, D. J., Kohmoto, M., Nightingale, M. P., den 
Nijs, M.: Quantized Hall Conductance in a Two-Dimensional Periodic Potential. 
Phys. Rev. Lett. {\bf 49}, 405--408 (1982). 

\end{thebibliography}
\end{document}